\documentclass[12pt]{article}

\usepackage{amssymb,amsmath,amsthm,amsfonts,eurosym,geometry,ulem,graphicx,caption,color,setspace,sectsty,comment,footmisc,natbib,pdflscape,array,bbm,hyperref,mathrsfs,mathtools,enumerate,enumitem,indentfirst,todonotes,yhmath,bm,algorithm2e,cleveref,scalerel,stackengine,booktabs,threeparttable,multirow,float,xcolor,soul,tcolorbox,makecell,subcaption}

\stackMath
\newcommand\What[1]{%
\savestack{\tmpbox}{\stretchto{%
  \scaleto{%
    \scalerel*[\widthof{\ensuremath{#1}}]{\kern-.6pt\bigwedge\kern-.6pt}%
    {\rule[-\textheight/2]{1ex}{\textheight}}
  }{\textheight}%
}{0.5ex}}%
\stackon[1pt]{#1}{\tmpbox}%
}
\usepackage[toc,page]{appendix}
\RestyleAlgo{ruled}
\SetKwComment{Comment}{/* }{ */}

\allowdisplaybreaks[1]

\newcounter{question}
\def\sym#1{\ifmmode^{#1}\else\(^{#1}\)\fi}

\newcommand{\eps}{\varepsilon}
\newcommand{\mA}{\mathcal{A}}
\newcommand{\mB}{\mathcal{B}}
\newcommand{\mC}{\mathcal{C}}

\newcommand{\mF}{\mathcal{F}}
\newcommand{\mG}{\mathcal{G}}

\newcommand{\mI}{\mathcal{I}}
\newcommand{\mJ}{\mathcal{J}}
\newcommand{\mK}{\mathcal{K}}
\newcommand{\mL}{\mathcal{L}}

\newcommand{\mN}{\mathcal{N}}

\newcommand{\mP}{\mathcal{P}}

\newcommand{\mU}{\mathcal{U}}
\newcommand{\mV}{\mathcal{V}}
\newcommand{\mW}{\mathcal{W}}
\newcommand{\mX}{\mathcal{X}}

\newcommand{\bE}{\mathbb{E}}

\newcommand{\bR}{\mathbb{R}}

\newtheorem{theorem}{Theorem}
\newtheorem{assumption}{Assumption}
\newtheorem{proposition}{Proposition}

\newtheorem{lemma}{Lemma}
\newtheorem{definition}{Definition}

\theoremstyle{definition}
\newtheorem{example}{Example}

\newlist{assumptionitems}{enumerate}{1} 
\setlist[assumptionitems]{label=(\roman*), ref=\theassumption(\roman*), align=left, leftmargin=*}
\Crefname{assumption}{Assumption}{Assumptions}
\crefname{assumptionitemsi}{assumption}{assumptions}

\newlist{lemmaitems}{enumerate}{1} 
\setlist[lemmaitems]{label=(\roman*), ref=\thelemma(\roman*), align=left, leftmargin=*}
\Crefname{lemma}{Lemma}{Lemmas}
\crefname{lemmaitemsi}{lemma}{lemmas}

\newlist{propositionitems}{enumerate}{1} 
\setlist[propositionitems]{label=(\roman*), ref=\theproposition(\roman*), align=left, leftmargin=*}
\Crefname{proposition}{Proposition}{Propositions}
\crefname{propositionitemsi}{proposition}{propositions}

\newlist{theoremitems}{enumerate}{1} 
\setlist[theoremitems]{label=(\roman*), ref=\thetheorem(\roman*), align=left, leftmargin=*}
\Crefname{theorem}{Theorem}{Theorems}
\crefname{theoremitemsi}{theorem}{theorems}

\newlist{definitionitems}{enumerate}{1} 
\setlist[definitionitems]{label=(\roman*), ref=\thedefinition(\roman*), align=left, leftmargin=*}
\Crefname{definition}{Definition}{Definitions}
\crefname{definitionitemsi}{definition}{definitions}

\newtheorem{myalgo}{Algorithm}

\theoremstyle{definition}
\newtheorem{remark}{Remark}

\newcommand{\blue}[1]{{\color{blue} #1}}

\newcolumntype{L}[1]{>{\raggedright\let\newline\\arraybackslash\hspace{0pt}}m{#1}}
\newcolumntype{C}[1]{>{\centering\let\newline\\arraybackslash\hspace{0pt}}m{#1}}
\newcolumntype{R}[1]{>{\raggedleft\let\newline\\arraybackslash\hspace{0pt}}m{#1}}

\newcommand{\continuation}{??}

\hypersetup{colorlinks=true,linkcolor=black,citecolor=blue,urlcolor=blue}

\newcommand{\blueref}[1]{\hyperref[#1]{\textcolor{blue}{\ref*{#1}}}}
\begin{document}

\begin{titlepage}
\title{Robust Structural Estimation under Misspecified Latent-State Dynamics\thanks{I would like to thank Lars Nesheim, Dennis Kristensen, and Aureo de Paula for their continuous guidance and support. I am grateful to Victor Aguirregabiria, Dante Amengual, Debopam Bhattacharya, Stephane Bonhomme, Christian Bontemps, Cuicui Chen, Tim Christensen, Kevin Dano, Ben Deaner, Ivan Fernandez-Val, Hugo Freeman, Alfred Galichon, Raffaella Giacomini, Gautam Gowrisankaran, Joao Granja, Jiaying Gu, Sukjin Han, Stefan Hubner, Koen Jochmans, Hiroaki Kaido, Hiro Kasahara, Yuichi Kitamura, Adam Lee, Arthur Lewbel, Lorenzo Magnolfi, Chuck Manski, Matt Masten, Robert McCann, Adam McCloskey, Nour Meddahi, Bob Miller, Marcel Nutz, Ariel Pakes, Silvana Pesenti, Lei Qiao, Mathias Reynaert, Eduardo Souza-Rodrigues, Pietro Emilio Spini, Jörg Stoye, Elie Tamer, Xun Tang, Weining Wang, Martin Weidner, Leonard Wong, Andrei Zeleneev for helpful discussions and comments. I also thank seminar participants at UCL, TSE, and University of Toronto, as well as conference audiences at 2025 Bristol Econometric Study Group, 2025 Midwest Econometrics Group Conference, and IAAE 2025.}}
\author{Ertian Chen\thanks{Department of Economics, University College London and CeMMAP. Email: \href{mailto:ertian.chen.19@ucl.ac.uk}{ertian.chen.19@ucl.ac.uk}.}}
\date{November 13, 2025}
\maketitle

\vspace{-1cm}
\begin{center}
\href{https://ertianchen.github.io/pdf/Distributional_Robustness.pdf}{Click here for the latest version}
\end{center}

\begin{abstract}   
\noindent Estimation and counterfactual analysis in dynamic structural models rely on assumptions about the dynamic process of latent variables, which may be misspecified. We propose a framework to quantify the sensitivity of scalar parameters of interest (e.g., welfare, elasticity) to such assumptions. We derive bounds on the scalar parameter by perturbing a reference dynamic process, while imposing a stationarity condition for time-homogeneous models or a Markovian condition for time-inhomogeneous models. The bounds are the solutions to optimization problems, for which we derive a computationally tractable dual formulation. We establish consistency, convergence rate, and asymptotic distribution for the estimator of the bounds. We demonstrate the approach with two applications: an infinite-horizon dynamic demand model for new cars in the United Kingdom, Germany, and France, and a finite-horizon dynamic labor supply model for taxi drivers in New York City. In the car application, perturbed price elasticities deviate by at most 15.24\% from the reference elasticities, while perturbed estimates of consumer surplus from an additional \$3,000 electric vehicle subsidy vary by up to 102.75\%. In the labor supply application, the perturbed Frisch labor supply elasticity deviates by at most 76.83\% for weekday drivers and 42.84\% for weekend drivers.

\bigskip



\bigskip
\end{abstract}
\setcounter{page}{0}
\thispagestyle{empty}
\end{titlepage}
\pagebreak \newpage

\pagenumbering{arabic}
\setcounter{page}{1}

\onehalfspacing
\section{Introduction}

\noindent Dynamic structural models are useful tools for counterfactual policy analysis in various fields of economics. The dynamic process of latent variables is a key feature of these models, as it captures the persistence of unobserved factors that affect agents' decisions over time. Examples of potentially serially correlated latent variables include product characteristics in demand estimation (\cite{nair2007intertemporal,schiraldi2011automobile,gowrisankaran2012dynamics}), search costs in consumer search (\cite{koulayev2014search}), firm productivity in trade (\cite{piveteau2021empirical}), patent profitability in optimal stopping (\cite{pakes1984patents}), quality in technology adoption (\cite{de2019subsidies}), health shocks in insurance (\cite{fang2021life}), and beliefs about ability in labor economics (\cite{miller1984job,arcidiacono2025college}).

Assumptions about the dynamic process governing the serial dependence of latent variables are central to the estimation and counterfactual analysis in dynamic structural models. These assumptions capture agents' uncertainty about the future, such as a consumer's uncertainty about future product characteristics in demand estimation. The misspecification of these assumptions can lead to biased estimates of future continuation value, which in turn bias counterfactual predictions (e.g., welfare and elasticity). However, the direction and magnitude of this bias are unclear, because the models are dynamic and nonlinear. This raises the need for sensitivity analysis of empirical results to these distributional assumptions.

In this paper, we propose a framework to quantify this sensitivity by perturbing a reference dynamic process to compute bounds on a scalar parameter of interest. The scalar parameter (e.g., welfare and elasticity) is a function of model primitives, such as model parameters, the distribution of latent variables, and the value function in dynamic structural models. In our proposed framework, the bounds on this scalar parameter are the solutions to constrained optimization problems whose feasible region (identified set) is defined by moment conditions for estimation, structural constraints (e.g., the Bellman equation), and the perturbation set around the reference transition distribution. The distribution that achieves the bound is called the worst-case distribution.

A central challenge is to define the perturbation set in a way that simplifies computation while maintaining key structural features of the distribution of latent variables. For time-homogeneous models, the structural feature is the stationarity condition, i.e., the perturbed dynamic process must be stationary. For finite-horizon, time-inhomogeneous models, the perturbed trajectory of latent variables must be Markovian. In addition, the terminal distribution is fixed because it can be nonparametrically point identified (\cite{lewbel2000semiparametric,matzkin2007nonparametric}). Because the stationarity and Markov conditions are functional constraints imposed directly on the distribution being optimized, they are computationally difficult to impose. We contribute to the sensitivity analysis literature (e.g., \cite{christensen2023counterfactual}) by providing a computationally tractable framework to deal with these constraints.

There are further practical challenges. First, the properties (e.g., closed-form, smoothness) of the worst-case distribution are typically unclear, complicating the choice of approximation methods. Second, the expectations that define the scalar parameter, the model-implied moments, and the structural constraints are all calculated with respect to the perturbed distribution. Because this distribution is itself an optimization variable that changes during the optimization, the numerical integration can be difficult to implement. Third, the value function in dynamic structural models is infinite-dimensional due to the serial dependence of latent variables, which complicates the optimization problem further.

To address these challenges, we define the perturbation set as a Kullback–Leibler (KL) divergence ball around the reference distribution. For the time-homogeneous case, the reference distribution is a joint distribution of current and future latent variables. For the time-inhomogeneous case, it is the joint distribution of the entire trajectory of latent variables. The KL radius controls the size of the perturbation set. We then employ the Optimal Transport (OT) framework to impose structural constraints on the distribution of latent variables. In the time-homogeneous case, the stationarity condition requires that the marginal distributions of current and future latent variables of the perturbed joint distribution coincide, which can be imposed using OT. This marginal constraint, together with the KL divergence penalty from the problem’s Lagrangian, yields a computationally tractable Entropic Optimal Transport (EOT) problem. In the time-inhomogeneous case, the perturbed trajectory must be Markovian, and the terminal distribution is fixed. We show how to formulate this as an EOT problem.

The EOT problem can be solved using its dual formulation, which has three key advantages over the primal formulation. First, it provides a closed-form expression for the worst-case distribution and characterizes its smoothness. During our proposed optimization algorithm, this closed-form expression is used to update the value function in dynamic structural models. Second, the Sinkhorn algorithm (\cite{sinkhorn1967concerning,cuturi2013sinkhorn}) allows us to solve the EOT problem and compute the worst-case distribution efficiently. Finally, because the expectations in the dual problem are taken with respect to the reference distribution, an appropriate numerical integration method can be chosen in advance.

Then we consider three complementary sensitivity measures to interpret the results. First, the global sensitivity approach computes the largest deviation from the reference value. It progressively increases the KL radius until the bounds flatten. We show that it provides a tractable approximation to the nonparametric bounds on the scalar parameter when the KL divergence constraint is removed. Moreover, we derive an explicit upper bound on the approximation error, which enables us to control the error within a desired level. Second, the local sensitivity approach analyzes the effect of small perturbations. It computes the right derivatives of the bounds with respect to the KL radius at zero, which serves as our local sensitivity measure. Finally, the robustness metric approach, inspired by \cite{spini2021robustness}, computes the smallest deviation from the reference distribution required to produce sensitive results. It is the smallest KL divergence from the reference distribution required for the scalar parameter's value to fall below a user-specified threshold (e.g., 5\% below the reference value). In addition to our three sensitivity measures, we can also estimate an alternative model and set the radius as the KL divergence between the alternative and reference models.

For large sample properties, we propose an estimator for the bound, establishing its consistency and convergence rate. To this end, we first establish the consistency and convergence rate of the estimator of the identified set, following \cite{chernozhukov2007estimation}. We then derive the asymptotic distribution of the plug-in estimator by proving the Hadamard directional differentiability of the bound.

We apply our framework to an infinite-horizon dynamic demand model for new cars in the UK, Germany, and France. We consider the sensitivity of the price elasticity and consumer surplus from an additional \$3,000 electric vehicle subsidy. In the model, the indirect utility of purchasing is a latent variable due to unobserved product characteristics, and its transition is typically modeled as an AR(1) process (e.g., \cite{schiraldi2011automobile,gowrisankaran2012dynamics}). For the price elasticity, we find that the French market is the least sensitive to the distributional assumption (at most 6.20\% deviation from the reference elasticity), while the German market is the most sensitive (at most 15.24\% deviation). We also find that this sensitivity is relatively stable over time for all three markets. For the consumer surplus, the German market is also the most sensitive (at most 102.75\% deviation from the reference consumer surplus), while the UK and French markets are less sensitive (at most 25.17\% and 24.73\% deviation, respectively). Importantly, the results remain economically meaningful even under the worst-case distribution, with the consumer surplus remaining at least \$309 million for the German market, \$1,243 million for the French market, and \$2,584 million for the UK market.

We also apply our framework to a finite-horizon dynamic labor supply model for taxi drivers in New York City. We consider the sensitivity of the elasticity of stopping work and the Frisch elasticity of labor supply. In the model, the market-level supply shock is a latent variable, and its transition is also modeled as an AR(1) process. For the elasticity of stopping work, we find that weekday drivers' elasticity is more sensitive in the morning, while weekend drivers' elasticity is more sensitive in the afternoon. For the Frisch elasticity, both weekday and weekend drivers' elasticities are sensitive to the distributional assumption, with at most 76.83\% and 42.84\% deviation from the reference elasticity, respectively.

\subsection{Related Literature}

\noindent \textit{Identification and estimation.} Many papers have focused on identification in a range of dynamic structural models including finite mixture models (\cite{kasahara2009nonparametric,luo2022identification,higgins2023identification,higgins2025learning}), unobservable Markov processes (\cite{hu2012nonparametric}), and counterfactual conditional choice probabilities in dynamic binary choice models (\cite{norets2014semiparametric}). \cite{berry2023instrumental} uses the generalized instrumental variable approach for unobserved state variables in dynamic discrete choice (DDC) models. In fixed effects DDC models, \cite{aguirregabiria2021sufficient} considers identification of structural parameters using sufficient statistics, while \cite{aguirregabiria2024identification} studies identification of average marginal effects. \cite{hwang2024identification} employs proxy variables for the latent variables. \cite{kalouptsidi2021linear} proposes the Euler Equations in Conditional Choice Probabilities (ECCP) estimator. By leveraging finite-dependence properties and cross-sectional data, it identifies structural parameters in the presence of serially correlated market-level unobserved variables without distributional assumptions. \cite{arcidiacono2011conditional} adapts the Expectation-Maximization algorithm to estimate DDC models with discrete latent types. \cite{norets2009inference} extends the Bayesian estimation of DDC models by \cite{imai2009bayesian} to allow for serially correlated latent variables, while \cite{blevins2016sequential} proposes a sequential Monte Carlo method. \cite{chiong2016duality} shows identification in DDC models is an optimal transport problem under serial independence of utility shocks. In addition, \cite{chen2011sensitivity}, \cite{schennach2014entropic}, and \cite{fan2023partial,fan2025partial} consider inference of finite-dimensional parameters in the presence of an infinite-dimensional parameter, namely the distribution of latent variables.

Our framework contributes to this literature in three ways. First, we focus directly on a scalar parameter (e.g., welfare, elasticity) rather than on the identified set of model primitives. Second, to analyze sensitivity, we consider a perturbation set around the reference distribution rather than the set of all distributions. Third, we complement identification strategies that do not rely on distributional assumptions. In the labor supply application, we use the ECCP estimator of \cite{kalouptsidi2021linear} to point identify the utility parameters and then conduct sensitivity analysis of the labor supply elasticity with respect to assumptions about the dynamic process of the market-level supply shock.

\textit{Sensitivity analysis and robustness.} In DDC models, \cite{kalouptsidi2021counterfactual,kalouptsidi2021identification} relax common normalizations on the utility function. \cite{bugni2019inference} considers the local misspecification of the transition density of observable variables, assuming the transition density is correctly specified in the limit (as the sample size goes to infinity) in DDC models. \cite{andrews2017measuring} considers a setting in which the moments are locally misspecified under the reference distribution. Subsequent work \cite{armstrong2021sensitivity} constructs near-optimal confidence intervals in such models. \cite{kitamura2013robustness} considers the robust estimation under moment restrictions. \cite{bonhomme2022minimizing} perturbs the reference model and assumes the size of the perturbation shrinks to zero as the sample size goes to infinity. \cite{chen2024robust} relaxes the rational expectation assumption. \cite{gu2023dual} considers the identification of scalar counterfactual parameters using optimal transport. \cite{spini2021robustness} studies the robustness of policy effects to changes in the distribution of covariates. \cite{armstrong2025misspecification} provides a selective review of misspecification in econometrics. Most closely related is \cite{christensen2023counterfactual}, who conducts sensitivity analysis with respect to parametric assumptions about the distribution of latent variables in structural models. However, they focus on relaxing the marginal distribution assumption while maintaining serial independence. 

We contribute to this literature in three ways. First, our focus on misspecified dynamic processes complements previous analysis of misspecification in dynamic structural estimation (e.g., \cite{bugni2019inference,kalouptsidi2021counterfactual,kalouptsidi2021identification,christensen2023counterfactual}). Second, the size of our perturbation set is fixed and does not shrink with the sample size. Third, we provide a computationally tractable dual formulation for the optimization problem. While prior work (\cite{schennach2014entropic,gu2023dual,christensen2023counterfactual}) uses duality to convert the infinite-dimensional problem to a finite-dimensional one, our dual problem is still infinite-dimensional. This is because the stationarity condition for time-homogeneous models and fixed terminal distribution for time-inhomogeneous models are infinite-dimensional, and serial dependence leads to an infinite-dimensional value function in dynamic structural models. However, by leveraging EOT duality, we provide a tractable implementation and demonstrate our approach through two empirical applications.

\textit{Distributionally robust optimization (DRO).} The literature on DRO (\cite{kuhn2019wasserstein,rahimian2019distributionally,blanchet2022confidence,gao2023distributionally,wang2021sinkhorn}) usually studies the uncertainty due to limited observability of data, noisy measurements, or estimation errors. 

Our work is distinct in that we study the misspecification due to assumptions about the serial dependence of latent variables—a problem of model specification rather than data limitation. Moreover, because our framework can be treated as a moment-constrained DRO problem, we employ the minimax theorem (see \cite{fan1953minimax} and \cite{ricceri1998minimax}, Theorem 1.3) to exchange the order of the supremum over Lagrangian multipliers (for the moment conditions and structural constraints) and the infimum over distributions in the perturbation set. This exchange allows a direct application of duality results from the DRO literature, which simplifies our proof of the dual formulation significantly.

\textit{Applied work.} Building on the seminal work on the estimation of DDC models (\cite{rust1987optimal,hotz1993conditional,aguirregabiria2002swapping,aguirregabiria2007sequential,pesendorfer2008asymptotic,arcidiacono2011conditional}), most applied work assumes the serial independence of utility shocks. In the presence of serially correlated latent variables, parametric models are often used, as seen in the works of \cite{schiraldi2011automobile,gowrisankaran2012dynamics,blevins2018firm,piveteau2021empirical} and others mentioned in the introduction. 

We contribute to this literature by developing a computationally tractable framework for sensitivity analysis of scalar parameters of interest to these distributional assumptions. We also demonstrate our framework through two empirical applications: an infinite-horizon dynamic demand model, and a finite-horizon dynamic labor supply model.

\textbf{Outline:} Sections \ref{sec: Methodology} and \ref{sec:extension nonstationary} present our framework for time-homogeneous and time-inhomogeneous models. \Cref{sec: Large Sample Properties and Inference} establishes the large sample properties of our estimator. \Cref{sec: Derivative with respect to delta} introduces three sensitivity measures. \Cref{sec: Practical Implementation} discusses practical implementation. Sections \ref{sec: Application} and \ref{sec: Application finite horizon} present two empirical applications. \Cref{sec: Conclusion} concludes. Appendix \ref{sec: Additional Examples} presents additional examples. All proofs are in Appendix \ref{sec:proofs}.

\textbf{Notation:} Let $U \in \mU \subseteq \bR^{d}$ be a vector of latent variables with support $\mU$ where $\mU$ is assumed to be Polish. Let $\mP(\mU)$ be the space of Borel probability measures on $\mU$ and $\mB(\mU)$ be the Borel $\sigma$-algebra on $\mU$. In this paper, all measures are assumed to be absolutely continuous with respect to the Lebesgue measure. Denote by $\bE_{F}\left[\cdot\right]$, $\bE_{x}\left[\cdot\right]$ the expectations with respect to $F \in \mP(\mU)$ and the probability distribution of the random variable $x$, respectively. For variables in a stationary dynamic context, a prime (e.g., $x'$) denotes the variable's value in the next period. Let $F_{1}, F_{2} \in \mP(\mU)$, we write $F_{1} \ll F_{2}$ if $F_{1}$ is absolutely continuous with respect to $F_{2}$, and $F_{1} \otimes F_{2}$ as the product measure. For a finite-dimensional vector, denote by $\|\cdot\|_{p}$ the $p$-norm. Denote by $L^{p}(F)$ the space of functions for which $\int |f|^{p} dF < \infty$. For a set $\mA$, let $\text{int}(\mA)$ denote its interior. Denote by $\bR_{+}$ the set of non-negative real numbers, and $\mathbb{N}$ the set of natural numbers. Finally, let $\mJ := \left\{1,\cdots,J\right\}$.

\section{Methodology for Time-Homogeneous Models} \label{sec: Methodology}

\noindent This section presents our methodology for time-homogeneous models. \Cref{sec: Definition of Perturbation} defines the perturbation set. \Cref{sec: Examples} gives two examples. \Cref{sec: Framework} presents the general framework and derives duality results. \Cref{sec: perturbing stationary distribution} discusses how to perturb the stationary distribution.

\subsection{Definition of Perturbation Set} \label{sec: Definition of Perturbation}

\noindent We partition a vector of latent variables $U \in \mU \subseteq \bR^{d}$ into $2 \leq k \leq d $ subvectors, i.e., $U = (U_{1},\cdots,U_{k})$.\footnote{The vector $U$ can also contain observable variables.} Each subvector $U_{i} \in \mU_{i} \subseteq \bR^{d_{i}}$ has a marginal distribution $\nu_{i} \in \mP(\mU_{i})$ for $i = 1,\cdots,k$. Let $F_{0}$ denote the reference distribution for $U$. The perturbation set around this reference distribution is defined as:
\begin{equation*}
    \mF := \left\{ F \in \mP(\mU) \mid F \in \Pi(\nu_{1},\cdots,\nu_{k}), D_{KL}(F \| F_{0}) \leq \delta \right\}
\end{equation*}
where $\Pi(\nu_{1},\cdots,\nu_{k})$ is the set of joint distributions on $\mU$ with marginals $\{\nu_{i}\}_{i=1}^{k}$, and $\delta \geq 0$ measures the “size” of the perturbation set, defined by the Kullback-Leibler (KL) divergence:
\begin{equation*}
    D_{KL}(F \| F_{0}) :=
    \begin{cases}
    \int \log \left(\frac{dF(U)}{dF_{0}(U)} \right) dF(U) & \text{if } F \ll F_{0} \\
    +\infty & \text{otherwise}
    \end{cases}
\end{equation*}

In time-homogeneous models, the marginal distribution constraints are used to impose the stationarity. Consider an unobserved stationary first-order Markov process $\{\xi_t\}_{t \in \mathbb{Z}}$ with state space $\Xi \subseteq \mathbb{R}^{d_\xi}$. The process is stationary with respect to $\nu_{0} \in \mP(\Xi)$ if $\int F_{0}(d\xi'|\xi) \nu_{0}(d\xi) = \nu_{0}(d\xi')$ for all $\xi' \in \Xi$, where $F_{0}(d\xi'|\xi)$ is the reference transition kernel (e.g., conditional Gaussian distribution for an AR(1) process). This is equivalent to requiring that the marginal distributions of the joint distribution $dF_{0}(\xi,\xi') := F_{0}(d\xi'|\xi)\nu_{0}(d\xi)$ are both $\nu_{0}$, i.e., $F_{0} \in \Pi(\nu_{0},\nu_{0})$. Moreover, for any $F \in \Pi(\nu_{0},\nu_{0})$, its conditional density $F(d\xi'|\xi)$ preserves $\nu_{0}$ as a stationary distribution. For this example, let $U := (\xi,\xi')$. Then, the perturbation set is:
\begin{equation*}
    \mF := \{ F \in \mP(\mU) \mid F \in \underbrace{\Pi(\nu_{0},\nu_{0})}_{\text{Stationarity}}, \underbrace{D_{KL}(F \| F_{0}) \leq \delta}_{\text{Perturbation}} \}
\end{equation*}
This definition allows us to perturb the transition kernel of the Markov process while keeping its stationary distribution unchanged. It can introduce non-linear dynamics into the latent variable process. For example, if the reference model is an AR(1) process, the perturbation set includes any nonlinear first-order Markov process with the same stationary distribution as the AR(1) process. \Cref{sec: perturbing stationary distribution} discusses how to perturb the stationary distribution. In this case, we replace the condition $F \in \Pi(\nu_{0},\nu_{0})$ with $F \in \Pi(\nu,\nu)$, where $\nu$ is the perturbed stationary distribution that is in a neighborhood of $\nu_{0}$.

The marginal constraints with the KL constraint form a computationally tractable EOT problem whose implementation depends on its dual formulation (see Section \ref{sec: Framework}).\footnote{For duality results of general divergence constrained OT problem, see \cite{bayraktar2025stability}.}

\begin{remark}
    \begin{enumerate}[label=(\roman*)]
        \item We can further partition $(\xi,\xi')$ into some subvectors to analyze sensitivity to distributional assumptions about cross-sectional dependence.
        \item We can also consider higher-order Markov processes by expanding the state space. For example, if the perturbed process is a second-order Markov process, then perturbation set can be defined on the joint distribution of $(\tilde{\xi}_{t},\tilde{\xi}_{t+1})$ where $\tilde{\xi}_{t} = (\xi_{t},\xi_{t+1})$.
    \end{enumerate}
\end{remark}

\subsection{Examples} \label{sec: Examples}

\noindent Our framework applies to a variety of latent variables, including utility shocks, productivity characteristics, labor supply shocks, etc. This section focuses on a parametric model for latent variables beyond utility shocks in infinite and finite-horizon DDC models. Appendix \ref{sec: Additional Examples} considers: (i) serial independence of utility shocks in DDC models, and (ii) serial independence of consumption shocks in dynamic discrete-continuous choice models. The bounds on a scalar parameter are solutions to constrained optimization problems over the perturbation set subject to structural constraints (e.g., Bellman equation) and moment conditions.

\begin{example}[Infinite Horizon Dynamic Discrete Choice Models with Serially Correlated Latent Variables] \label{ex: Dynamic Discrete Choice Models with Serially Correlated Unobserved State Variables}
This example considers a parametric model for serially correlated latent variables in a single-agent DDC model as in \cite{rust1994structural}. Let $\xi \in \Xi$ be the exogenously evolving latent variable (e.g., unobserved productivity characteristics). Agents solve the smoothed\footnote{We assume that the utility shock is additively separable in the period utility function and follows an i.i.d. Extreme Value Type I distribution, leading to the log-sum-exp form of the value function \cite{rust1987optimal}.} Bellman equation for the conditional value function $v \in \mV$ where $\mV$ is a function class (e.g., square integrable functions, the H\"older class, etc.): for $\forall \ (x,\xi,j) \in \mX \times \Xi \times \mJ$,
\begin{equation} \label{eq: softmax Bellman equation}
    v_{j}(x,\xi) = u_{j}(x,\xi;\theta) + \beta \bE_{\xi'|\xi}\bE_{x'|x,j}\left[\log\left(\sum_{j' \in \mJ} \exp(v_{j'}(x',\xi'))\right)\right] + \beta \gamma
\end{equation}
where $x \in \mX$ is the observable state variable, $\beta \in (0,1)$ is the discount factor, $\gamma$ is the Euler constant, and $u_{j}(x,\xi;\theta)$ is the period utility of choosing action $j\in \mJ$ parameterized by $\theta \in \Theta$. The model-implied Conditional Choice Probability (CCP) is $p(j|x,\xi) = \frac{\exp\left(v_{j}(x,\xi)\right)}{\sum_{j' \in \mJ} \exp\left(v_{j'}(x,\xi)\right)}$.

Let $U:=(\xi,\xi')$ be a vector of current and future latent variables. An AR(1) process is often used to model the transition of $\xi$. Therefore, the reference distribution is $dF_{\theta_{f}}(U) := F_{\theta_{f}}(d\xi'|\xi)\nu_{\theta_{f}}(d\xi)$ where $F_{\theta_{f}}(d\xi'|\xi)$ is the conditional distribution parameterized by $\theta_{f} \in \Theta_{f}$ (e.g., the parameters of the AR(1) process), and $\nu_{\theta_{f}}$ is its stationary distribution. The perturbation set for a given $\theta_{f}$ is defined as:
\begin{equation*}
    \mF_{\theta_{f}} := \left\{ F \in \mP(\mU) \mid F \in \Pi(\nu_{\theta_{f}},\nu_{\theta_{f}}), D_{KL}(F\|F_{\theta_{f}}) \leq \delta \right\}
\end{equation*}

Suppose the scalar parameter of interest is the average elasticity of action $j$ with respect to variable $x_{l}$, defined as:
\begin{equation*}
    \bE_{\nu_{\theta_{f}}}\bE_{x}\left[\frac{\partial p(j|x,\xi)}{\partial x_{l}} \frac{x_{l}}{P_{0}(j|x)}\right]
\end{equation*}
where $P_{0}(j|x)$ is the population CCP, and the expectation is taken with respect to the joint distribution $F \in \mF_{\theta_{f}}$ and the distribution of $x$.

We convert the smoothed Bellman equation \eqref{eq: softmax Bellman equation} into an unconditional moment restriction that depends on the joint distribution $F$. We assume\footnote{The Lagrange multiplier function converts the continuum of conditional moment restrictions into a single unconditional moment restriction (see for example \cite{andrews2013inference} and \cite{schennach2014entropic}).} there exists a class of Lagrange multiplier functions $\mG$\footnote{For example, if $\mV$ is the class of square integrable functions, $\mG$ is the class of square integrable functions.} such that $v$ solves the Bellman equation \eqref{eq: softmax Bellman equation} if and only if:
\begin{equation*}
    \sup_{g \in \mG} \bE_{F} \bE_{x,j,x'} \left[g_{j}(x,\xi)\left(v_{j}(x,\xi) - u_{j}(x,\xi;\theta) - \beta \log\left(\sum_{j' \in \mJ} \exp(v_{j'}(x',\xi'))\right) - \beta \gamma\right)\right] = 0
\end{equation*}
where the inner expectation is taken with respect to the stationary distribution of $x$, the population CCPs, and the conditional distribution of $x'$ given $(x,j)$. Let $g:=(g_{j})_{j \in \mJ}$. Then, we rewrite the structural constraints as:
\begin{equation*}
    \sup_{g \in \mG} \bE_{F} \left[\psi(U;\theta,v,g)\right] = 0
\end{equation*}

We consider the following moment conditions for estimation:
\begin{equation*}
    \bE_{\nu_{\theta_{f}}} \left[p(j|x,\xi) \right] = P_{0}(j|x) \quad \forall \ (j,x) \in \mJ \times \mX
\end{equation*}
We assume $\mX$ has discrete support, and rewrite the moment conditions as:
\begin{equation*}
    \bE_{F} \left[m(U;v)\right] = P_{0}
\end{equation*}
where $m(U;v)$ stacks the model-implied CCPs, and $P_{0}$ stacks the population CCPs.

Then, the lower bound on the elasticity is given by:
\begin{align*}
    \inf_{\theta_{f} \in \Theta_{f}}\inf_{(\theta,v,F) \in \Theta \times \mV \times \mF_{\theta_{f}}} \
    & \bE_{\nu_{\theta_{f}}} \bE_{x} \left[\frac{\partial p(j|x,\xi)}{\partial x_{l}} \frac{x_{l}}{P_{0}(j|x)}\right] \\
    \text{s.t.} \quad
    & \bE_{F} \left[m(U;v)\right] = P_{0} \\
    & \sup_{g \in \mG} \bE_{F} \left[\psi(U;\theta,v,g)\right] = 0
\end{align*}

In the next section, we will discuss the implementation for a fixed $\theta_{f}$. The overall lower bound requires an additional optimization over $\theta_{f} \in \Theta_{f}$. In practice, we can discretize the estimated AR(1) process, and scale the grid points according to the candidate $\theta_{f}$ during the optimization. Then, the optimization over $\theta_{f}$ can be implemented using the algorithm proposed in \Cref{sec: Practical Implementation of the Proposed Framework}.

\end{example}

\begin{example}[Finite Horizon Dynamic Discrete Choice Models] \label{ex: Finite Horizon Dynamic Discrete Choice Models}

This example considers a finite-horizon DDC model where the latent variable $\xi \in \Xi$ (e.g., labor supply shocks) follows a first-order stationary Markov process. The conditional value function $v_{t} \in \mV$ at time period $t \leq T < + \infty$ solves the smoothed Bellman equation: for $\forall \ (x_{t},\xi_{t},j) \in \mX \times \Xi \times \mJ$,
\begin{equation} \label{eq: finite Bellman equation}
    v_{jt}(x_{t},\xi_{t}) = u_{j}(x_{t},\xi_{t};\theta) + \beta \bE_{\xi_{t+1}|\xi_{t}}\bE_{x_{t+1}|x_{t},j}\left[\log\left(\sum_{j' \in \mJ} \exp(v_{jt+1}(x_{t+1},\xi_{t+1}))\right)\right] + \beta \gamma
\end{equation}
where $x \in \mX$ is the observable state variable, $\beta \in (0,1)$ is the discount factor, $\gamma$ is the Euler constant, $u_{j}(x,\xi;\theta)$ is the period utility parameterized by $\theta \in \Theta$, and $v_{jT}(x_{T},\xi_{T}) = u_{j}(x_{T},\xi_{T};\theta)$. The model-implied CCP is $p_{t}(j|x,\xi) = \frac{\exp\left(v_{jt}(x,\xi)\right)}{\sum_{j' \in \mJ} \exp\left(v_{j't}(x,\xi)\right)}$.

Let $U:=(\xi,\xi')$ be a vector of current and future latent variables. In practice, we may set the reference distribution as the estimated distribution from a parametric model, such as an AR(1) process. The reference distribution $F_{0}$ is the product of the conditional distribution and its stationary distribution, $\nu_{0}$. The perturbation set is defined as:
\begin{equation*}
    \mF := \left\{ F \in \mP(\mU) \mid F \in \Pi(\nu_{0},\nu_{0}), D_{KL}(F\|F_{0}) \leq \delta \right\}
\end{equation*}
Suppose the scalar parameter of interest is the consumer surplus derived from the choice set $\mJ$ at period $t$:
\begin{equation*}
    \bE_{\nu_{0}} \bE_{x_{t}} \left[\frac{1}{\alpha} \log\left(\sum_{j \in \mJ} \exp(v_{jt}(x_{t},\xi_{t}))\right)\right]
\end{equation*}
where we assume $u_{j}(x_{t},\xi_{t};\theta)$ is linear in price and $\alpha$ is the price coefficient.

We convert the smoothed Bellman equation \eqref{eq: finite Bellman equation} into restrictions that depend on the joint distribution $F \in \mF$. We assume there exists a class of Lagrange multiplier functions $\mG$ such that for each $t\leq T-1$, $v_{t}$ solves \eqref{eq: finite Bellman equation} if and only if:
\begin{equation*}
    \begin{aligned}
        &\sup_{g_{t} \in \mG} \bE_{F} \bE_{x_{t},j_{t},x_{t+1}} \Big[g_{jt}(x_{t},\xi_{t}) \\
        &\qquad \times \Big(v_{jt}(x_{t},\xi_{t}) - u_{j_{t}}(x_{t},\xi_{t};\theta) - \beta \log\left(\sum_{j' \in \mJ} \exp(v_{j't+1}(x_{t+1},\xi_{t+1}))\right) - \beta \gamma\Big)\Big] = 0
    \end{aligned}
\end{equation*}
where $(\xi_{t},\xi_{t+1}) \sim F$, $(x_{t},j_{t})$ is distributed according to the observed data at time $t$, and $x_{t+1}$ follows the conditional distribution given $(x_{t},j_{t})$. Let $g:=(g_{jt})_{j \in \mJ, t \leq T-1}$ and $v:=(v_{jt})_{j \in \mJ, t \leq T-1}$. Then, we rewrite the structural constraints as:
\begin{equation*}
    \sup_{g \in \mG} \bE_{F} \left[\psi(U;\theta,v,g)\right] = 0
\end{equation*}
where $\psi$ is the sum of the objective functions in the above equation for each $t \leq T-1$.

We consider the following moment conditions for estimation: for each $t \leq T$,
\begin{equation*}
    \bE_{\nu_{0}} \left[p_{t}(j|x_{t},\xi)\right] = P_{0t}(j|x_{t}) \quad \forall \ (j,x_{t}) \in \mJ \times \mX
\end{equation*}
where $P_{0t}(j|x_{t})$ is the population CCP at period $t$. We assume $\mX$ has discrete support, and rewrite the moment conditions as:
\begin{equation*}
    \bE_{F} \left[m(U;v)\right] = P_{0}
\end{equation*}
where $m(U;v)$ and $P_{0}$ stack the model-implied and population CCPs for all $t \leq T$.

Then, the lower bound on consumer surplus at period $t$ is given by:
\begin{align*}
    \inf_{(\theta,v,F) \in \Theta \times \mV \times \mF} 
    & \bE_{\nu_{0}} \bE_{x_{t}} \left[\frac{1}{\alpha} \log\left(\sum_{j \in \mJ} \exp(v_{jt}(x_{t},\xi_{t}))\right)\right] \\
    & \text{s.t.} \quad
    \bE_{F} \left[m(U;v)\right] = P_{0} \\
    & \phantom{\text{s.t.} \quad} \sup_{g \in \mG} \bE_{F} \left[\psi(U;\theta,v,g)\right] = 0 
\end{align*}

\end{example}

\subsection{Framework and Duality} \label{sec: Framework}

\noindent We now present a general framework that nests the above examples. In general, the model is not point-identified when $\mF$ is not a singleton.\footnote{For example, see \cite{schennach2014entropic,molinari2020microeconometrics}.} Therefore, we propose to compute upper and lower bounds on the outcome of interest. Let the scalar parameter of interest, $s : \mU \times \Theta \times \mV \rightarrow \bR$, be a function of the latent variable $U$, and the model primitives $(\theta,v) \in \Theta \times \mV$. The lower bound is the solution to the following optimization problem:
\begin{align}
    \kappa(\delta,P) :=
    & \inf_{(\theta,v,F) \in \Theta \times \mV \times \mF} \bE_{F} \left[s(U;\theta,v)\right] \nonumber \\
    & \text{s.t.} \quad 
    \bE_{F} \left[m(U;\theta,v)\right] = P \tag{\blue{Primal}} \label{Primal}\\
    & \phantom{\text{s.t.} \quad} \sup_{g \in \mG} \bE_{F} \left[\psi(U;\theta,v,g)\right] = 0 \nonumber
\end{align}
where $ \mF := \{ F \in \mP(\mU) \mid F \in \Pi(\nu_{1},\cdots,\nu_{k}), D_{KL}(F\|F_{0}) \leq \delta \}$.

The first constraint is a moment condition where the moment function $m : \mU \times \Theta \times \mV \rightarrow \bR^{d_{P}}$ is finite-dimensional as we assume the observable variable $X \in \mX$ has discrete support and stack the moment functions for each $x \in \mX$. The second constraint is a structural constraint defined by $\psi : \mU \times \Theta \times \mV \times \mG \rightarrow \bR$ that is linear in the Lagrange multiplier function $g \in \mG$ for a given $(\theta,v)$. Finally, $v \in \mV$ is the solution to the structural constraint.

\begin{remark}
    \begin{enumerate}[label=(\roman*)]
        \item The upper bound can be obtained by replacing $s(U;\theta,v)$ with $-s(U;\theta,v)$.
        \item The moment condition can contain restrictions linear in $F$, e.g., covariance restrictions.
        \item If some model primitives (e.g., a subvector of $\theta$) are point-identified, they are treated as fixed values rather than optimized over.
        \item The framework can also be applied to models without structural constraints, such as static and panel discrete choice models.
    \end{enumerate}
\end{remark}

The \blueref{Primal} problem can be intractable due to the optimization over $\mF$. First, the properties (e.g., closed-form and smoothness) of the optimal $F^{*}$ are typically unknown, complicating the choice of approximation methods. Second, the marginal distribution conditions are functional constraints imposed directly on the distribution being optimized, which are computationally difficult to impose. Third, expectations are taken with respect to the perturbed distribution, making numerical integration difficult.

To overcome these issues, we derive the \blueref{Dual problem} problem corresponding to the \blueref{Primal}. The \blueref{Dual problem} provides the closed-form of the optimal $F^{*}$, and characterizes its smoothness. Moreover, in the dual, the expectation is taken with respect to the reference distribution $F_{0}$. \Cref{sec: Practical Implementation of the Proposed Framework} proposes a computationally tractable algorithm that utilizes the optimal $F^{*}$. To motivate the duality, consider the Lagrangian of the \blueref{Primal}:
\begin{equation}
    \kappa(\delta,P) = \inf_{\substack{(\theta,v) \in \Theta \times \mathcal{V} \\ F \in \Pi(\nu_{1},\ldots,\nu_{k})}} \sup_{\substack{\lambda \in \bR^{d_{P}} \\ \lambda_{KL} \geq 0, g \in \mG}} \bE_{F} \left[c(U;\theta,v,g,\lambda)\right] + \lambda_{KL}(D_{KL}(F\|F_{0}) - \delta) - \lambda^{T}P \label{eq: Lagrangian}
\end{equation}
where $c(U;\theta,v,g,\lambda) := s(U;\theta,v) + \lambda^{T}m(U;\theta,v) + \psi(U;\theta,v,g)$, $\lambda \in \bR^{d_{P}}$ is the Lagrange multiplier for the moment condition and $\lambda_{KL}$ is the Lagrange multiplier for the KL divergence constraint. For given $(\theta,v)$, under regularity conditions, we can swap the order of the infimum over $F$ and the supremum over $(\lambda,\lambda_{KL},g)$. Then, we can rewrite \eqref{eq: Lagrangian} as:
\begin{equation*}
    \inf_{(\theta,v) \in \Theta \times \mV} \sup_{\substack{\lambda \in \bR^{d_{P}} \\ \lambda_{KL} \geq 0, g \in \mG}} \inf_{F \in \Pi(\nu_{1},\ldots,\nu_{k})} \bE_{F} \left[c(U;\theta,v,g,\lambda)\right] + \lambda_{KL}D_{KL}(F\|F_{0}) - \lambda_{KL} \delta - \lambda^{T}P
\end{equation*}
The inner infimum is the Entropic Optimal Transport (EOT) problem (for $\lambda_{KL} > 0$)\footnote{See \cite{nutz2021introduction} for a comprehensive introduction to the EOT problem. The EOT problem has close connection to the static Schr{\"o}dinger Bridge problem. It is called the Optimal Transport (OT) problem when $\lambda_{KL} = 0$ (see \cite{villani2009optimal}). The optimal value is called the optimal OT value.}:
\begin{equation*}
    \mC(\theta,v,g,\lambda,\lambda_{KL}) := \inf_{\substack{F \in \Pi(\nu_{1},\cdots,\nu_{k})}} \bE_{F} \left[c(U;\theta,v,g,\lambda)\right] + \lambda_{KL} D_{KL}(F\|F_{0})
\end{equation*}
where $c(U;\theta,v,g,\lambda)$ is the cost function, and $\mC(\theta,v,g,\lambda,\lambda_{KL})$ is called the optimal EOT value. The EOT problem is computationally fast to solve using the Sinkhorn algorithm, which relies on the duality of the EOT problem. Moreover, the closed-form and smoothness of the unique solution $F^{*}$ to the EOT problem can be derived from its duality (see \Cref{thm: Strong Duality} and \Cref{thm: Properties of the optimal test function}). \Cref{sec: EOT and Sinkhorn Algorithm} reviews the EOT duality and the Sinkhorn algorithm. Next, we impose assumptions for the minimax theorem to swap the order of infimum and supremum, and the EOT duality to hold:
\begin{assumption} \label{assumption: Strong Duality}
    We assume:
    \begin{assumptionitems}
        \item The marginals $\{\nu_{i}\}_{i=1}^{k}$ have finite $p$-th moment for some integer $p \geq 1$. \label{assumption: finite moment}
        \item $F_{0} \ll F_{\otimes} := \otimes_{i=1}^{k} \nu_{i}$, and let $\rho(U) := \log \frac{dF_{\otimes}(U)}{dF_{0}(U)}$. \label{assumption: reference}
        \item $\mG$ is convex and symmetric, i.e., for $g_{1},g_{2} \in \mG$, $\eta g_{1} + (1-\eta)g_{2} \in \mG$ for $\forall \ \eta \in [0,1]$, and $-g \in \mG$ if $g \in \mG$. Moreover, if $g \in \mG$, then $\eta g \in \mG$ for $\forall \ \eta \geq 0$. \label{assumption: symmetric and convex}
        \item For $\forall \ (\theta,v,g) \in \Theta \times \mV \times \mG$, it holds that $\rho(U),s(U;\theta,v),m(U;\theta,v), \psi(U;\theta,v,g)$ are lower semicontinuous in $U$. \label{assumption: lower semi-continuous}
        \item For $\forall \ (\theta,v,g) \in \Theta \times \mV \times \mG$, there exist a finite positive constant $C_{\theta,v,g}$ and $\hat{U} \in \mU$ such that for $\forall \ U \in \mU$, it holds that $|\rho(U)|+|s(U;\theta,v)|+\|m(U;\theta,v)\|_{1} + |\psi(U;\theta,v,g)| \leq C_{\theta,v,g}(1 + d(U,\hat{U}))$ where $d(U,\hat{U}) := \sum_{i=1}^{k} d_{i}(U_{i}, \hat{U}_{i})^{p}$ and $d_{i}$ is a metric on $\mU_{i}$. \label{assumption: growth rate}
    \end{assumptionitems}
\end{assumption}
Assumptions \ref{assumption: finite moment}-\ref{assumption: lower semi-continuous} are mild. \Cref{assumption: lower semi-continuous} holds for indicator functions. There is no particular necessity to write $\rho(U)$ in log-density form in \Cref{assumption: reference}. Our notation is chosen to simplify the expression in Assumptions \ref{assumption: lower semi-continuous} and \ref{assumption: growth rate}. \Cref{assumption: growth rate} imposes the growth rate condition. It is satisfied for all Examples in \Cref{sec: Examples} if we assume $u$, $g$, and $v$ satisfy the growth rate condition. It ensures that $c(U;\theta,v,g,\lambda) \in L^{1}(F)$ for $\forall \ F \in \mF$, and can also be used to show the convergence of the Sinkhorn algorithm and the convergence of optimal EOT value to optimal OT value as $\lambda_{KL} \downarrow 0$ (\cite{eckstein2022quantitative,eckstein2024convergence}).\footnote{For the convergence of optimal EOT value to optimal OT value, lower semicontinuity alone is not sufficient (\cite{nutz2021introduction} Example 5.1). A sufficient condition is the continuity condition on the cost function.}
\begin{theorem} \label{thm: Strong Duality}
    Let $c(U;\theta,v,g,\lambda) := s(U;\theta,v) + \lambda^{T}m(U;\theta,v) + \psi(U;\theta,v,g)$ where $\lambda \in \bR^{d_{P}}$. Under \Cref{assumption: Strong Duality}, the following holds:
    \begin{theoremitems}
        \item (Minimax Duality) \label{thm: time homogenous minimax duality}
        \begin{equation}
            \kappa(\delta,P) = \inf_{(\theta,v) \in \Theta \times \mV} \sup_{\lambda \in \bR^{d_{P}}, \lambda_{KL} \geq 0, g \in \mG} \mC(\theta,v,g,\lambda,\lambda_{KL}) - \lambda_{KL} \delta - \lambda^{T}P \tag{\blue{Dual}} \label{Dual problem}
        \end{equation}
        where $\mC(\theta,v,g,\lambda,\lambda_{KL})$ is the EOT problem with regularization parameter $\lambda_{KL}$:
        \begin{equation*}
            \mC(\theta,v,g,\lambda,\lambda_{KL}) := \inf_{F \in \Pi(\nu_{1},\cdots,\nu_{k})} \bE_{F} \left[c(U;\theta,v,g,\lambda)\right] + \lambda_{KL} D_{KL}(F\|F_{0})
        \end{equation*}
        \item (Entropic Optimal Transport Duality)
        For $\lambda_{KL} > 0$, we have:\footnote{See \cite{villani2021topics} for optimal transport duality ($\lambda_{KL} = 0$).}
        \begin{equation*}
            \mC(\theta,v,g,\lambda,\lambda_{KL}) = \sup_{\left\{\phi_{i} \in L^{1}(\nu_{i})\right\}_{i=1}^{k}}\sum_{i=1}^{k} \bE_{\nu_{i}} \phi_{i}(U_{i}) - \lambda_{KL} \bE_{F_{0}} \exp\left(\frac{\sum_{i=1}^{k} \phi_{i}(U_{i}) - c(U;\theta,v,g,\lambda)}{\lambda_{KL}}\right) + \lambda_{KL}
        \end{equation*}
        Moreover, there are unique maximizers $\left\{\phi_{i}^{*}\right\}_{i=1}^{k}$ up to additive constants $F_{0}$-almost surely, and the unique worst-case distribution $F^{*}$ has the density of the form:
        \begin{equation*}
            \frac{dF^{*}(U)}{dF_{0}(U)} = \exp(\frac{\sum_{i=1}^{k} \phi_{i}^{*}(U_{i}) - c(U;\theta,v,g,\lambda)}{\lambda_{KL}}) \quad F_{0}\text{-a.s.}
        \end{equation*}
        Furthermore, we have $\mC(\theta,v,g,\lambda,\lambda_{KL}) = \sum_{i=1}^{k} \bE_{\nu_{i}} \phi_{i}^{*}(U_{i})$.
        \item If the lower semicontinuity in \Cref{assumption: lower semi-continuous} is strengthened to continuity, then optimizing over $\lambda_{KL} > 0$ is equivalent to optimizing over $\lambda_{KL} \geq 0$.
    \end{theoremitems}
\end{theorem}
\Cref{thm: Strong Duality} establishes the duality and provides the closed-form of $F^{*}$. $\phi_{i}$ is the test function for the marginal distribution condition $\nu_{i}$ for $i = 1,\cdots,k$. The optimal $\left\{\phi_{i}^{*}\right\}_{i=1}^{k}$ are the optimal EOT potentials, which can be efficiently obtained using the Sinkhorn algorithm.

The \blueref{Dual problem} is computationally tractable. First, it provides the closed-form of the worst-case distribution $F^{*}$, which can be used in the \blueref{Primal} problem even if we want to solve it directly. In \Cref{sec: Practical Implementation of the Proposed Framework}, we propose an iterative algorithm that alternates between solving the EOT to obtain the worst-case distribution and updating $v$ by solving the structural constraint with the worst-case distribution. Second, for given $(\theta,v,g,\lambda,\lambda_{KL})$, the optimal $\left\{\phi_{i}^{*}\right\}_{i=1}^{k}$ (and thus $F^{*}$) can be efficiently computed using the Sinkhorn algorithm, which is computationally very fast. Third, the expectations in the \blueref{Dual problem} are taken with respect to the marginal distributions and the reference distribution. Therefore, numerical integration methods can be determined in advance. Finally, we have:
\begin{proposition} \label{thm: Properties of the optimal test function}
    If $c(U;\theta,v,g,\lambda)$ is $k$-times continuously differentiable in $U$, and $\lambda_{KL} > 0$, then $\left\{\phi_{i}^{*}\right\}_{i=1}^{k}$ are $k$-times continuously differentiable in $U_{i}$. Therefore, $\frac{dF^{*}(U)}{dF_{0}(U)}$ is also $k$-times continuously differentiable in $U$.
\end{proposition}
\Cref{thm: Properties of the optimal test function} shows the smoothness of $\left\{\phi_{i}^{*}\right\}_{i=1}^{k}$ for $\lambda_{KL} > 0$ (the EOT case). For $\lambda_{KL} = 0$ (the OT case), it is not straightforward to obtain the smoothness of the worst-case distribution (see \cite{villani2009optimal} Chapter 12).

\subsection{Perturbation of Stationary Distribution} \label{sec: perturbing stationary distribution}

\noindent This section discusses how to perturb the stationary distribution in the time-homogeneous setting. We consider the serial dependence of the latent variables, as detailed in the example in Section \ref{sec: Definition of Perturbation}. Let $U:=(\xi,\xi')$ be the vector of current and future latent variables. We define the perturbation set for the stationary distribution as $\mN = \{\nu \in \mP(\Xi) \mid D_{KL}(\nu\|\nu_{0}) \leq \delta_{1}\}$ where $\delta_{1} \geq 0$. The perturbation set for the joint distribution is:
\begin{equation*}
    \mF := \left\{ F \in \mP(\Xi^{2}) \mid F \in \Pi(\nu,\nu), \nu \in \mN, D_{KL}(F \| F_{0}) \leq \delta \right\}
\end{equation*}
Let $\kappa_{stationary}(\delta_{1},\delta,P)$ denote the lower bound on the scalar parameter of interest under this perturbation set. Under regularity conditions, we can swap the order of infimum over $F$ and the supremum over $(\lambda,\lambda_{KL},g)$:
\begin{equation*}
    \kappa_{stationary}(\delta_{1},\delta,P) = \inf_{(\theta,v) \in \Theta \times \mV} \sup_{\lambda \in \bR^{d_{P}}, \lambda_{KL} \geq 0, g \in \mG} \inf_{\nu \in \mN} \ \mC(\theta,v,g,\lambda,\lambda_{KL},\nu) - \lambda_{KL} \delta - \lambda^{T}P
\end{equation*}
where $\mC(\theta,v,g,\lambda,\lambda_{KL},\nu)$ is the EOT problem with respect to the perturbed stationary distribution $\nu$: $\mC(\theta,v,g,\lambda,\lambda_{KL},\nu) := \inf_{F \in \Pi(\nu,\nu)} \bE_{F} \left[c(\xi,\xi';\theta,v,g,\lambda)\right] + \lambda_{KL} D_{KL}(F\|F_{0})$. Its dual formulation is:
\begin{equation*}
    \mC(\theta,v,g,\lambda,\lambda_{KL},\nu) = \sup_{\phi_{1},\phi_{2} \in L^{1}(\nu)} \bE_{\nu} \left[\phi_{1}(\xi) + \phi_{2}(\xi')\right] - \lambda_{KL} \bE_{F_{0}} \exp\left(\frac{\phi_{1}(\xi) + \phi_{2}(\xi') - c(\xi,\xi';\theta,v,g,\lambda)}{\lambda_{KL}}\right) + \lambda_{KL}
\end{equation*}

Under regularity conditions, we can swap the order of infimum over $\nu$ and the supremum over $(\phi_{1},\phi_{2})$. Then, the inner infimum over $\nu$ is:
\begin{equation*}
    \inf_{\nu \in \mN} \bE_{\nu} \left[\phi_{1}(\xi) + \phi_{2}(\xi')\right]
\end{equation*}
which is a KL-divergence distributionally robust optimization problem (see \cite{hu2013kullback,rahimian2019distributionally}) whose dual formulation is:
\begin{equation*}
    \sup_{\eta \geq 0 } -\eta \log \bE_{\nu_{0}} \exp\left(-\frac{\phi_{1}(\xi) + \phi_{2}(\xi')}{\eta}\right) - \eta \delta_{1}
\end{equation*}
where $\eta$ is the Lagrange multiplier for the KL divergence constraint for $\nu$. To summarize:
\begin{theorem} \label{thm: stationary perturbation duality}
    Suppose \Cref{assumption: Strong Duality} holds for any $\nu \in \mN$, and $\Xi$ is compact. Then, we have:
    \begin{align*}
        \kappa_{stationary}(\delta_{1},\delta, P) 
        & = \inf_{(\theta,v) \in \Theta \times \mV} \sup_{\substack{\lambda \in \bR^{d_{P}}, \eta \geq 0 \\ \lambda_{KL} \geq 0, g \in \mG \\ \phi_{1},\phi_{2} \in L^{\infty}(\Xi) \\ \phi_1, \phi_2 \text{ l.s.c.}}} -\eta \log \bE_{\nu_{0}} \exp\left(-\frac{\phi_{1}(\xi) + \phi_{2}(\xi')}{\eta}\right) - \eta \delta_{1} + \lambda_{KL} - \lambda_{KL} \delta - \lambda^{T}P \\
        & \hspace{4cm} - \lambda_{KL} \bE_{F_{0}} \exp\left(\frac{\phi_{1}(\xi) + \phi_{2}(\xi') - c(\xi,\xi';\theta,v,g,\lambda)}{\lambda_{KL}}\right)
    \end{align*}
\end{theorem}
\Cref{thm: stationary perturbation duality} shows the dual formulation when the stationary distribution is perturbed. Although the Sinkhorn algorithm does not apply directly, we can sequentially update $(\phi_{1},\phi_{2})$ using the first-order optimality conditions like the Sinkhorn algorithm (see \Cref{sec: EOT and Sinkhorn Algorithm} for a review of the Sinkhorn algorithm).

\section{Methodology for Time-Inhomogeneous Models} \label{sec:extension nonstationary}

\noindent This section extends our framework to the time-inhomogeneous setting. We begin by defining the perturbation set for these models in \Cref{sec: perturbation set nonstationary}. \Cref{sec: finite horizon dynamic discrete choice models} provides an example of a finite-horizon dynamic discrete choice model. \Cref{sec: nonstationary duality} presents the duality result. Finally, \Cref{sec: perturbation of initial distribution} discusses how to perturb the initial distribution.

\subsection{Definition of Perturbation Set} \label{sec: perturbation set nonstationary}

\noindent Consider a sequence of latent variables over a finite horizon, $U:=(\xi_{1},\xi_{2},\cdots,\xi_{T})$, that follows a first-order Markov chain. The reference joint distribution is the product of an initial distribution and transition kernels:
\begin{equation*}
    dF_{0}(U) = \nu_{1}(d\xi_{1}) F_{1}(d\xi_{2}|\xi_{1}) \cdots F_{T-1}(d\xi_{T}|\xi_{T-1})
\end{equation*}
where $\nu_{1}(d\xi_{1})$ is the initial distribution, and $F_{t}(d\xi_{t+1}|\xi_{t})$ is the transition kernel from period $t$ to $t+1$. Let $\nu_{T}(d\xi_{T})$ be the terminal distribution implied by this process.

We consider perturbing the reference distribution while holding its initial and terminal distributions fixed, i.e.,
\begin{equation*}
    \mF_\text{Markov} := \left\{ F \in \mP(\mU) \mid F \in \Pi_{\text{Markov}}(\nu_{1},\nu_{T}), D_{KL}(F\|F_{0}) \leq \delta \right\}
\end{equation*}
where $\Pi_{\text{Markov}}(\nu_{1},\nu_{T})$ is the set of all joint distributions over $\mU$ that satisfy the first-order Markov property\footnote{That is, for any $t \leq T-1$, $F(d\xi_{t+1}|\xi_{1},\cdots,\xi_{t}) = F(d\xi_{t+1}|\xi_{t})$ almost surely under $F$.} and have $\nu_1$ and $\nu_T$ as their initial and terminal marginal distributions. Our perturbation set allows for any transition dynamics of the latent variables between the initial and terminal periods, as long as the overall process remains Markovian and the initial and terminal distributions are fixed.

We will discuss how to perturb the initial distribution in \Cref{sec: perturbation of initial distribution}. The terminal distribution can often be nonparametrically identified; therefore, we fix it. For instance, in a finite horizon DDC model, the terminal period is a static discrete choice problem where the distribution of the latent variable can be identified (\cite{lewbel2000semiparametric,matzkin2007nonparametric}). Moreover, if $\xi_{t}$ is the market-level latent variable, and the model has the finite dependence property (see \cite{arcidiacono2011conditional})\footnote{For example, a model with a terminating action has the finite dependence property.}, then the utility parameters can be identified without a distributional assumption for the latent variables (see \cite{kalouptsidi2021linear}), which in turn identifies the terminal distribution.

\subsection{Example} \label{sec: finite horizon dynamic discrete choice models}

\begin{example}[Finite Horizon DDC with Time-Inhomogeneous Transition of Latent Variables] \label{ex: Finite Horizon Dynamic Discrete Choice Models with Time-Inhomogeneous Transition of Latent Variables}
    This example considers a time-inhomogeneous transition for the latent variables, $U:=(\xi_{1},\xi_{2},\cdots,\xi_{T})$, whose perturbation set is $\mF_\text{Markov}$. The model is similar to \Cref{ex: Finite Horizon Dynamic Discrete Choice Models} but with a time-inhomogeneous transition.
    
    We convert the smoothed Bellman equation \eqref{eq: finite Bellman equation} into restrictions that depend on the sequence of two-period marginal distributions $\{F_{t,t+1}\}_{t=1}^{T-1}$, where:   
    \begin{equation*}
        dF_{t,t+1}(\xi_{t},\xi_{t+1}) = \int_{\xi_{1},\dots,\xi_{t-1},\xi_{t+2},\dots,\xi_{T}} dF(\xi_{1},\dots,\xi_{T}) 
    \end{equation*}
    We assume there exists a class of Lagrange multiplier functions $\mG$ such that for each $t\leq T-1$, $v_{jt}$ solves \eqref{eq: finite Bellman equation} if and only if:
    \begin{equation*}
        \begin{aligned}
            &\sup_{g_{t} \in \mG} \bE_{F_{t,t+1}} \bE_{x_{t},j_{t},x_{t+1}} \Big[g_{jt}(x_{t},\xi_{t}) \\
            &\qquad \times \Big(v_{jt}(x_{t},\xi_{t}) - u_{j_{t}}(x_{t},\xi_{t};\theta) - \beta \log\left(\sum_{j' \in \mJ} \exp(v_{j't+1}(x_{t+1},\xi_{t+1}))\right) - \beta \gamma\Big)\Big] = 0
        \end{aligned}
    \end{equation*}
    where $(\xi_{t},\xi_{t+1}) \sim F_{t,t+1}$, and $(x_{t},j_{t})$ is distributed according to the observed data at time $t$, and $x_{t+1}$ follows the conditional distribution given $(x_t, j_t)$. Let $g := (g_{jt})_{j \in \mJ, t \leq T-1}$ and $v:= (v_{jt})_{j \in \mJ, t\leq T-1}$. We then rewrite the structural constraints as:
    \begin{equation*}
        \sup_{g \in \mG} \bE_{F} \left[\psi(U;\theta,v,g)\right] = 0
    \end{equation*}
    where $\psi$ is the sum over $t \leq T-1$ of terms inside the expectation of the previous equation.
    
    Then, the lower bound on consumer surplus at period $t$ is given by:
    \begin{align*}
        \inf_{(\theta,v,F) \in \Theta \times \mV \times \mF_{\text{Markov}}} 
        & \bE_{\nu_{t}} \bE_{x_{t}} \left[\frac{1}{\alpha} \log\left(\sum_{j \in \mJ} \exp(v_{jt}(x_{t},\xi_{t}))\right)\right] \\
        & \text{s.t.} \quad
        \bE_{F} \left[m(U;\theta,v)\right] = P_{0} \\
        & \phantom{\text{s.t.} \quad} \sup_{g \in \mG} \bE_{F} \left[\psi(U;\theta,v,g)\right] = 0 
    \end{align*}
    where $\nu_{t}$ is the marginal distribution of $\xi_{t}$ implied by $F$.
\end{example}

\subsection{Framework and Duality} \label{sec: nonstationary duality}

\noindent The bound for the time-inhomogeneous case, $\kappa_{\text{TI}}(\delta, P)$, is defined similarly to the time-homogeneous case, but with the optimization performed over $\mF_{\text{Markov}}$:
\begin{align}
    \kappa_{\text{TI}}(\delta, P) := 
    & \inf_{(\theta,v,F) \in \Theta \times \mV \times \mF_{\text{Markov}}} \bE_{F} \left[s(U;\theta,v)\right] \nonumber \\
    & \text{s.t.} \quad
    \bE_{F} \left[m(U;\theta,v)\right] = P \tag{\blue{TI}} \label{eq: nonstationary} \\
    & \phantom{\text{s.t.} \quad} \sup_{g \in \mG} \bE_{F} \left[\psi(U;\theta,v,g)\right] = 0 \nonumber
\end{align}
where “TI" stands for time-inhomogeneous. To solve \blueref{eq: nonstationary}, we follow the procedure in \Cref{thm: Strong Duality}. However, the set $\mF_{\text{Markov}}$ is not necessarily convex, which prevents the proof strategy of the minimax duality in \Cref{thm: Strong Duality} from being applied directly. Therefore, we propose to solve a relaxed problem where the Markov property condition is removed. We then show that, under certain reasonable conditions, the solution to the relaxed problem is Markovian, thereby also solving the original problem. The perturbation set for the relaxed problem is defined as:
\begin{equation*}
    \mF_{\text{relaxed}} := \left\{ F \in \mP(\mU) \mid F \in \Pi(\nu_{1},\nu_{T}), D_{KL}(F\|F_{0}) \leq \delta \right\}
\end{equation*}
where $\Pi(\nu_{1},\nu_{T})$ is the set of joint distributions with initial distribution $\nu_{1}$ and terminal distribution $\nu_{T}$. The relaxed problem is given by:
\begin{align}
    \tilde{\kappa}_{\text{TI}}(\delta, P) := 
    & \inf_{(\theta,v,F) \in \Theta \times \mV \times \mF_{\text{relaxed}}} \bE_{F} \left[s(U;\theta,v)\right] \nonumber \\
    & \text{s.t.} \quad
    \bE_{F} \left[m(U;\theta,v)\right] = P \tag{\blue{Relaxed}} \label{eq: relaxed nonstationary} \\
    & \phantom{\text{s.t.} \quad} \sup_{g \in \mG} \bE_{F} \left[\psi(U;\theta,v,g)\right] = 0 \nonumber
\end{align}
whose Lagrangian is:
\begin{equation}
    \tilde{\kappa}_{\text{TI}}(\delta,P) = \inf_{\substack{(\theta,v) \in \Theta \times \mathcal{V} \\ F \in \Pi(\nu_{1},\nu_{T})}} \sup_{\substack{\lambda \in \bR^{d_{P}} \\ \lambda_{KL} \geq 0, g \in \mG}} \bE_{F} \left[c(U;\theta,v,g,\lambda)\right] + \lambda_{KL}(D_{KL}(F\|F_{0}) - \delta) - \lambda^{T}P \label{eq: Lagrangian1}
\end{equation}
where $c(U;\theta,v,g,\lambda) := s(U;\theta,v) + \lambda^{T}m(U;\theta,v) + \psi(U;\theta,v,g)$, $\lambda \in \bR^{d_{P}}$ is the Lagrange multiplier for the moment condition and $\lambda_{KL}$ is the Lagrange multiplier for the KL divergence constraint. For given $(\theta,v)$, under regularity conditions, we can swap the order of the infimum over $F$ and the supremum over $(\lambda,\lambda_{KL},g)$. Then, we can rewrite \eqref{eq: Lagrangian1} as:
\begin{equation*}
    \inf_{(\theta,v) \in \Theta \times \mV} \sup_{\substack{\lambda \in \bR^{d_{P}} \\ \lambda_{KL} \geq 0, g \in \mG}} \inf_{F \in \Pi(\nu_{1},\nu_{T})} \bE_{F} \left[c(U;\theta,v,g,\lambda)\right] + \lambda_{KL}D_{KL}(F\|F_{0}) - \lambda_{KL} \delta - \lambda^{T}P
\end{equation*}
The inner infimum is:
\begin{equation*}
    \mC_{\text{TI}}(\theta,v,g,\lambda,\lambda_{KL}) := \inf_{\substack{F \in \Pi(\nu_{1},\nu_{T})}} \bE_{F} \left[c(U;\theta,v,g,\lambda)\right] + \lambda_{KL} D_{KL}(F\|F_{0})
\end{equation*}
which can be rewritten as the discrete-time dynamic Schr{\"o}dinger Bridge (SB) problem (see \cite{leonard2013survey,de2021diffusion}). Because it only restricts the initial and terminal distributions, we can decompose it into two parts: the two-period marginal distribution part (the first and last period) and the conditional distribution part (the intermediate variables conditional on the first and last period). The second part is unconstrained, thereby having a closed-form solution. The first part is the static SB (or EOT) problem whose duality is similar to \Cref{thm: Strong Duality}. We impose the following assumptions for the minimax duality, decomposition, static SB duality, and the Markov property:
\begin{assumption} \label{assumption: Strong Duality1}
    Let $dF^{1,T}_{0}(\xi_{1},\xi_{T}) := \int_{\xi_{2},\dots,\xi_{T-1}} dF_{0}(\xi_{1},\xi_{2},\dots,\xi_{T})$ be the two-period marginal of $F_{0}$ at  periods $1$ and $T$. We assume:
    \begin{assumptionitems}
        \item $\mU$ is compact. \label{assumption: compact Support}
        \item $F^{1,T}_{0} \sim \nu_{1} \otimes \nu_{T}$, i.e., $F^{1,T}_{0}$ and $\nu_{1} \otimes \nu_{T}$ are mutually absolutely continuous. Moreover, $\log \frac{d(\nu_{1} \otimes \nu_{T})}{dF^{1,T}_{0}}\in L^{1}(\nu_{1} \otimes \nu_{T})$. \label{assumption: reference1}
        \item \label{assumption: bounded1} For $\forall \ (\theta,v,g) \in \Theta \times \mV \times \mG$, it holds that $|s(U;\theta,v)|+\|m(U;\theta,v)\|_{1} + |\psi(U;\theta,v,g)| < \infty$.
        \item \label{assumption: nonstationary Markov} The functionals $s(U;\theta,v)$, $m(U;\theta,v)$, and $\psi(U;\theta,v,g)$ are pairwise additive, i.e., $s(U;\theta,v) = \sum_{t=1}^{T-1} s_{t}(\xi_{t},\xi_{t+1};\theta,v_{t},v_{t+1})$, $m(U;\theta,v) = \sum_{t=1}^{T-1} m_{t}(\xi_{t},\xi_{t+1};\theta,v_{t},v_{t+1})$, and $\psi(U;\theta,v,g) = \sum_{t=1}^{T-1} \psi_{t}(\xi_{t},\xi_{t+1};\theta,v_{t},v_{t+1},g_{t},g_{t+1})$ for some functions $s_{t}$, $m_{t}$, and $\psi_{t}$.
    \end{assumptionitems}
\end{assumption}
The boundedness condition in \Cref{assumption: bounded1} is stronger than \Cref{assumption: growth rate}, which does not guarantee that $c(U;\theta,v,g,\lambda) \in L^{1}(F)$ for any $F \in \mF_{\text{relaxed}}$. \Cref{assumption: reference1} is a sufficient condition for the SB duality to hold. Finally, \Cref{assumption: nonstationary Markov} is the key to the Markov property of the solution to the \blueref{eq: relaxed nonstationary} problem, and is satisfied in \Cref{ex: Finite Horizon Dynamic Discrete Choice Models with Time-Inhomogeneous Transition of Latent Variables}. It does not hold if the moment function depends on the entire path of the latent variables.
\begin{theorem} \label{theorem: nonstationary duality}
    Let $c(U;\theta,v,g,\lambda) := s(U;\theta,v) + \lambda^{T}m(U;\theta,v) + \psi(U;\theta,v,g)$
    where $\lambda \in \bR^{d_{P}}$. Under Assumptions \ref{assumption: finite moment}, \ref{assumption: symmetric and convex}, \ref{assumption: lower semi-continuous}, and \ref{assumption: Strong Duality1}, the following holds:
    \begin{theoremitems}
        \item (Minimax Duality) \label{thm: nonstationary duality}
        \begin{equation*}
            \tilde{\kappa}_{\text{TI}}(\delta,P) = \inf_{(\theta,v) \in \Theta \times \mV} \sup_{\lambda \in \bR^{d_{P}}, \lambda_{KL} \geq 0, g \in \mG} \mC_{\text{TI}}(\theta,v,g,\lambda,\lambda_{KL}) - \lambda_{KL} \delta - \lambda^{T}P
        \end{equation*}
        where $\mC_{\text{TI}}(\theta,v,g,\lambda,\lambda_{KL})$ is defined as:
        \begin{equation}
            \mC_{\text{TI}}(\theta,v,g,\lambda,\lambda_{KL}) := \inf_{F \in \Pi(\nu_{1},\nu_{T})} \bE_{F} \left[c(U;\theta,v,g,\lambda)\right] + \lambda_{KL} D_{KL}(F\|F_{0}) \tag{\blue{$\text{S}_{\text{dyn}}$}} \label{eq: nonstationary EOT}
        \end{equation}
        \item \label{thm: nonstationary worst-case} For $\lambda_{KL} > 0$, the unique worst-case distribution to \blueref{eq: nonstationary EOT} has the density of the form:
        \begin{equation*}
           \frac{dF^{*}(U)}{dF_{0}(U)} = \exp\left(\frac{\phi_{1}^{*}(\xi_{1}) + \phi_{T}^{*}(\xi_{T}) - c(U;\theta,v,g,\lambda)}{\lambda_{KL}}\right)
        \end{equation*}
        where $\phi_{1}^{*}(\xi_{1})$ and $\phi_{T}^{*}(\xi_{T})$ are the unique maximizers (up to an additive constant) to:
        \begin{equation*}
            \sup_{\phi_{1} \in L^{1}(\nu_{1}),\phi_{T} \in L^{1}(\nu_{T})}\bE_{\nu_{1}} \phi_{1}(\xi_{1}) + \bE_{\nu_{T}} \phi_{T}(\xi_{T}) - \lambda_{KL} \bE_{R_{1,T}} \exp\left(\frac{\phi_{1}(\xi_{1})+\phi_{T}(\xi_{T})}{\lambda_{KL}}\right) + \lambda_{KL}
        \end{equation*}
        where the auxiliary reference measure $R_{1,T}$ is defined as:
        \begin{equation*}
            dR_{1,T}(\xi_{1},\xi_{T}) := \int_{\xi_2, \dots, \xi_{T-1}} \exp\left(\frac{-c(U;\theta,v,g,\lambda)}{\lambda_{KL}}\right) \, dF_0(\xi_1, \dots, \xi_T)
        \end{equation*}
        Furthermore, the solution $F^{*}$ has the Markov property, i.e., $F^{*} \in \Pi_{\text{Markov}}(\nu_{1},\nu_{T})$.
        \item (Equivalence) Suppose there exists an optimal $\lambda_{KL}^{*} > 0$, then: $\kappa_{\text{TI}}(\delta,P) = \tilde{\kappa}_{\text{TI}}(\delta,P)$.
    \end{theoremitems}
\end{theorem}
\Cref{theorem: nonstationary duality} shows the duality for the \blueref{eq: relaxed nonstationary} problem. The difference between \Cref{thm: nonstationary duality} and \Cref{thm: time homogenous minimax duality} is that \eqref{eq: nonstationary EOT} does not fix the intermediate marginal distributions. Therefore, we can decompose \eqref{eq: nonstationary EOT} into the sum of two-period marginal distribution ($F^{1,T}_{0}$) part, and the conditional distribution (the distribution of $(\xi_{2},\cdots,\xi_{T-1})$ given $\xi_{1}$ and $\xi_{T}$) part. The latter part is an unconstrained optimization problem, thereby having a closed-form solution. The first part is the static SB problem, whose duality is given by \Cref{thm: nonstationary worst-case}.

\Cref{assumption: nonstationary Markov} is crucial for the solution to have the Markov property. Under this assumption, the cost function $c(U;\theta,v,g,\lambda)$ is also pairwise additive. Therefore, the density ratio in \Cref{thm: nonstationary worst-case} has the Markov property.\footnote{It can be treated as the (unnormalized) pairwise Markov random field \cite{wainwright2008graphical}.} If there exists one optimal $\lambda_{KL}^{*} > 0$, then the \blueref{eq: nonstationary} problem is equivalent to the \blueref{eq: relaxed nonstationary} problem.

The \blueref{eq: relaxed nonstationary} problem can also be solved using the iterative algorithm proposed in \Cref{sec: Practical Implementation of the Proposed Framework}. There is an additional step to obtain the two-period auxiliary reference distribution $R_{1,T}$. The \blueref{eq: nonstationary EOT} can also be solved using the Sinkhorn algorithm.

\subsection{Perturbation of Initial Distribution} \label{sec: perturbation of initial distribution}

\noindent Let $\mN$ be a convex closed set around the initial distribution $\nu_{1}$, e.g., $\mN = \{\nu \in \mP(\Xi) \mid D_{KL}(\nu\|\nu_{1}) \leq \delta_{1}\}$ where $\delta_{1} \geq 0$. Then, the perturbation set is defined as:
\begin{equation*}
    \mF_{\mN,\text{Relaxed}} := \left\{ F \in \mP(\mU) \mid F \in \Pi(\nu,\nu_{T}), \nu \in \mN, D_{KL}(F \| F_{0}) \leq \delta \right\}
\end{equation*}
Let $\tilde{\kappa}_{\text{TI,Initial}}(\delta,P)$ be the lower bound on the scalar parameter for the \blueref{eq: relaxed nonstationary} problem with the perturbation set $\mF_{\mN,\text{Relaxed}}$. The following minimax duality similar to \Cref{theorem: nonstationary duality} holds:
\begin{theorem}[Minimax Duality with Perturbation of Initial Distribution] \label{thm: nonstationary duality with perturbation}
    Suppose $\mN$ is convex and closed, and that the assumptions in \Cref{theorem: nonstationary duality} hold for each $\nu \in \mN$. Then,
    \begin{equation*}
        \tilde{\kappa}_{\text{TI,Initial}}(\delta,P) = \inf_{(\theta,v) \in \Theta \times \mV} \sup_{\lambda \in \bR^{d_{P}}, \lambda_{KL} \geq 0, g \in \mG} \mC_{\text{TI}}(\theta,v,g,\lambda,\lambda_{KL}) - \lambda_{KL} \delta - \lambda^{T}P
    \end{equation*}
    where $\mC_{\text{TI}}(\theta,v,g,\lambda,\lambda_{KL})$ is defined as follows:
    \begin{equation}
        \mC_{\text{TI,Initial}}(\theta,v,g,\lambda,\lambda_{KL}) := \inf_{\nu \in \mN} \inf_{F \in \Pi(\nu,\nu_{T})} \bE_{F} \left[c(U;\theta,v,g,\lambda)\right] + \lambda_{KL} D_{KL}(F\|F_{0}) \label{eq: nonstationary EOT with perturbation}
    \end{equation}
\end{theorem}
To solve \eqref{eq: nonstationary EOT with perturbation}, as shown in the proof of \Cref{thm: nonstationary worst-case} and \Cref{sec: nonstationary duality}, we can decompose the problem into two parts: the two-period marginal distribution part, and the conditional distribution part. The first part requires solving:
\begin{equation*}
    \inf_{\nu \in \mN} \inf_{F_{1,T} \in \tilde{\Pi}(\nu,\nu_{T})} \bE_{F_{1,T}} \left[D_{KL}(F_{1,T}\|R_{1,T})\right]
\end{equation*}
where $\tilde{\Pi}(\nu,\nu_{T})$ is the set of distributions of $(\xi_{1},\xi_{T})$ whose marginal distributions are $\nu$ and $\nu_{T}$, respectively. It is equivalent to solving:
\begin{equation*}
    \inf_{\nu \in \mN} \inf_{F_{1,T} \in \tilde{\Pi}(\nu,\nu_{T})} \bE_{F_{1,T}} \left[\log\left(\frac{d(\nu \otimes \nu_{T})}{dR_{1,T}}\right)\right] + D_{KL}(F_{1,T}\| \nu \otimes \nu_{T})
\end{equation*}
The inner infimum is an EOT problem. Let $\text{EOT}(\nu,\nu_{T})$ be its optimal value.
\begin{lemma} \label{lemma: convexity of EOT}
    Under the assumptions in \Cref{thm: nonstationary duality with perturbation}, $\text{EOT}(\nu,\nu_{T})$ is convex in $\nu$. Its directional derivative with respect to $\nu$ in the direction $\nu'$ is given by:
    \begin{equation*}
        \lim_{\epsilon \downarrow 0} \frac{\text{EOT}(\nu + \epsilon(\nu' - \nu), \nu_{T}) - \text{EOT}(\nu, \nu_{T})}{\epsilon} = \int \phi^{*} d(\nu' - \nu)
    \end{equation*}
    where $\phi^{*}$ is the optimal EOT potential for $\nu$.
\end{lemma}
\Cref{lemma: convexity of EOT} shows that \eqref{eq: nonstationary EOT with perturbation} is a convex optimization problem with respect to $\nu$. Moreover, $\phi^{*}$ is a result of the Sinkhorn algorithm, which can be used to search the optimal $\nu$ efficiently.

\section{Large Sample Properties} \label{sec: Large Sample Properties and Inference}

\noindent This section establishes the large sample properties of the estimator for the bound. \Cref{sec: Consistency and Convergence Rate} proposes a consistent estimator and shows its convergence rate. \Cref{sec: Inference} establishes the asymptotic distribution of the plug-in estimator for the bound.

\subsection{Consistency and Convergence Rate} \label{sec: Consistency and Convergence Rate}

\noindent The bound on the scalar parameter is the projection of the identified set defined by the moment conditions and structural constraints onto the scalar parameter. We follow \cite{chernozhukov2007estimation} to propose an estimator for the identified set and show its consistency and convergence rate. Let $P_{n}$ be an estimator for $P_{0}$ where $n$ is the sample size. Denote by $\epsilon_{n} \in \bR_{+}$ the tolerance level for the moment conditions that goes to zero at a suitable rate as $n \to \infty$. Our estimators $\kappa(\delta,P_{n},\epsilon_{n})$, and $\tilde{\kappa}_{\text{TI}}(\delta,P_{n},\epsilon_{n})$ for the bounds replace the moment conditions by the approximate moment conditions.\footnote{To compute the estimator, the number of moment conditions is doubled due to the use of approximate moment conditions. The duality is similar to Theorems \ref{thm: Strong Duality} and \ref{theorem: nonstationary duality}, thus is omitted for brevity.}
\begin{assumption} \label{assumption: Asymptotic Properties1}
    Let $\mA$ be either $\Theta \times \mF$ or $\Theta \times \mF_{\text{relaxed}}$, and $\alpha := (\theta,F) \in \mA$. Assume:
    \begin{assumptionitems}
        \item $\Theta \subseteq \bR^{d_{\theta}}$ is convex and compact. \label{assumption: compact}
        \item If $\mA = \Theta \times \mF_{\text{relaxed}}$, then \Cref{assumption: compact Support} holds.
        \item For $\forall \ \alpha \in \mA$, the \textit{structural constraint} $F$-a.s. has a unique solution $v(\alpha) \in \mV$. \label{assumption: Asymptotic Property unique solution}
        \item The identified set $ \mA_{I} := \left\{ \alpha \in \mA \mid \bE_{F} \left[m(U;\theta,v(\alpha))\right] = P_{0} \right\}$ is nonempty. \label{assumption: Asymptotic Property nonempty}
        \item \label{assumption: Asymptotic Property continuous} $\bE_{F} \left[m(U;\theta,v(\alpha))\right]$ is continuous in $\alpha \in \mA$, i.e., the preimages of closed sets are closed.
    \end{assumptionitems}
\end{assumption}
\Cref{assumption: compact} is mild. \Cref{assumption: Asymptotic Property unique solution} holds in single-agent DDC models. It rules out dynamic games with multiple equilibria. \Cref{assumption: Asymptotic Property nonempty} is also mild as the identified set for $\theta$ is usually nonempty under the reference distribution $F_{0}$. Moreover, if the identified set for $\delta = +\infty$\footnote{In this case, the KL divergence constraint is replaced by the absolute continuity constraint, i.e., $F \ll F_{0}$.} is nonempty, then \Cref{assumption: Asymptotic Property nonempty} implicitly assumes the radius $\delta$ is large enough. The smallest radius such that the identified set is nonempty can be estimated (see \Cref{remark: delta star}). \Cref{assumption: Asymptotic Property continuous} implies that the identified set and its estimator are compact as $\mF$ is compact (see \Cref{lemma: perturbation set properties}) and thus $\mA$ is compact (see \Cref{lemma: identified and estimator are compact}).

Under \Cref{assumption: Asymptotic Properties1}, the estimator for $\mA_{I}$ is defined as:
\begin{equation*}
    \hat{\mA}_{I} := \left\{ \alpha \in \mA \mid \|\bE_{F} \left[m(U;\theta,v(\alpha))\right] - P_{n}\|_{\infty} \leq \epsilon_{n} \right\}
\end{equation*}
The analysis of the consistency and convergence rate uses the \textit{Hausdorff Distance}:
\begin{equation*}
    d_{H}(\mA_{1}, \mA_{2}) := \max \left\{\sup_{\alpha_{1} \in \mA_{1}} d(\alpha_{1},\mA_{2}), \sup_{\alpha_{2} \in \mA_{2}} d(\alpha_{2},\mA_{1})\right\}
\end{equation*}
where $d(\alpha_{1},\mA_{1}) := \inf_{\alpha_{2} \in \mA_{2}} d(\alpha_{1},\alpha_{2})$ and $d(\alpha_{1},\alpha_{2})$ is a metric on $\mA$.

\begin{assumption} \label{assumption: Asymptotic Properties2}
    Assume $P_{n}$ is a $\sqrt{n}$-consistent estimator for $P_{0}$, and there exists $c_{n}$ such that $\sqrt{n} \|P_{0} - P_{n}\|_{\infty} \leq c_{n}$ with probability approaching 1 where $c_{n}$ can be data-dependent. Let $\epsilon_{n} = \frac{c_{n}}{\sqrt{n}}$, and assume $\epsilon_{n} \xrightarrow{p} 0$.
\end{assumption}
\Cref{assumption: Asymptotic Properties2} is mild as we assumed the observable variable has discrete support, e.g., $P_{n}$ can be the frequency estimator. In practice, we can set $c_{n} \propto \log n$. Then, the convergence rate in \Cref{thm: Convergence Rate of the Identified Set} is $\sqrt{n}$-consistent up to a logarithmic factor. We show some properties of the identified set and its estimator:
\begin{lemma} \label{lemma: identified and estimator are compact}
    Under Assumptions \ref{assumption: Asymptotic Properties1} and \ref{assumption: Asymptotic Properties2}, $\mA_{I}, \hat{\mA_{I}}$ are closed and compact. Moreover, $\hat{\mA_{I}}$ is nonempty.
\end{lemma}
By the extreme value theorem, the infimum is achieved if the scalar parameter is continuous on $\mA$, i.e., \Cref{assumption: continuity} holds. Therefore, the optimization problem has a solution. Next, we impose the polynomial minorant condition as in \cite{chernozhukov2007estimation} for the convergence rate of the estimator:
\begin{assumption}[Polynomial Minorant Condition] \label{assumption: Polynomial Minorant}
    There exists positive constants $C_{1}$ and $C_{2}$ such that: $\|\bE_{F} \left[m(U;\theta,v(\alpha))\right] - P_{0}\|_{\infty} \geq C_{1} \min\left\{C_{2}, d(\alpha, \mA_{I})\right\}$.
\end{assumption}

\begin{theorem} \label{thm: Consistency and Convergence Rate of the Identified Set}
    Under Assumptions \ref{assumption: Asymptotic Properties1} and \ref{assumption: Asymptotic Properties2}, we have:
    \begin{theoremitems}
        \item \label{thm: Consistency} (Consistency) $d_{H}(\hat{\mA}_{I},\mA_{I}) = o_{p}(1)$.
        \item \label{thm: Convergence Rate of the Identified Set} (Convergence Rate) Under \Cref{assumption: Polynomial Minorant}, $d_{H}(\hat{\mA}_{I},\mA_{I}) = O_{p}(\frac{\max \left\{1, c_{n}\right\}}{\sqrt{n}})$.
    \end{theoremitems}
\end{theorem}
\Cref{thm: Consistency and Convergence Rate of the Identified Set} establishes the $\sqrt{n}$-consistency up to a logarithmic factor (if $c_{n} \propto \log n$). Then, we impose the following continuity assumption on the scalar parameter of interest:
\begin{assumption}
    Let $s(\alpha) := \bE_{F}\left[s(U;\theta,v(\alpha))\right]$, assume one of the following:
    \begin{assumptionitems}
        \item $s(\alpha)$ is continuous in $\alpha \in \mA$. \label{assumption: continuity}
        \item \label{assumption: Lipschitz Continuity} $s(\alpha)$ is Lipschitz continuous in $\alpha \in \mA$.
    \end{assumptionitems}
\end{assumption}

\begin{theorem} \label{thm: Consistency and Convergence Rate of the Scalar of Interest}
    Under Assumptions \ref{assumption: Asymptotic Properties1} and \ref{assumption: Asymptotic Properties2}, we have:
    \begin{theoremitems}
        \item \label{thm: Consistency of the Scalar of Interest} (Consistency) Under \Cref{assumption: continuity}, $\kappa(\delta,P_{n},\epsilon_{n}) \xrightarrow{p} \kappa(\delta,P_{0})$, and $\tilde{\kappa}_{\text{TI}}(\delta,P_{n},\epsilon_{n}) \xrightarrow{p} \tilde{\kappa}_{\text{TI}}(\delta,P_{0})$.
        \item \label{thm: Convergence Rate of the Scalar of Interest} (Convergence Rate) Under \Cref{assumption: Lipschitz Continuity}, $|\kappa(\delta,P_{n},\epsilon_{n}) - \kappa(\delta,P_{0})| = O_{p}(\frac{\max \left\{1, c_{n}\right\}}{\sqrt{n}})$, and $|\tilde{\kappa}_{\text{TI}}(\delta,P_{n},\epsilon_{n}) - \tilde{\kappa}_{\text{TI}}(\delta,P_{0})| = O_{p}(\frac{\max \left\{1, c_{n}\right\}}{\sqrt{n}})$.
    \end{theoremitems}
\end{theorem}

\subsection{Asymptotic Distribution} \label{sec: Inference}

\noindent This section establishes the asymptotic distribution of $\kappa(\delta,P_{n})$ and $\tilde{\kappa}_{\text{TI}}(\delta,P_{n})$. To this end, we first show the Hadamard directional differentiability of $\kappa(\delta,P)$ and $\tilde{\kappa}_{\text{TI}}(\delta,P)$ with respect to $P$ at $P_{0}$ similar to \cite{christensen2023counterfactual}.  We begin with the definition of Hadamard directional differentiability:
\begin{definition}
    The map $f: \bR^{d_{P}} \to \bR$ is Hadamard directionally differentiable at $P \in \bR^{d_{P}}$, if there exists a continuous map $f' : \bR^{d_{P}} \to \bR$ such that for $h \in \bR^{d_{P}}$, we have:
    \begin{equation*}
        \lim_{i \to \infty} \frac{f(P + t_{i} h_{i}) - f(P)}{t_{i}} = f'(P;h)
    \end{equation*}
    for all sequences $\left\{h_{i}\right\} \subseteq \bR^{d_{P}}$, $t_{i} \downarrow 0$, and $h_{i} \to h \in \bR^{d_{P}}$ as $i \to \infty$.
\end{definition}

Under \Cref{assumption: Asymptotic Properties1}, we can restate the optimization problem as:
\begin{align*}
    \inf_{\alpha \in \mA} s(\alpha) \quad \text{s.t.} \quad P(\alpha) = P_{0}
\end{align*}
where $P(\alpha) := \bE_{F} \left[m(U;\theta,v(\alpha))\right]$. Moreover, the identified set $\mA_{I}$ is nonempty, which means the feasible set for the optimization problem is nonempty. By \Cref{lemma: identified and estimator are compact}, $\mA_{I}$ is compact. Under \Cref{assumption: continuity}, the Extreme Value Theorem (see \cite{rudin1976principles} Theorem 4.16.) implies that the infimum is attained. Denote by $\mA_{I,\text{TH}}^{*}$, $\mA_{I,\text{TI}}^{*}$ the nonempty sets of optimizers for the problems $\kappa(\delta,P_{0})$ and $\tilde{\kappa}_{\text{TI}}(\delta,P_{0})$, respectively.

To establish Hadamard directional differentiability of $\kappa(\delta,P)$ and $\tilde{\kappa}_{\text{TI}}(\delta,P)$ at $P_{0}$, we impose assumptions similar to those in \cite{bonnans2013perturbation} Theorem 4.25.\footnote{We work on the primal problem to show the Hadamard directional differentiability, while \cite{christensen2023counterfactual} works on the dual problem (see their Theorem 6.2).}
\begin{assumption} \label{assumption: continuous differentiability of s and v}
    Assume $s(\alpha)$ and $P(\alpha)$ are continuously differentiable on $\mA$. That is, they are G\^{a}teaux differentiable on $\mA$ and the corresponding derivatives $Ds(\alpha)$ and $DP(\alpha)$ are continuous on $\mA$ (in the operator norm topology).\footnote{For a given direction $\alpha_{1} \in \mA$, the G\^{a}teaux derivatives are understood as $Ds(\alpha)(\alpha_{1}-\alpha)$ and $DP(\alpha)(\alpha_{1}-\alpha)$. See \cite{bonnans2013perturbation} Page 35 for the definition of G\^{a}teaux derivative.}
\end{assumption}
\begin{assumption} \label{assumption: hadamard}
Let $\mA_{I}^{*}$ be either $\mA_{I,\text{TH}}^{*}$ or $\mA_{I,\text{TI}}^{*}$. Assume:
\begin{assumptionitems}
    \item $0 \in \text{int} \{ DP(\alpha)(\mA-\alpha) \}$ for $\forall \ \alpha \in \mA_{I}^{*}$. \label{assumption: constraint qualification}
    \item For $\forall \ h \in \bR^{d_{P}}$, it holds that for $\forall \ P_{t} := P_{0} + th + o(t)$ and $t > 0$ small enough, the problem $\kappa(\delta,P_{t})$ has an $o(t)$-optimal solution $\alpha(t)$ such that $d(\alpha(t),\mA_{I}^{*}) = O(t)$.
    \item For $\forall \ t_{n} \downarrow 0$ the sequence $\{\alpha(t_{n})\}$ has a limit point (in the norm topology) $\alpha_{0} \in \mA_{I}^{*}$. \label{assumption: convergence of optimizer}
\end{assumptionitems}
\end{assumption}

\begin{theorem} \label{thm: Hadamard Directional Differentiability}
    Under Assumptions \ref{assumption: Asymptotic Properties1}, \ref{assumption: continuous differentiability of s and v}, and \ref{assumption: hadamard}, the maps $\kappa(\delta,P)$ and $\tilde{\kappa}_{\text{TI}}(\delta,P)$ are Hadamard directionally differentiable at $P_{0}$ in any direction $h \in \bR^{d_{P}}$, and:
    \begin{equation*}
        \kappa'(\delta,P_{0};h) = \inf_{\alpha \in \mA_{I,\text{TH}}^{*}} \sup_{\lambda \in \Lambda(\alpha,P_{0})} - \lambda^{T} h, \quad \tilde{\kappa}_{\text{TI}}'(\delta,P_{0};h) = \inf_{\alpha \in \mA_{I,\text{TI}}^{*}} \sup_{\lambda \in \Lambda(\alpha,P_{0})} - \lambda^{T} h
    \end{equation*}
    where $\Lambda(\alpha,P_{0})$ is the nonempty set of Lagrange multipliers corresponding to $\alpha \in \mA_{I}^{*}$.\footnote{See \cite{bonnans2013perturbation} Definition 3.8 and Theorem 3.9. Robinson's constraint qualification is satisfied under Assumptions \ref{assumption: Asymptotic Properties1}, \ref{assumption: continuous differentiability of s and v}, and \ref{assumption: hadamard} (see \Cref{appendix: proof of Hadamard Directional Differentiability}). Therefore, $\Lambda(\alpha^{*},P_{0})$ is nonempty.}
    
    Moreover, if $\sqrt{n}(P_{n} - P_{0}) \xrightarrow{d} Z \sim \mathcal{N}(0,\Sigma)$, then $\sqrt{n}(\kappa(\delta,P_{n}) - \kappa(\delta,P_{0})) \xrightarrow{d} \kappa'(\delta,P_{0};Z)$ and $\sqrt{n}(\tilde{\kappa}_{\text{TI}}(\delta,P_{n}) - \tilde{\kappa}_{\text{TI}}(\delta,P_{0})) \xrightarrow{d} \tilde{\kappa}_{\text{TI}}'(\delta,P_{0};Z)$.
\end{theorem}
\Cref{thm: Hadamard Directional Differentiability} shows the asymptotic distribution of the bound's estimator. To conduct inference, we may follow the procedure in \cite{fang2019inference}. In addition, the numerical delta method \cite{hong2018numerical} combined with our practical implementation in \Cref{sec: Practical Implementation} can be used to overcome the computational challenge.

\section{Interpreting the Results} \label{sec: Derivative with respect to delta}

\noindent In practice, we can estimate an alternative (parametric) model and set the radius to be the KL divergence between the alternative and the reference distribution. In addition, this section considers three complementary sensitivity measures to interpret the results: global sensitivity, local sensitivity, and robustness metric.

\subsection{Global Sensitivity}

\noindent The global sensitivity\footnote{See \cite{christensen2023counterfactual} Theorem 2.1 for similar results. However, their results are silent about how large the radius should be so that the bounds are close to the nonparametric bounds.} approach progressively increases the radius until the bounds flatten. We show that it provides a computationally tractable approximation to the nonparametric bounds when the KL divergence constraint is removed. Moreover, we provide an explicit upper bound on the approximation error. We focus on the time-homogeneous case, for which the “nonparametric" perturbation set is:
\begin{equation*}
    \mF_{+\infty} := \Pi(\nu_{1},\cdots,\nu_{k})
\end{equation*}
After applying the minimax duality, we need to solve the following problem:
\begin{equation*}
    \inf_{(\theta,v) \in \Theta \times \mV} \sup_{\lambda \in \bR^{d_{P}}, g \in \mG} \inf_{F \in \Pi(\nu_{1},\ldots,\nu_{k})} \bE_{F} \left[c(U;\theta,v,g,\lambda)\right] - \lambda^{T}P
\end{equation*}
where the inner problem is an OT problem, which is computationally challenging in high-dimensional settings. The EOT is a computationally tractable approximation to the OT problem. Recall $\mC(\theta,v,g,\lambda,\lambda_{KL}) := \inf_{F \in \Pi(\nu_{1},\cdots,\nu_{k})} \bE_{F} \left[c(U;\theta,v,g,\lambda)\right] + \lambda_{KL} D_{KL}(F\|F_{0})$.
\begin{theorem}[Adapted from \cite{eckstein2024convergence} Theorem 3.1(i)] \label{thm: global sensitivity}
    Suppose \Cref{assumption: Strong Duality} holds. Assume the marginals $\{\nu_{i}\}_{i=1}^{k}$ have finite $p+\eta$-th moment for some $\eta > 0$ and integer $p \geq 1$, and $c(U;\theta,v,g,\lambda)$ satisfies the $A_{L,C}$ condition in \cite{eckstein2024convergence} where $L,C$ depend on $(\theta,v,g,\lambda)$. Let $d_{i}$ be the dimension of $U_{i}$. Then, for any $\lambda_{KL} \in (0,1]$,
    \begin{equation*}
        0 \leq \mC(\theta,v,g,\lambda,\lambda_{KL}) - \mC(\theta,v,g,\lambda,0) \leq \left(\sum_{i=2}^{k} d_{i}\right) \lambda_{KL} \log\left(\frac{1}{\lambda_{KL}}\right) + (k-1)^{\frac{1}{p}} L C \lambda_{KL}
    \end{equation*}
\end{theorem}
\Cref{thm: global sensitivity} provides an explicit upper bound on the approximation error of $\mC(\theta,v,g,\lambda,\lambda_{KL})$ to $\mC(\theta,v,g,\lambda,0)$. For DDC models, the constants $L$ and $C$ can be explicitly characterized under additional conditions (see \cite{eckstein2022quantitative} Lemma 3.5, and \cite{eckstein2024convergence} Remark 2.1.) The upper bound strictly decreases to zero as $\lambda_{KL} \downarrow 0$. Therefore, we can choose a sufficiently small $\lambda_{KL}$ (or sufficiently large $\delta$) to achieve a desired accuracy for the approximation. Our framework thus approximates the nonparametric bounds in a computationally tractable way with an explicitly quantifiable approximation error.

\begin{remark}
    Under certain conditions, one can establish the convergence of the EOT worst-case distribution to the OT worst-case distribution as the regularization parameter $\lambda_{KL} \downarrow 0$. \cite{nutz2021introduction} Theorem 5.5 provides one sufficient condition: the existence of a solution $F^{*}$ to the OT problem such that $D_{KL}(F^{*}\|F_{0}) < +\infty$.
\end{remark}

\subsection{Local Sensitivity}

\noindent The local sensitivity\footnote{See \cite{bartl2021sensitivity} Theorem 2.2 and \cite{christensen2023counterfactual} Page 276 for similar results.} approach computes the right derivative of the bounds at $\delta = 0$, which measures the effect of a small perturbation of the reference distribution on the bounds. We show the right differentiability of the bounds with respect to $\delta$. Define:
\begin{align*}
    & \Pi_{\text{TH}} := \left\{ F \in \mP(\mU) \mid \Pi(\nu_{0},\cdots,\nu_{k}), C_{3} \leq dF \leq C_{4}, \|dF\|_{Lip} \leq L \right\} \\
    & \Pi_{\text{TI}} := \left\{ F \in \mP(\mU) \mid \Pi(\nu_{0},\nu_{T}), C_{3} \leq dF \leq C_{4}, \|dF\|_{Lip} \leq L \right\}
\end{align*}
where $\|\cdot\|_{Lip}$ is the Lipschitz constant, and $C_{3},C_{4},L$ are positive constants. We assume:
\begin{assumption} \label{assumption: hadamard for delta}
    Let $\mA_{I,Lip}^{\delta}$ be either $\mA_{I,\text{TH}}^{\delta,Lip}$ or $\mA_{I,\text{TI}}^{\delta,Lip}$ defined as:
    \begin{align*}
        & \mA_{I,\text{TH}}^{\delta,Lip} := \left\{ \alpha \in \Theta \times \Pi_{\text{TH}} \mid \bE_{F} \left[m(U;\theta,v(\alpha))\right] = P_{0}, D_{KL}(F\|F_{0}) \leq \delta \right\} \\
        & \mA_{I,\text{TI}}^{\delta,Lip} := \left\{ \alpha \in \Theta \times \Pi_{\text{TI}} \mid \bE_{F} \left[m(U;\theta,v(\alpha))\right] = P_{0}, D_{KL}(F\|F_{0}) \leq \delta \right\}
    \end{align*}
    and $\mA_{I,Lip}^{\delta,*}$ be either $\mA_{I,\text{TH}}^{\delta,Lip,*}$ or $\mA_{I,\text{TI}}^{\delta,Lip,*}$ that are the sets of solutions to the optimization problems $\kappa(\delta,P_{0})$ and $\tilde{\kappa}_{\text{TI}}(\delta,P_{0})$ over $\mA_{I,\text{TH}}^{\delta,Lip}$ and $\mA_{I,\text{TI}}^{\delta,Lip}$. Assume:
    \begin{assumptionitems}
        \item $\delta^{*} := \inf\{\delta \geq 0 \mid \mA_{I,Lip}^{\delta} \neq \emptyset\}$ is finite.\footnote{The smallest radius can be computed, see \Cref{remark: delta star}.}
        \item $\mU$ is compact. \label{assumption: compact Support for derivative}
        \item $\|dF_{0}\|_{Lip} \leq L$.
        \item $0 \in \text{int} \{ DP(\alpha)(\Theta \times \Pi -\alpha) \}$ for $\forall \ \alpha \in \mA_{I,Lip}^{\delta,*}$. \label{assumption: constraint qualification for delta}
        \item For $\forall \ \delta_{t} := \delta + t + o(t)$ and $t > 0$ small enough, the optimization problem corresponding to $\delta_{t}$ has an $o(t)$-optimal solution $\alpha(t)$ such that $d(\alpha(t),\mA_{I,Lip}^{\delta,*}) = O(t)$. \label{assumption: convergence of optimizer for delta}
        \item \label{assumption: convergence of optimizer for delta1} For $\forall \ t_{n} \downarrow 0$ the sequence $\{\alpha(t_{n})\}$ has a limit point (in the norm topology) $\alpha_{0} \in \mA_{I,Lip}^{\delta,*}$.
    \end{assumptionitems}
\end{assumption}

\begin{theorem} \label{thm: derivative delta}
    Under Assumptions \ref{assumption: Asymptotic Properties1}, \ref{assumption: continuous differentiability of s and v}, and \ref{assumption: hadamard for delta}, $\kappa(\delta,P_{0})$ and $\tilde{\kappa}_{\text{TI}}(\delta,P_{0})$ are right differentiable at $\delta \geq \delta^{*}$ and their right derivatives are given by:
    \begin{align*}
        & \lim_{\epsilon \downarrow 0} \frac{\kappa(\delta + \epsilon,P_{0}) - \kappa(\delta,P_{0})}{\epsilon} = \inf_{\alpha \in \mA_{I,\text{TH}}^{\delta,Lip,*}} \sup_{\lambda_{KL} \in \Lambda_{KL}(\alpha,\delta)} - \lambda_{KL} \\
        & \lim_{\epsilon \downarrow 0} \frac{\tilde{\kappa}_{\text{TI}}(\delta + \epsilon,P_{0}) - \tilde{\kappa}_{\text{TI}}(\delta,P_{0})}{\epsilon} = \inf_{\alpha \in \mA_{I,\text{TI}}^{\delta,Lip,*}} \sup_{\lambda_{KL} \in \Lambda_{KL}(\alpha,\delta)} - \lambda_{KL}
    \end{align*}
    where $ \Lambda_{KL}(\alpha,\delta)$ is the nonempty set of Lagrange multipliers corresponding to $(\alpha,\delta)$.
\end{theorem}
\Cref{thm: derivative delta} shows the right differentiability. We can also compute the derivative of the length of the bounds. In practice, we may need to compute it numerically due to the optimization over the set of optimizers and the Lagrange multipliers.

\subsection{The Robustness Metric} \label{sec: robust metrics}

\noindent The robustness metric is the smallest deviation from the reference distribution that can lead to sensitive results (\cite{spini2021robustness}). In practice, we begin by estimating a reference scalar parameter, $\hat{s}_{F_{0}}$, under the reference distribution. If the perturbed scalar parameter $s_{F}$ is below a certain threshold, e.g., $\bar{s} = 0.95 \cdot \hat{s}_{F_{0}}$, then we may be concerned about the robustness of the results. The robustness metric is defined as:
\begin{align} \label{eq: smallest delta}
\hspace{-1em}
\begin{aligned}
    \delta(\bar{s},P) := 
    & \inf_{(\theta,v,F) \in \Theta \times \mV \times \Pi(\nu_{1},\cdots,\nu_{k})} D_{KL}(F\|F_{0}) \\
    & \text{s.t.} \quad
    \bE_{F} \left[m(U;\theta,v)\right] = P \\
    & \phantom{\text{s.t.} \quad} \bE_{F} \left[s(U;\theta,v)\right] \leq \bar{s} \\
    & \phantom{\text{s.t.} \quad} \sup_{g \in \mG} \bE_{F} \left[\psi(U;\theta,v,g)\right] = 0
\end{aligned}
\quad
\begin{aligned}
    \tilde{\delta}_{\text{TI}}(\bar{s},P) := 
    & \inf_{(\theta,v,F) \in \Theta \times \mV \times \Pi(\nu_{1},\nu_{T})} D_{KL}(F\|F_{0}) \\
    & \text{s.t.} \quad
    \bE_{F} \left[m(U;\theta,v)\right] = P \\
    & \phantom{\text{s.t.} \quad} \bE_{F} \left[s(U;\theta,v)\right] \leq \bar{s} \\
    & \phantom{\text{s.t.} \quad} \sup_{g \in \mG} \bE_{F} \left[\psi(U;\theta,v,g)\right] = 0
\end{aligned}
\end{align}
where $\bar{s} \in \mathbb{R}$ is a user-specified threshold. The optimization problem searches for a distribution $F$ in the identified set that results in $\bE_{F} \left[s(U;\theta,v)\right] \leq \bar{s}$ and is the closest to the reference distribution $F_{0}$ in terms of KL divergence.

\begin{remark} \label{remark: delta star}
    The $\delta^{*}$ in \Cref{sec: Derivative with respect to delta} can be obtained by removing the constraint for $\bar{s}$ in \eqref{eq: smallest delta}. That is, we seek the smallest radius $\delta^{*}$ such that the identified set is nonempty. See \cite{schennach2014entropic} Page 356 and \cite{christensen2023counterfactual} Section 3.3 for similar definitions.
\end{remark}

We can plot the bounds against $\delta$ and then find the radius corresponding to $\bar{s}$. Alternatively, we can compute it directly by solving \eqref{eq: smallest delta} whose duality results are given by:
\begin{theorem} \label{thm: Strong Duality smallest delta}
    Let $c(U;\theta,v,g,\lambda,\lambda_{s}) := \lambda^{T}m(U;\theta,v) + \lambda_{s}s(U;\theta,v) + \psi(U;\theta,v,g)$ where $\lambda \in \bR^{d_{P}}$. Under \Cref{assumption: Strong Duality}, the following holds:
    \begin{enumerate}[label=(\roman*)]
        \item (Minimax Duality)
        \begin{equation*}
            \delta(\bar{s},P) = \inf_{(\theta,v) \in \Theta \times \mV} \sup_{\lambda \in \bR^{d_{P}}, \lambda_{s} \geq 0, g \in \mG} \mC(\theta,v,g,\lambda,\lambda_{s}) - \lambda^{T}P - \lambda_{s} \bar{s}
        \end{equation*}
        where $\mC(\theta,v,g,\lambda,\lambda_{s})$ is the EOT problem with regularization parameter 1:
        \begin{equation*}
            \mC(\theta,v,g,\lambda,\lambda_{s}) := \inf_{\substack{F \in \Pi(\nu_{1},\cdots,\nu_{k})}} \bE_{F} \left[c(U;\theta,v,g,\lambda,\lambda_{s})\right] + D_{KL}(F\|F_{0})
        \end{equation*}
        \item (Entropic Optimal Transport Duality) We have:
            \begin{equation*}
                \mC(\theta,v,g,\lambda,\lambda_{s}) = \sup_{\left\{\phi_{i} \in L^{1}(\nu_{i})\right\}_{i=1}^{k}} \sum_{i=1}^{k} \bE_{\nu_{i}} \phi_{i}(U_{i}) - \bE_{F_{0}} \exp\left(\sum_{i=1}^{k} \phi_{i}(U_{i}) - c(U;\theta,v,g,\lambda,\lambda_{s})\right) + 1
            \end{equation*}
            Moreover, there are unique maximizers $\left\{\phi_{i}^{*}\right\}_{i=1}^{k}$ up to additive constants $F_{0}$ almost surely, and the unique worst-case distribution $F^{*}$ has the density of the form:
            \begin{equation*}
                \frac{dF^{*}(U)}{dF_{0}(U)} = \exp(\sum_{i=1}^{k} \phi_{i}^{*}(U_{i}) - c(U;\theta,v,g,\lambda,\lambda_{s})) \quad F_{0}\text{-a.s.}
            \end{equation*}
    \end{enumerate}
\end{theorem}

\begin{theorem} \label{thm: Strong Duality smallest delta nonstationary}
    Let $c(U;\theta,v,g,\lambda,\lambda_{s}) := \lambda^{T}m(U;\theta,v) + \lambda_{s}s(U;\theta,v) + \psi(U;\theta,v,g)$ where $\lambda \in \bR^{d_{P}}$. Under Assumptions \ref{assumption: finite moment}, \ref{assumption: symmetric and convex}, \ref{assumption: lower semi-continuous}, and \ref{assumption: Strong Duality1}, the following holds:
    \begin{enumerate}[label=(\roman*)]
        \item (Minimax Duality)
        \begin{equation*}
            \delta_{\text{TI}}(\bar{s},P) = \inf_{(\theta,v) \in \Theta \times \mV} \sup_{\lambda \in \bR^{d_{P}}, \lambda_{s} \geq 0, g \in \mG} \mC_{\text{TI}}(\theta,v,g,\lambda,\lambda_{s}) - \lambda^{T}P - \lambda_{s} \bar{s}
        \end{equation*}
        where $\mC_{\text{TI}}(\theta,v,g,\lambda,\lambda_{s})$ is defined as follows:
        \begin{equation}
            \mC_{\text{TI}}(\theta,v,g,\lambda,\lambda_{s}) := \inf_{F \in \Pi(\nu_{1},\nu_{T})} \bE_{F} \left[c(U;\theta,v,g,\lambda,\lambda_{s})\right] + D_{KL}(F\|F_{0}) \label{eq: nonstationary EOT1}
        \end{equation}
        \item The unique worst-case distribution to \eqref{eq: nonstationary EOT1} has the form:
        \begin{equation*}
           \frac{dF^{*}(U)}{dF_{0}(U)} = \exp\left(\phi_{1}^{*}(\xi_{1}) + \phi_{T}^{*}(\xi_{T}) - c(U;\theta,v,g,\lambda,\lambda_{s})\right)
        \end{equation*}
        where $\phi_{1}^{*}(\xi_{1})$ and $\phi_{T}^{*}(\xi_{T})$ are the unique maximizers (up to an additive constant) to:
        \begin{equation*}
            \sup_{\phi_{1} \in L^{1}(\nu_{1}),\phi_{T} \in L^{1}(\nu_{T})}\bE_{\nu_{1}} \phi_{1}(\xi_{1}) + \bE_{\nu_{T}} \phi_{T}(\xi_{T}) - \bE_{R_{1,T}} \exp\left(\phi_{1}(\xi_{1})+\phi_{T}(\xi_{T})\right) + 1
        \end{equation*}
        where the auxiliary reference measure $R_{1,T}$ is defined as:
        \begin{equation*}
            dR_{1,T}(\xi_{1},\xi_{T}) := \int_{\xi_2, \dots, \xi_{T-1}} \exp\left(-c(U;\theta,v,g,\lambda,\lambda_{s})\right) \, dF_0(\xi_1, \dots, \xi_T)
        \end{equation*}
        Furthermore, the solution $F^{*}$ has the Markov property, i.e., $F^{*} \in \Pi_{\text{Markov}}(\nu_{1},\nu_{T})$.
        \item (Equivalence) $\tilde{\delta}_{\text{TI}}(\bar{s},P) = \delta_{\text{TI}}(\bar{s},P)$ where $\delta_{\text{TI}}(\bar{s},P)$ is the optimal value of the optimization problem in \eqref{eq: smallest delta} for the time-inhomogeneous case with the first-order Markov property constraint on $F$.
    \end{enumerate}
\end{theorem}
\Cref{thm: Strong Duality smallest delta nonstationary} provides the dual formulation for computing the smallest radius in the time-inhomogeneous case. The equivalence holds as the regularization parameter is 1.

\section{Practical Implementation} \label{sec: Practical Implementation}

\noindent This section presents the practical implementation of the proposed framework. \Cref{sec: EOT and Sinkhorn Algorithm} reviews the entropic optimal transport problem and the Sinkhorn algorithm. \Cref{sec: Practical Implementation of the Proposed Framework} proposes a computationally feasible algorithm.

\subsection{Entropic Optimal Transport and Sinkhorn Algorithm} \label{sec: EOT and Sinkhorn Algorithm}
\noindent This section reviews the Sinkhorn algorithm for the entropic optimal transport problem.\footnote{See \cite{sinkhorn1967concerning}, \cite{cuturi2013sinkhorn} and \cite{nutz2021introduction}.} Let $(\mU_{i}, \nu_{i})$ for $i=1, \dots, k$ be probability spaces, where $\mU_{i}$ is the support for the random variable $U_{i}$. For a cost function $c: \mU_{1} \times \dots \times \mU_{k} \to \mathbb{R}$, the entropic optimal transport problem\footnote{If $F_{0} \neq F_{\otimes}$, then \Cref{lemma: convert arbitrary reference to product measure} reformulates the problem as the \ref{EOT problem} problem.} with regularization parameter $\lambda_{KL} > 0$ is defined as:
\begin{equation}
    \mathcal{C}_{\lambda_{KL}} := \inf_{F \in \Pi(\nu_{1}, \dots, \nu_{k})} \bE_{F} \left[c(U_{1}, \dots, U_{k})\right] + \lambda_{KL} D_{KL}(F \| F_{\otimes}) \tag{\blue{EOT}} \label{EOT problem}
\end{equation}
whose dual is given by:
\begin{equation}
    \mathcal{C}_{\lambda_{KL}} = \sup_{\left\{\phi_{i} \in L^{1}(\nu_{i})\right\}_{i=1}^{k}} \sum_{i=1}^{k} \bE_{\nu_{i}} \phi_{i}(U_{i}) - \lambda_{KL} \bE_{F_{\otimes}} \exp\left(\frac{\sum_{i=1}^{k} \phi_{i}(U_{i}) - c(U_{1}, \dots, U_{k})}{\lambda_{KL}}\right) + \lambda_{KL} \ \label{Dual to EOT}
\end{equation}
where $\phi_{i}$ is the test function for the marginal distribution constraint $\nu_{i}$. The dual problem is a concave maximization problem over $\{\phi_{i}\}_{i=1}^{k}$. The worst-case distribution is given by:
\begin{equation*}
    \frac{dF^{*}(U)}{dF_{\otimes}(U)} = \exp\left(\frac{\sum_{i=1}^{k} \phi_{i}^{*}(U_{i}) - c(U_{1}, \dots, U_{k})}{\lambda_{KL}}\right)
\end{equation*}
where the optimizers $(\phi_{1}^{*}, \dots, \phi_{k}^{*})$ are known as the optimal EOT potentials (also called Schr\"odinger potentials), which are the solutions to the Schr\"odinger equation (SE):
\begin{align}
    \phi_{1}(U_{1}) &= -\lambda_{KL} \log \left(\mathbb{E}_{F_{\otimes, -1}} \exp\left(\frac{\sum_{i=2}^{k} \phi_{i}(U_{i}) - c(U_{1}, \dots, U_{k})}{\lambda_{KL}}\right) \right) \quad \nu_{1} \text{-a.s.} \tag{SE1} \label{eq:SE1} \\
    &\vdots \notag \\
    \phi_{k}(U_{k}) &= -\lambda_{KL} \log \left(\mathbb{E}_{F_{\otimes, -k}} \exp\left(\frac{\sum_{i=1}^{k-1} \phi_{i}(U_{i}) - c(U_{1}, \dots, U_{k})}{\lambda_{KL}}\right) \right) \quad \nu_{k} \text{-a.s.} \tag{SEk} \label{eq:SEk}
\end{align}
where $F_{\otimes, -i}$ is the product measure of all marginals except for the $i$-th marginal. The Schr\"odinger equations \eqref{eq:SE1} to \eqref{eq:SEk} can be interpreted as the variational first-order conditions for optimality (see \cite{nutz2021introduction} Remark 3.4). Moreover, they also characterize the marginal constraints. To see this, define:
\begin{equation*}
    dF(U) = \exp\left(\frac{\sum_{i=1}^{k} \phi_{i}(U_{i}) - c(U_{1}, \dots, U_{k})}{\lambda_{KL}}\right) dF_{\otimes}(U)
\end{equation*}
The marginal density can be obtained by integrating out the other marginals; therefore:
\begin{equation*}
    \text{(SEi)} \Leftrightarrow \text{the }i\text{-th marginal of } F \text{ is } \nu_{i}
\end{equation*}
The Sinkhorn algorithm can be interpreted as a coordinate ascent scheme for the optimization problem \eqref{Dual to EOT}. It is a computationally fast\footnote{For its convergence rate, see \cite{peyre2019computational}, \cite{carlier2022linear}, and \cite{eckstein2022quantitative}.}, iterative method for solving \eqref{eq:SE1}-\eqref{eq:SEk}.
\begin{myalgo}[Sinkhorn Algorithm] \label{alg:multimarginal_sinkhorn}
    Initialize $\phi_{i}^{(0)} := 0$ for $i=1, \dots, k$. For iteration $t$, sequentially update for $j=1, \dots, k$ by:
    \begin{align*}
        \phi_{j}^{(t+1)}(U_{j}) &:= -\lambda_{KL} \log \int \exp\left( \frac{\sum_{i \neq j} \phi_{i}^{(t)}(U_{i}) - c(U_{1}, \dots, U_{k})}{\lambda_{KL}} \right) dF_{\otimes, -j}.
    \end{align*}
    Stop if $\sup_{j} \|\phi_{j}^{(t+1)} - \phi_{j}^{(t)}\|_{2} < \epsilon$ for a tolerance $\epsilon > 0$.
\end{myalgo}

\subsection{Proposed Algorithm} \label{sec: Practical Implementation of the Proposed Framework}

\noindent For given $(\theta, v, g, \lambda, \lambda_{KL})$, the EOT problem provides the worst-case distribution, $F^{*}$. This allows us to update $v$ by solving the structural constraint with $F^{*}$. We therefore propose an iterative algorithm to solve the minimax problem. The algorithm proceeds by alternating between updating $(\theta, v)$ and the dual (Lagrange multiplier) variables, $(g, \lambda, \lambda_{KL})$. After initializing all parameters, each iteration $t$ involves the following steps:
\begin{myalgo} \label{alg: Iterative Algorithm for Distributional Robustness}
    Initialize $(\theta^{(0)},v^{(0)},g^{(0)},\lambda^{(0)},\lambda_{KL}^{(0)})$. At iteration $t$,
    \begin{enumerate}
        \item \textbf{Update Model Primitives $(\theta, v)$:} For $(\theta^{(t)}, v^{(t)}, g^{(t)}, \lambda^{(t)}, \lambda_{KL}^{(t)})$, update\footnote{If gradient-based methods are used, then we smooth the non-differentiable components (e.g., the indicator function in \Cref{ex: Dynamic Discrete Choice Models}) using a smooth approximation.} the model parameters $(\theta, v)$ by:
        \begin{enumerate}
            \item[(a)] Propose a new candidate $\theta^{(t+1)}$.
            \item[(b)] Solve the EOT problem with $(\theta^{(t+1)}, v^{(t)}, g^{(t)}, \lambda^{(t)}, \lambda_{KL}^{(t)})$ and obtain $F^{*}$.
            \item[(c)] Update $v^{(t+1)}$ by solving the structural constraint with $F^{*}$.
            \item[(d)] Accept/reject the proposed $(\theta^{(t+1)}, v^{(t+1)})$.
        \end{enumerate}
        \item \textbf{Update Dual Variables $(g, \lambda, \lambda_{KL})$:} For $(\theta^{(t+1)}, v^{(t+1)}, g^{(t)}, \lambda^{(t)}, \lambda_{KL}^{(t)})$, update $(g, \lambda, \lambda_{KL})$ following the same procedure as in the previous step.
        \item Iterate until convergence, or a pre-specified number of iterations is reached.
    \end{enumerate}
\end{myalgo}
The computational cost per iteration mainly comes from solving the EOT problem and the structural constraint, which are both computationally fast. However, the number of iterations required for convergence can be much larger, as our optimization problem is a minimax problem that is potentially non-differentiable.

\section{Empirical Application: Infinite Horizon DDC} \label{sec: Application}

\noindent This section applies our framework to an infinite-horizon dynamic demand for new cars in the UK, France, and Germany. Due to the unobserved product characteristics, the indirect utility of purchasing is the latent variable. To estimate the price elasticity and conduct welfare analysis of electric vehicles (EV) subsidy, we require a distributional assumption for the latent variable to solve the Bellman equation. Existing literature often uses an AR(1) process (e.g., \cite{schiraldi2011automobile,gowrisankaran2012dynamics}), which may be misspecified. For example, the indirect utility may exhibit nonlinear dynamics. Therefore, we conduct a sensitivity analysis with respect to this reference distribution.

\subsection{Data}

\noindent We use the trim-level\footnote{In the automobile industry, a trim-level refers to a specific version of a vehicle model that comes with a particular set of features, options, and styling elements.} data from IHS Markit during the period from 2014 to 2023. The monthly level dataset contains sales, list price, and characteristics of car models in the UK, France and Germany, which are treated as three independent markets in our analysis. To construct the final dataset, we first aggregate data from the trim-level to the model-level. Then, we aggregate fuel types into: petrol, diesel, electric, and hybrid. We remove models whose total sales during the data period are less than 20,000.\footnote{In addition, we exclude car-month observations with sales below 150 units in the German market and below 5 units in the French market.} Finally, we adjust list prices by subtracting EV subsidies. The initial market size for January 2014 is calculated by subtracting the number of registered cars from the total population of each country. The market size is then updated each subsequent month by subtracting the total number of cars sold in the preceding period.

Table \ref{tab:merged_stats} presents the summary statistics. The three markets offer around 141–215 products from around 23 to 30 brands per month. In terms of average sales per model, Germany has 81,401 units, closely followed by France with 81,357 units, and the UK with 70,682 units. The average price is around \$33,444 in the UK, \$27,830 in France, and \$36,117 in Germany.

\begin{table}[!htbp]
    \centering
    \footnotesize
    \caption{Summary Statistics by Country (Monthly, 2014-2023)}
    \label{tab:merged_stats}
    \begin{threeparttable}
    \begin{tabular}{l c c c cc cc cc}
        \toprule
        & \textbf{Avg \#} & \textbf{Avg \#} & \textbf{Avg \#} & \textbf{Price (USD)} & \textbf{Horsepower} & \textbf{Weight (kg)} \\
        \cmidrule(lr){2-2} \cmidrule(lr){3-3} \cmidrule(lr){4-4} \cmidrule(lr){5-5} \cmidrule(lr){6-6} \cmidrule(lr){7-7}
        \textbf{Country} & \textbf{Products} & \textbf{Brands} & \textbf{Sales} & \textbf{Mean} & \textbf{Mean} & \textbf{Mean} \\
        \midrule
        UK & 196 & 30 & 70,682 & 33,444 & 140 & 1,836 \\
        France & 141 & 23 & 81,357 & 27,830 & 112 & 1,702 \\
        Germany & 215 & 26 & 81,401 & 36,117 & 151 & 1,952 \\
        \bottomrule
    \end{tabular}
    \textbf{Note:} The price is adjusted for EV subsidies. First two columns are average number of products and brands per month. Average sales is the average number of cars sold per model. Mean price, horsepower, and weight are weighted by total sales.
    \end{threeparttable}
\end{table}

\subsection{The Model}

\noindent The model is infinite horizon. At each month $t$, a consumer $i$ chooses $j \in \mJ_{t} \bigcup \{0\}$ where $\mJ_{t}$ is the set of available cars at $t$, and 0 is the outside option of not purchasing. Each car $j \in \mJ_{t}$ is characterized by a vector of observable characteristics $x_{jt}$, price $p_{jt}$, and an unobserved characteristic $\xi_{jt}$. The period utility of choosing $j$ is given by:
\begin{align*}
    u(j,x_{t},p_{t},\xi_{t},\eps_{it}) = 
    \begin{cases}
        \alpha p_{jt} + x_{jt}^{T} \theta + \xi_{jt} + \eps_{ijt} & \text{if } j \in \mJ_{t} \\
        \eps_{0it} & \text{if } j = 0
    \end{cases}
\end{align*}
where $\eps_{it}$ is a vector of i.i.d. type I extreme value utility shocks, and $x_{t},p_{t},\xi_{t}$ are the vectors of observable characteristics, prices, and unobserved characteristics for all cars in $\mJ_{t}$.

We assume a purchase is a terminating decision, i.e., consumers exit the market after the purchase. The conditional value function of purchasing car $j$ can be written as the sum of the current period utility and the flow utility after purchase:
\begin{equation*}
    v_{j}(x_{t},p_{t},\xi_{t}) = \frac{x_{jt}^{T} \theta + \xi_{jt}}{1-\beta} + \alpha p_{jt}
\end{equation*}
where $\beta = 0.975$ is the discount factor. The inclusive value of purchasing is defined as:
\begin{equation*}
    \omega_{t} = \log \sum_{j \in \mJ_{t}} \exp\left(\frac{x_{jt}^{T} \theta + \xi_{jt}}{1-\beta} + \alpha p_{jt}\right)
\end{equation*}
Following \cite{schiraldi2011automobile} and \cite{gowrisankaran2012dynamics}, we assume:
\begin{assumption}[Inclusive Value Sufficiency\footnote{The IVS assumption has also been used in \cite{hendel2006measuring,melnikov2013demand,osborne2018approximating}.} (IVS)]
    $G(\omega_{t+1} | x_{t},p_{t},\xi_{t})$ can be summarized by $G(\omega_{t+1} | \omega_{t})$ where $G$ is the conditional distribution function.
\end{assumption}
Under the IVS assumption, $\omega_{t}$ is the only state variable, and the value function $V(\omega)$ is the solution to the smoothed Bellman equation:
\begin{equation}
    V(\omega) = \log \left(\exp\left(v_{0}(\omega)\right) + \exp\left(v_{1}(\omega)\right)\right) \label{eq: value function}
\end{equation}
where $v_{0}(\omega) = \beta \bE\left[V(\omega')|\omega\right]$ and $v_{1}(\omega) = \omega$ are the conditional value functions of not purchasing and purchasing, respectively. The market share of car $j$ at time $t$ is given by:
\begin{equation}
    s_{jt}(x_{t},p_{t},\xi_{t}) = \underbrace{\frac{\exp\left(\omega_{t}\right)}{\exp(V(\omega_{t}))}}_{\text{Probability of Purchasing a Car}} \times \underbrace{\frac{\exp\left(v_{j}(x_{t},p_{t},\xi_{t})\right)}{\exp(\omega_{t})}}_{\text{Conditional Probability of Purchasing Car } j} \label{eq: market share}
\end{equation}

\subsection{First-Stage Estimation}

\noindent In each market, the car with the highest total sales is set as the reference product, denoted by $r$.\footnote{The reference cars are Volkswagen Golf (Petrol) in Germany, Peugeot 208 (Petrol) in France, and Ford Fiesta (Petrol) in the UK. The reference car for each country is always available in that country's market.} Taking the log-odds ratio for cars $j$ and $r$ at time $t$ yields:
\begin{equation}
    \log \left(\frac{s_{jt}}{s_{rt}}\right) =  \alpha \Delta p_{jt} + \frac{\Delta x_{jt}^{T}\theta}{1-\beta} + \frac{\Delta \xi_{jt}}{1-\beta} \label{eq: log odds ratio}
\end{equation}
where $\Delta p_{jt} := p_{jt} - p_{rt}$, $\Delta x_{jt} := x_{jt} - x_{rt}$, and $\Delta \xi_{jt} := \xi_{jt} - \xi_{rt}$. The parameters $(\alpha,\theta)$ are identified\footnote{The intercept and reference fixed effects are not identified from \eqref{eq: log odds ratio}; they are absorbed into $\xi_{rt}$.} by the BLP instruments \cite{berry1993automobile}. Moreover, $\Delta \xi_{jt}$ can be recovered by fitting the relative market share $\frac{s_{jt}}{s_{rt}}$, while the unobserved characteristic of the reference car, $\xi_{rt}$, cannot.

The exogenous characteristics, $x_{jt}$, include vehicle log-weight, log-horsepower, brand fixed effects, SUV fixed effect, and fuel type fixed effects.\footnote{The instruments are the exogenous product characteristics, average log-weight and log-horsepower of competitors' products, the proportion of competitors' products, the proportion of hybrid cars squared, and the number of brands. A competitor product is defined as a car whose brand is not that of $r$ or $j$.} Table \ref{tab:blp_results} presents the regression results. Price coefficients are negative and significant in all markets, ranging from -0.158 in France to -0.192 in Germany. Relative to petrol vehicles, EVs are preferred in the UK (0.019) and France (0.004), but not in Germany (0.000). Hybrid vehicles are valued positively across all three markets, with coefficients ranging from 0.029 in Germany to 0.037 in France. In contrast, diesel has a negative coefficient in the UK (-0.015) but positive effects in Germany (0.008) and France (0.006).

\begin{table}[h!]
\centering
\footnotesize
\caption{Instrumental Variable Regression Results}
\label{tab:blp_results}
\begin{threeparttable}
\begin{tabular}{l cc cc cc}
\toprule
& \multicolumn{2}{c}{\textbf{UK}} & \multicolumn{2}{c}{\textbf{Germany}} & \multicolumn{2}{c}{\textbf{France}} \\
\cmidrule(lr){2-3} \cmidrule(lr){4-5} \cmidrule(lr){6-7}
& \textbf{Coef.} & \textbf{Std. Err.} & \textbf{Coef.} & \textbf{Std. Err.} & \textbf{Coef.} & \textbf{Std. Err.} \\
\midrule
\textbf{Price} & $-0.178$ & $(0.043)$ & $-0.192$ & $(0.004)$ & $-0.158$ & $(0.042)$ \\
\textbf{Log horsepower} & $\phantom{-}0.140$ & $(0.038)$ & $\phantom{-}0.165$ & $(0.003)$ & $\phantom{-}0.065$ & $(0.031)$ \\
\textbf{Log weight} & $-0.008$ & $(0.003)$ & $\phantom{-}0.046$ & $(0.003)$ & $\phantom{-}0.002$ & $(0.001)$ \\
\textbf{SUV} & $\phantom{-}0.018$ & $(0.002)$ & $-0.003$ & $(0.001)$ & $\phantom{-}0.013$ & $(0.002)$ \\
\textbf{Diesel} & $-0.015$ & $(0.003)$ & $\phantom{-}0.008$ & $(0.001)$ & $\phantom{-}0.006$ & $(0.004)$ \\
\textbf{Electric} & $\phantom{-}0.019$ & $(0.004)$ & $\phantom{-}0.000$ & $(0.001)$ & $\phantom{-}0.004$ & $(0.002)$ \\
\textbf{Hybrid} & $\phantom{-}0.035$ & $(0.010)$ & $\phantom{-}0.029$ & $(0.001)$ & $\phantom{-}0.037$ & $(0.009)$ \\[0.5em]
\midrule
Adjusted $R^{2}$ & \multicolumn{2}{c}{0.549} & \multicolumn{2}{c}{0.533} & \multicolumn{2}{c}{0.803} \\
\# of Month-Years & \multicolumn{2}{c}{120} & \multicolumn{2}{c}{120} & \multicolumn{2}{c}{120} \\
\# of Obs. & \multicolumn{2}{c}{23,451} & \multicolumn{2}{c}{25,715} & \multicolumn{2}{c}{16,862} \\
\bottomrule
\end{tabular}
\textbf{Note:} An observation is a pair of $(j,t)$ where $j$ is a car model other than the reference car $r$ and $t$ is the time period. Standard errors are in parentheses. Brand fixed effects are not reported.
\end{threeparttable}
\end{table}

\subsection{The Reference Distribution and Scalar Parameters of Interest}

\noindent We first define the reference distribution and then introduce the scalar parameters of interest. After the first stage estimation, we can calculate $\omega_{t}$ up to $\xi_{rt}$ as:
\begin{equation} \label{eq: inclusive value up to xi_rt}
    \omega_{t} = \log \sum_{j \in \mJ_{t}} \exp\left(\frac{x_{jt}^{T} \theta + \Delta \xi_{jt}}{1-\beta} + \alpha p_{jt}\right) + \frac{\xi_{rt}}{1-\beta}
\end{equation}
As we have identified the utility parameters, the potential sensitivity of the empirical results solely arises from the distributional assumption on $\omega_{t}$. 

The reference transition of $\omega_{t}$ to solve the Bellman equation is an AR(1) process:
\begin{equation}
    \omega_{t} = \gamma_{0} + \gamma_{1} \omega_{t-1} + \eta_{t} \label{eq: AR(1) process}
\end{equation}
where $\eta_{t}$ follows an i.i.d. normal distribution with mean 0 and variance $\sigma^{2}$.

The parameters $(\gamma_{0}, \gamma_{1}, \sigma^{2})$ are estimated using an iterative procedure. We begin with an initial guess of $(\gamma_{0}, \gamma_{1}, \sigma^{2})$ and circulate between: (i) solving the Bellman equation \eqref{eq: value function}, (ii) recovering $\{\omega_{t}\}_{t=1}^{T}$ from the market share of purchasing (the first part of \eqref{eq: market share}), (iii) updating $(\gamma_{0}, \gamma_{1}, \sigma^{2})$ by refitting an AR(1) process \eqref{eq: AR(1) process} until we find a fixed point. The reference distribution $F_{0}$ for $(\omega,\omega')$ is the product of the transition kernel of the estimated AR(1) process and its stationary distribution $\nu_{0}$. The perturbation set is defined as:
\begin{equation*}
    \mF := \left\{ F \in \mP(\mU) \mid F \in \Pi(\nu_{0},\nu_{0}), D_{KL}(F\|F_{0}) \leq \delta \right\}
\end{equation*}

We consider two scalar parameters: (i) the industrywide price elasticity of demand, (ii) the welfare analysis of an additional EV subsidy. For both cases, the transition of $\omega_{t}$ is unchanged, i.e., we assume consumers' beliefs about the transition of $\omega_{t}$ stay the same.

For industrywide price elasticity at period $t_{1}$, we consider a 1\% increase in the price of all cars. The future $\omega_{t_{1}+1}$ is conditional on:
\begin{equation*}
    \omega'_{t_{1}} = \log \sum_{j \in \mJ_{t_{1}}} \exp\left(\frac{x_{jt_{1}}^{T} \theta + \Delta \xi_{jt_{1}}}{1-\beta} + 1.01 \cdot \alpha p_{jt_{1}}\right) + \frac{\xi_{rt_{1}}}{1-\beta}
\end{equation*}
and the industrywide price elasticity at time $t_{1}$ is:
\begin{equation*}
    \frac{s_{0t_{1}} - s_{0}(\omega'_{t_{1}})}{1-s_{0t_{1}}} \times 100
\end{equation*}
where $s_{0}(\omega'_{t_{1}})$ is the model-implied market share of not purchasing.

For EV subsidy at $t_{1}$, we consider an additional 3,000 USD subsidy. Denote by $\mJ_{t_{1},\text{EV}}$ the set of EVs at $t_{1}$. The future $\omega_{t_{1}+1}$ is conditional on:
\begin{equation*}
    \omega^{\text{EV}}_{t_{1}} = \log \sum_{j \in \mJ_{t_{1}}} \exp\left(\frac{x_{jt_{1}}^{T} \theta + \Delta \xi_{jt_{1}}}{1-\beta} + \alpha \left(p_{jt_{1}} - \mathbbm{1}(j \text{ is an EV}) \cdot 3000\right)\right) + \frac{\xi_{rt_{1}}}{1-\beta}
\end{equation*}
and the consumer surplus from the subsidy is given by:
\begin{equation*}
    \text{Consumer Surplus} = \frac{V(\omega^{\text{EV}}_{t_{1}}) - V(\omega_{t_{1}})}{-\alpha} \times M_{t_{1}}
\end{equation*}
where $M_{t_{1}}$ is the market size at time $t_{1}$.\footnote{The cost of the subsidy is $\text{Cost} = 3000 \cdot M_{t_{1}} \cdot \sum_{j \in \mJ_{t_{1},\text{EV}}} \frac{\exp\left(v_{j}(x_{jt_{1}},p_{jt_{1}}-3000,\xi_{jt_{1}})\right)}{\exp(V(\omega^{\text{EV}}_{t_{1}}))}$.
}

\subsection{Sensitivity Analysis}

\noindent Our framework requires the constraints to be linear in $F$, while the Bellman equation \eqref{eq: value function} is not. We first reformulate it to fit into our framework. By the \textit{Hotz-Miller Inversion Lemma} (\cite{hotz1993conditional}), we have:
\begin{equation}
    V(\omega) = \omega - \log s_{1}(\omega) \label{eq: value function inversion}
\end{equation}
Taking the log-odds ratio of purchasing and not purchasing, and using \eqref{eq: value function inversion}, we have:

\begin{equation}
    \log \left(\frac{1-s_{0}(\omega)}{s_{0}(\omega)}\right) = \omega - \beta \bE \left[\omega' - \log (1-s_{0}(\omega')) |\omega\right] \label{eq: structural constraint}
\end{equation}
The above constraint is the fixed point problem on the market share space (see \cite{aguirregabiria2002swapping}). The right-hand side is linear in the conditional distribution. The following lemma establishes the relationship between the fixed point problems in \eqref{eq: value function} and \eqref{eq: structural constraint}.
\begin{assumption} \label{assumption: fixed point relationship}
    Assume the inclusive value $\omega$ has compact support $\Omega \subset \mathbb{R}$ with nonempty interior equipped with the sup-norm and $\bE\left[f(\omega')|\omega\right] \in C(\Omega)$ for any $f \in C(\Omega)$ where $C(\Omega)$ is the space of continuous functions on $\Omega$.
\end{assumption}
\begin{lemma} \label{lemma: fixed point relationship}
    The following holds:
    \begin{enumerate}[label=(\roman*)]
        \item Under \Cref{assumption: fixed point relationship}, the fixed point problem \eqref{eq: value function} has a unique solution on $C(\Omega)$.
        \item The fixed point problem \eqref{eq: structural constraint} has a unique solution if and only if the fixed point problem \eqref{eq: value function} has a unique solution.
        \item If \eqref{eq: value function} and \eqref{eq: structural constraint} both have unique solutions, then it holds that $1-s_{0}(\omega) = \exp\left(\omega-V(\omega)\right)$ where $s_{0}(\omega)$ and $V(\omega)$ are the solutions to \eqref{eq: value function} and \eqref{eq: structural constraint}, respectively.
    \end{enumerate}
\end{lemma}
\Cref{lemma: fixed point relationship} shows that solving the Bellman equation \eqref{eq: value function} is equivalent to solving \eqref{eq: structural constraint}. We further convert \eqref{eq: structural constraint} into an unconditional moment constraint by assuming that $s_{0}(\omega)$ is the solution to \eqref{eq: structural constraint} if and only if:
\begin{equation}
    \sup_{g \in C(\Omega)} \bE_{F} \left[g(\omega)\left(\log(\frac{1 - s_{0}(\omega)}{s_{0}(\omega)}) - \omega + \beta \omega' - \beta \log(1-s_{0}(\omega'))\right)\right] = 0 \label{eq: unconditional moment constraint}
\end{equation}
\begin{assumption} \label{assumption: invertible}
    For $\forall \ F \in \mF$, the solution $s_{0}(\omega)$ corresponding to \eqref{eq: structural constraint} satisfies the following: for all $t=1,\ldots,T$, there exists a unique $\omega_{t} \in \Omega$ such that $s_{0}(\omega_{t}) = s_{0t}$.
\end{assumption}
\Cref{assumption: invertible} allows us to profile out $\{\omega_{t}\}_{t=1}^{T}$, which is useful for the implementation. A sufficient condition is that $s_{0}(\omega)$ is continuous and strictly decreasing, and its smallest and largest values are small and large enough. Under \Cref{assumption: fixed point relationship}, we have $s_{0}(\omega) \in C(\Omega)$. Moreover, we can expect that a higher inclusive value $\omega$ corresponds to a lower market share of not purchasing, i.e., $s_{0}(\omega)$ is decreasing in $\omega$. Therefore, \Cref{assumption: invertible} is mild.

The last condition is the fixed point constraint similar to the procedure to estimate the AR(1) process. Suppose the distribution $F$ is used in \eqref{eq: unconditional moment constraint}, and $\{\omega_{t}\}_{t=1}^{T}$ is the sequence of recovered inclusive values. Denote by $\hat{F}$ the estimator of the joint distribution for the pairs $\{(\omega_{t},\omega_{t+1})\}_{t=1}^{T-1}$. Then, our fixed point constraint is:
\begin{equation*}
    D_{KL}(F\|\hat{F}) \leq \epsilon_{T}
\end{equation*}
where $\epsilon_{T}$ is the tolerance level. To choose $\epsilon_{T}$, we estimate the joint distribution of inclusive values recovered from the AR(1) process by the kernel density estimator with Gaussian kernel and bandwidth selected by the 5-fold cross-validation. Then, we set $\epsilon_{T}$ to be the KL divergence between the kernel density estimator and the reference distribution.

For EV subsidy, by \eqref{eq: value function inversion}, the consumer surplus (CS) is given by:
\begin{equation*}
    V(\omega^{\text{EV}}_{t_{1}}) - V(\omega_{t_{1}}) = \omega^{\text{EV}}_{t_{1}} - \omega_{t_{1}} + \log (1 + \frac{s_{1}(\omega^{\text{EV}}_{t_{1}}) - s_{1t_{1}}}{s_{1t_{1}}})
\end{equation*}
where $\omega^{\text{EV}}_{t_{1}} - \omega_{t_{1}}$ does not depend on $F$. Therefore, to bound CS, it is equivalent to bounding the change in the market share of purchase.

Putting everything together, the lower bound on the elasticity at $t_{1}$ is given by:
\begin{align*}
    \inf_{s_{0}(\omega) \in C(\Omega)} \inf_{F \in \mF} & \quad \frac{s_{0t_{1}} - s_{0}(\omega'_{t_{1}})}{1-s_{0t_{1}}} \times 100 \\
    \text{s.t.} \quad
    & s_{0}(\omega_{t}) = s_{0t} \text{ for } t=1,\ldots,T  \\
    & \sup_{g \in C(\Omega)} \bE_{F} \left[g(\omega)\left(\log(\frac{1 - s_{0}(\omega)}{s_{0}(\omega)}) - \omega + \beta \omega' - \beta \log(1-s_{0}(\omega'))\right)\right] = 0 \\
    & D_{KL}(F\|\hat{F}) \leq \epsilon_{T}
\end{align*}
where $\omega'_{t_{1}}$ is replaced by $\omega^{\text{EV}}_{t_{1}}$ for the EV subsidy case. For $\delta = 0$, the reference distribution is the unique solution to the above problem. The corresponding elasticity is the reference industrywide elasticity.

\subsection{Implementation} \label{sec: Implementation infinite horizon DDC}

\noindent We adapt \Cref{alg: Iterative Algorithm for Distributional Robustness} proposed in \Cref{sec: Practical Implementation of the Proposed Framework}. We will use numerical integration to discretize the support of the AR(1) process. Therefore, the recovered inclusive values $\{\omega_{t}\}_{t=1}^{T}$ are not differentiable with respect to the discretized market share function $s_{0}(\omega)$. To address this issue, we employ the MCMC optimization method. We alternate between solving the EOT problem to obtain the worst-case distribution, solving the Bellman equation to update $s_{0}(\omega)$, recovering the inclusive values, and checking the fixed point constraint. To derive a tractable dual formulation, we handle the fixed point constraint in a specific way. Because this constraint depends on an estimator, $\hat{F}$, that changes during optimization—potentially causing numerical instability—we first derive the dual formulation without it.\footnote{In principle, we can choose other constraints like the integral probability metrics and adapt the minimax theorem in \Cref{thm: Strong Duality}. However, it can lead to more optimization parameters.} Then, the fixed point constraint determines the acceptance/rejection of the candidate parameters in the Metropolis-Hastings step. Consider the following optimization problem:
\begin{align*}
    \inf_{s_{0}(\omega) \in C(\Omega)} \inf_{F \in \mF} & \quad \frac{s_{0t_{1}} - s_{0}(\omega'_{t_{1}})}{1-s_{0t_{1}}} \times 100 \\
    \text{s.t.} \quad
    & s_{0}(\omega_{t}) = s_{0t} \text{ for } t=1,\ldots,T \\
    & \sup_{g \in C(\Omega)} \bE_{F} \left[g(\omega)\left(\log(\frac{1 - s_{0}(\omega)}{s_{0}(\omega)}) - \omega + \beta \omega' - \beta \log(1-s_{0}(\omega'))\right)\right] = 0
\end{align*}
Applying \Cref{thm: Strong Duality}, its dual is:
\begin{align*}
    \inf_{s_{0}(\omega)  \in C(\Omega)} \sup_{g \in C(\Omega), \lambda_{KL} \geq 0}
    & \frac{s_{0t_{1}} - s_{0}(\omega'_{t_{1}})}{1-s_{0t_{1}}} \times 100 + \mC(s_{0},g,\lambda_{KL}) - \lambda_{KL} \delta \\
    \text{s.t.} \quad
    & s_{0}(\omega_{t}) = s_{0t} \text{ for } t=1,\ldots,T
\end{align*}
where $\mC(s_{0},g,\lambda_{KL})$ is the EOT problem: $\mC(s_{0},g,\lambda_{KL}) := \sup_{F \in \Pi(\nu_{0},\nu_{0})} \bE_{F} \left[c(\omega,\omega';s_{0},g)\right] + \lambda_{KL} D_{KL}(F\|F_{0})$ whose cost function is $c(\omega,\omega';s_{0},g) = g(\omega)\left(\log(\frac{1 - s_{0}(\omega)}{s_{0}(\omega)}) - \omega + \beta \omega' - \beta \log(1-s_{0}(\omega'))\right)$ and the worst-case conditional distribution $F^{*}$ is given by:
\begin{equation*}
    dF^{*}(\omega'|\omega) = \exp(\frac{\phi_{1}^{*}(\omega) + \phi_{2}^{*}(\omega') - c(\omega,\omega';s_{0},g)}{\lambda_{KL}}) dF_{0}(\omega'|\omega) \quad F_{0}\text{-a.s.}
\end{equation*}
where $\phi_{1}^{*}(\omega)$ and $\phi_{2}^{*}(\omega')$ are the optimal EOT potentials. During the optimization process, $dF^{*}(\omega'|\omega)$ is used to update $s_{0}(\omega)$ by solving \eqref{eq: structural constraint} using fixed point iteration.

As shown in \Cref{thm: Strong Duality}, the expectation in the dual is taken with respect to the reference distribution. Therefore, we discretize the estimated AR(1) process, which results in three approximation errors: (i) the Bellman equation is solved on the discretized support, (ii) $\{\omega_{t}\}_{t=1}^{T}$ are recovered approximately, and (iii) the elasticity is computed approximately. That is, we approximate $s_{0}(\omega)$ by the market share of the nearest grid point to $\omega$. Therefore, there is a trade-off between approximation error and computational cost. A finer discretization reduces the approximation error, while increasing the number of optimization parameters.

However, our dual formulation significantly improves computational efficiency. Suppose we discretize the AR(1) process into $N$ grid points. If we directly solve \eqref{eq: unconditional moment constraint}, the number of optimization parameters is $O(N^{2})$ due to the transition matrix. In contrast, the number of optimization parameters is $O(N)$ in the dual formulation.

Algorithm \ref{alg:Initial MCMC Optimization} summarizes the simulated annealing MCMC optimization algorithm (\cite{kirkpatrick1983optimization}). It starts with the reference market share $s^{(0)}$, and alternate between proposing new parameters $(g', \lambda'_{KL})$, solving the EOT problem, solving the Bellman equation using the worst-case distribution, and accepting or rejecting the proposed parameters using the Metropolis-Hastings step based on the change in the elasticity penalized by the violation of the market share and fixed-point constraints. At each improvement step, it pools previous results across all radii for initialization. We choose 51 grid points, 5,000 MCMC steps, 5 optimization steps, 14 radii (the last is $10^{10}$), and 100 as the simulated annealing multiplier.

\setcounter{algocf}{2}

\begin{algorithm}[h!]
    \caption{Simulated Annealing MCMC Optimization Algorithm}
    \label{alg:Initial MCMC Optimization}
    \footnotesize

    \SetKwInput{Parameters}{Parameters}
    \SetKwComment{Comment}{\tcc*[r]}{ }

    \Parameters{
                $N$: Number of grid points;
                $T$: MCMC steps per optimization run;
                $J$: Optimization steps;
                $m$: Simulated annealing multiplier;
                $l$: Number of radii;
                }
    \DontPrintSemicolon

    \For{$j=1$ \KwTo $J$ \tcc*[r]{Optimization Step $j$}}{
    \For{$i=0$ \KwTo $l-1$, set $\delta_{i} = 10^{-3 + i \cdot 0.25}$ \;}{
        \uIf{$i=0$}{
            \textbf{If $j=0$:} Set $s^{(0)}$ as the reference market share. Initialize $g^{(0)}$, $\lambda_{KL}^{(0)}$.

            \textbf{If $j>0$:} Set $(s^{(0)}, g^{(0)}, \lambda_{KL}^{(0)})$ to the optimal solution from the previous step's stored results with upper bound $10^{-3}$ on the KL divergence to the reference distribution.\;
        }
        \Else{
            \textbf{If $j=0$:} Set $(g^{(0)}, \lambda_{KL}^{(0)}, s^{(0)})$ to the optimal solution from the previous step. \\

            \textbf{If $j>0$:} Set $(s^{(0)}, g^{(0)}, \lambda_{KL}^{(0)})$ to the optimal solution from the previous stored results with upper bound $10^{-3 + i \cdot 0.25}$ on the KL divergence to the reference distribution.
        }

        \For{$t=1, \dots, T$ \tcc*[r]{Simulated Annealing MCMC optimization}}{
            \tcp{1. Propose New Parameters}
            Propose $(g', \lambda'_{KL})$ from the random walk, solve the EOT problem $\mC(s_{0}^{t},g',\lambda'_{KL})$ and obtain $F^{*}$, solve the Bellman equation \eqref{eq: structural constraint} with $F^{*}$, recover $\{\omega_{t}\}_{t=1}^{T}$, estimate the distribution of $\{(\omega_{t},\omega_{t+1})\}_{t=1}^{T}$ by kernel density estimator, and compute the elasticity.\;
            \tcp{2. Check Constraints and Apply Penalty}
            Calculate the sum of violations from the market share and fixed-point constraints. The market share violation is defined as $\max \{0, \text{vio}_{F} - \text{vio}_{ref}\}$ where $\text{vio}_{ref}$ is the violation of the reference model. If the total violation exceeds 0.005, add a large penalty (100). If $D_{KL}(F^{*}\|F_{0}) > \delta_{i}$, add a large penalty (100).\;

            \tcp{3. Accept/Reject (Metropolis-Hastings)}
            Apply a Metropolis-Hastings step based on the change in the (penalized) elasticity multiplied by $\frac{10*(1 + (s - 1) * (m - 1))}{(T-1)}$ where the prior is $\mathcal{N}(0,100)$.\;

            \tcp{4. Adapt (\cite{andrieu2008tutorial} Algorithm 4)}
            Update the random walk via vanishing adaptation scheme.\;
        }
    }
    }
\end{algorithm}

\subsection{Results}

\noindent We estimate two alternative transition densities for each market. The first assumes that the inclusive values are i.i.d. normally distributed. The second estimates a nonlinear AR(1) process using a cubic spline\footnote{The nonlinear AR(1) process is specified as:
$g(\omega) = \sum_{k=1}^{N+3} \rho_k \Phi\left(\frac{\omega - a}{h} - (k-2)\right)$ where $a$ is the minimum of the discretized support of $\omega$, $h$ is the distance between two adjacent grid points, $\rho_{k}$ are parameters to be estimated, and
$
\Phi(t) = \begin{cases} 
4 - 6t^2 + 3|t|^3 & \text{if } |t| \le 1 \\
(2-|t|)^3 & \text{if } 1 < |t| \le 2 \\
0 & \text{otherwise}
\end{cases}
$. For this model, we set $N=4$ to avoid overfitting. Then, we discretize the reference AR(1) process into 4 grid points, and compute the KL divergence between the nonlinear AR(1) process and the reference AR(1) process.}. In the following figures, we plot the KL divergence between the alternative models and the reference model. The independent model is closer to the reference model with KL divergence between 0.04 to 0.67, while the nonlinear AR(1) process is farther away, with KL divergence between 3.94 to 10.09.

Figures \ref{fig:elasticity_CS_bounds_UK}-\ref{fig:elasticity_CS_bounds_Germany} plot the bounds on the industrywide elasticities\footnote{\cite{schiraldi2011automobile} finds a average long-run price elasticities ranging from -3.54 to -4.34 across different car segments for the Italian market. \cite{d2019automobile} reports average elasticities between -3.94 and -6.40 across consumer groups for the French market. \cite{reynaert2021benefits} finds a mean own-price elasticity of -5.45 for the European market. \cite{grieco2024evolution} estimates an average elasticity of -5.36 for the U.S. market in 2015. \cite{Remmy2025} reports a mean price elasticity of -4.043 for the German market.} for the UK, Germany, and France in December 2023. The French market is the least elastic (reference elasticity: -4.048), while the Germany market the most elastic (reference elasticity: -6.073). The UK market's reference elasticity is -5.336. Based on our three sensitivity measures, we define the local (global) sensitivity as the ratio of the local (global) interval length to the reference value. For local sensitivity, we set $\delta=0.001$, while for global sensitivity, we set $\delta=10^{10}$. The robustness metric is defined as the smallest deviation from the reference distribution such that the elasticity can deviate by, for example, 2.5\% from the reference elasticity. 

The French market is the least sensitive in terms of local and global sensitivity, with 1.16\% local deviation and 6.20\% global deviation from the reference elasticity. The UK market is less sensitive locally (1.66\%) than the German market (3.52\%). They are both more sensitive globally (15.16\% for the UK vs. 15.24\% for Germany) than the French market. The bounds of UK market flatten around 0.178, while the French and German market flatten around 0.056. For the robustness metric, we consider 2.5\% deviation from the reference elasticity. The UK market's robustness metric is around 0.018 for the upper bound and 0.008 for the lower bound. The French market's robustness metric is around 0.025 for the upper bound and 0.018 for the lower bound. The German market's robustness metric is around 0.002 for the lower bound and 0.003 for the upper bound. Therefore, in terms of robustness metric, the French market is also the most robust, while the German market is the least robust.

Figures \ref{fig:elasticity_CS_bounds_UK}-\ref{fig:elasticity_CS_bounds_Germany} also plot the bounds on the consumer surplus from an additional \$3,000 EV subsidy. The subsidy is implemented between July and December 2023, when the reference consumer surplus is maximized. They are November for the UK (reference CS: \$2,880 million), October for France (reference CS: \$1,432 million), and September for Germany (reference CS: \$856 million). Overall, the EV subsidy is beneficial as the lower bounds are \$2,584 million for the UK, \$1,243 million for France, and \$309 million for Germany. The corresponding costs for subsidy are around \$12 million for the UK, \$23 million for France, and \$41 million for Germany. The costs are insensitive to the misspecification, as they only depend on the absolute change in the market share of purchases, and the conditional market share of EVs, instead of percentage change used for consumer surplus.

In terms of local sensitivity, the UK market is the least sensitive (2.52 \% local deviation), the French market exhibits similar local sensitivity (4.72\% local deviation), while the German market is the most sensitive (24.79\% local deviation). In terms of global sensitivity, the UK market and French markets share similar sensitivity (25.17 \% global deviation for the UK, and 24.73 \% for France), while the German market is the most sensitive (102.75 \% global deviation). For the robustness metric, if we consider 10\% deviation from the reference CS, the UK and French markets' robustness metric are more than 0.018, while the German market's robustness metric is around 0.001. Therefore, the German market is also the least robust in terms of robustness metric.

Figure \ref{fig:elasticity_bounds_time_series} plots the time series of bounds on the industrywide elasticities. We set $\delta=0.001$ for the local deviation and $\delta = 1$ for the global deviation, as the bounds flatten at a maximum of 0.178 in December 2023. Large points in the figure indicate that the KL divergence constraint is binding—specifically, when the KL divergence between the worst-case distribution and the reference distribution exceeds $0.95 \cdot \delta$. When two consecutive points align horizontally, this indicates that increasing the radius does not affect the bounds, as exemplified by the UK market in November 2023. In terms of both local and global sensitivity, the UK and French markets are less sensitive than the German market. All three markets exhibit some sensitive periods. For the UK market, the upper bound's global deviation in February 2022 is around 30\% away from the reference elasticity. For the French market, the upper bounds' global deviations in April 2021, and July 2022 are around 50\% away from the reference elasticity. For the German market, the lower bound's global deviation in March 2022 and April 2023 is around 50\% away from the reference elasticity. In terms of local sensitivity, the lower bounds' of German market in April 2021 is around 50\% away from the reference elasticity.

\begin{figure}[H]
    \centering
    \begin{minipage}{0.48\textwidth}
        \centering
        \includegraphics[width=\linewidth]{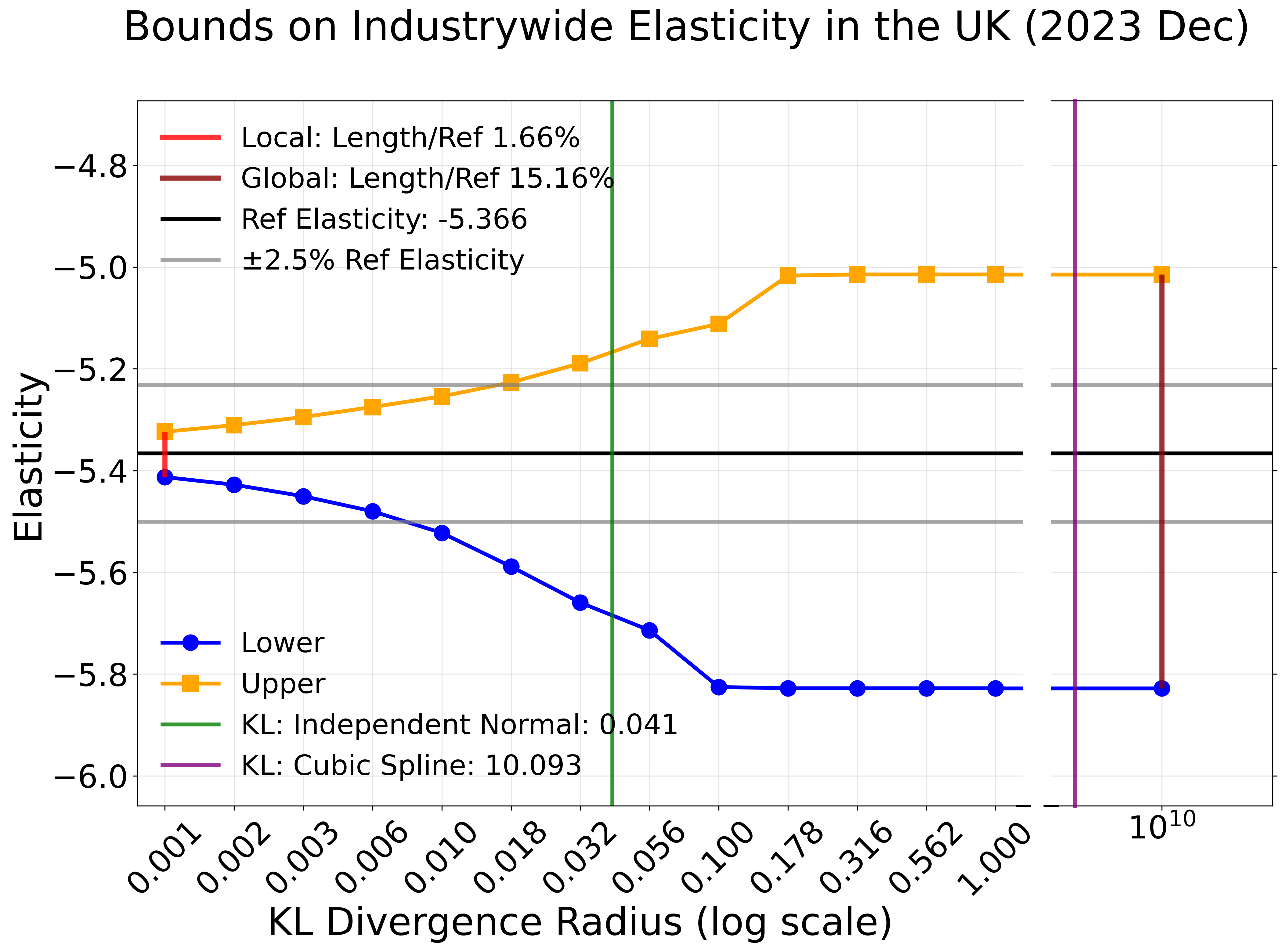}
    \end{minipage}
    \hfill
    \begin{minipage}{0.48\textwidth}
        \centering
        \includegraphics[width=\linewidth]{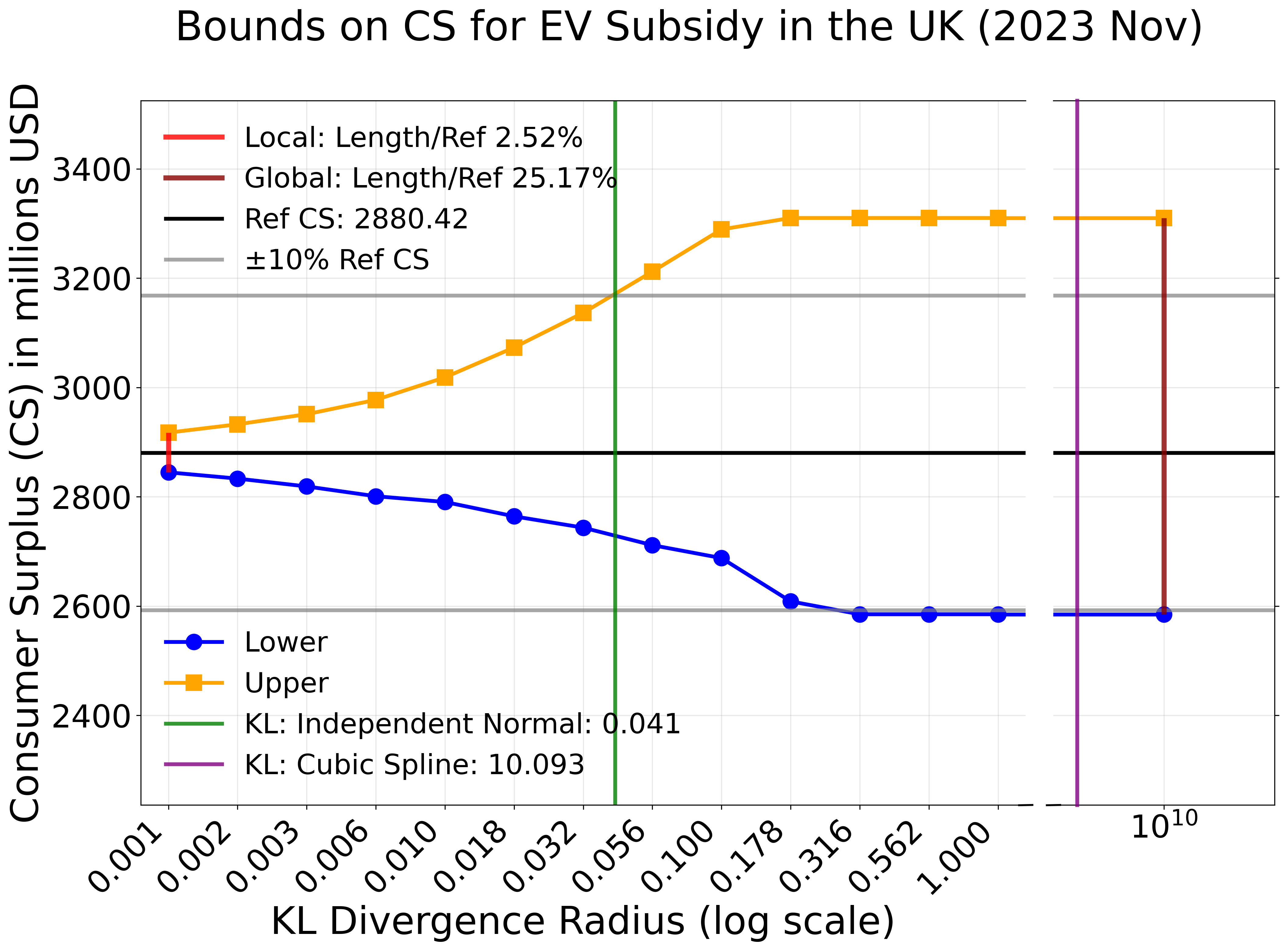}
    \end{minipage}
    \caption{Bounds on Industrywide Elasticity and Consumer Surplus for the UK}
    \label{fig:elasticity_CS_bounds_UK}
\end{figure}

\begin{figure}[H]
    \centering
    \begin{minipage}{0.48\textwidth}
        \centering
        \includegraphics[width=\linewidth]{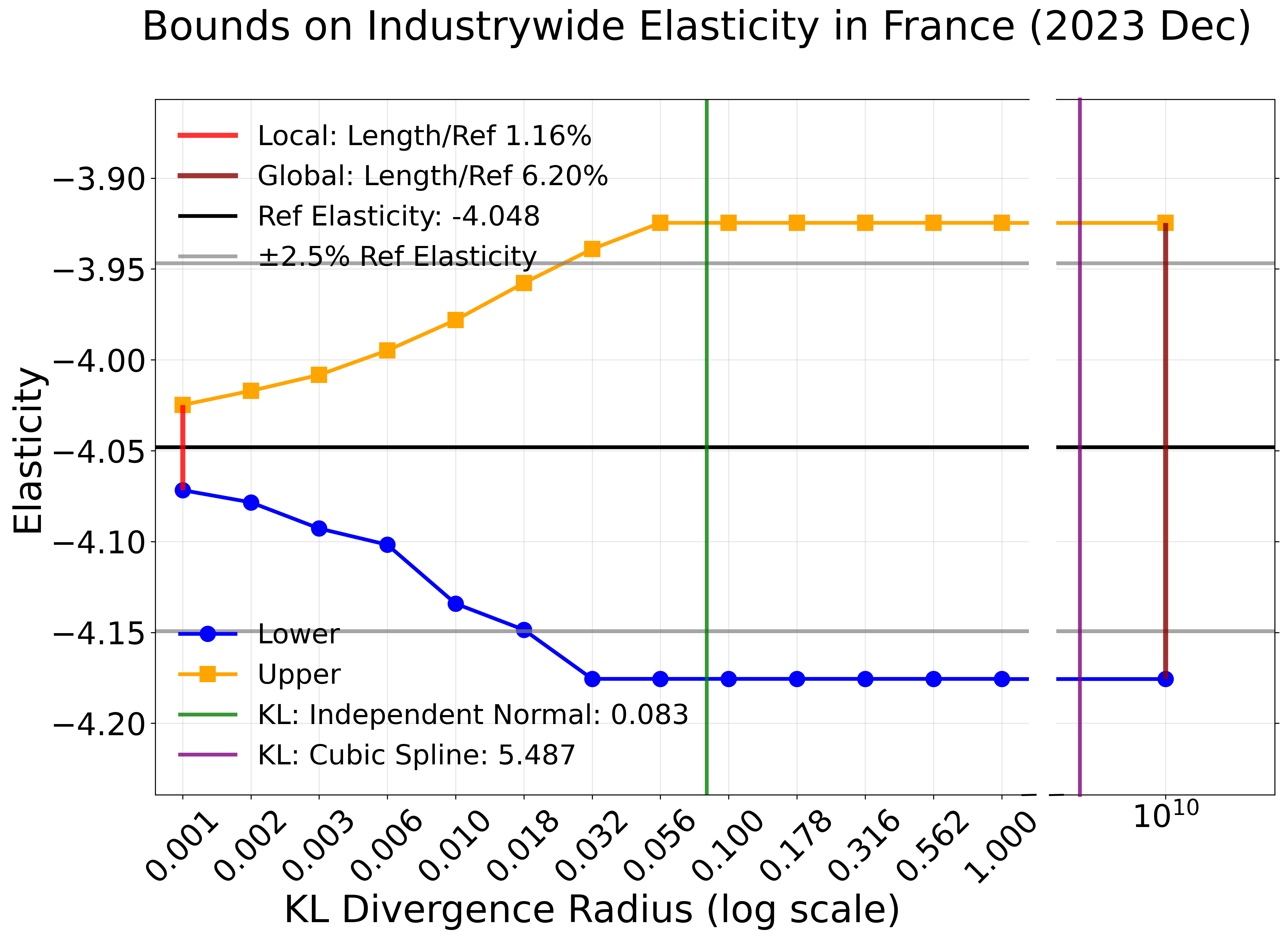}
    \end{minipage}
    \hfill
    \begin{minipage}{0.48\textwidth}
        \centering
        \includegraphics[width=\linewidth]{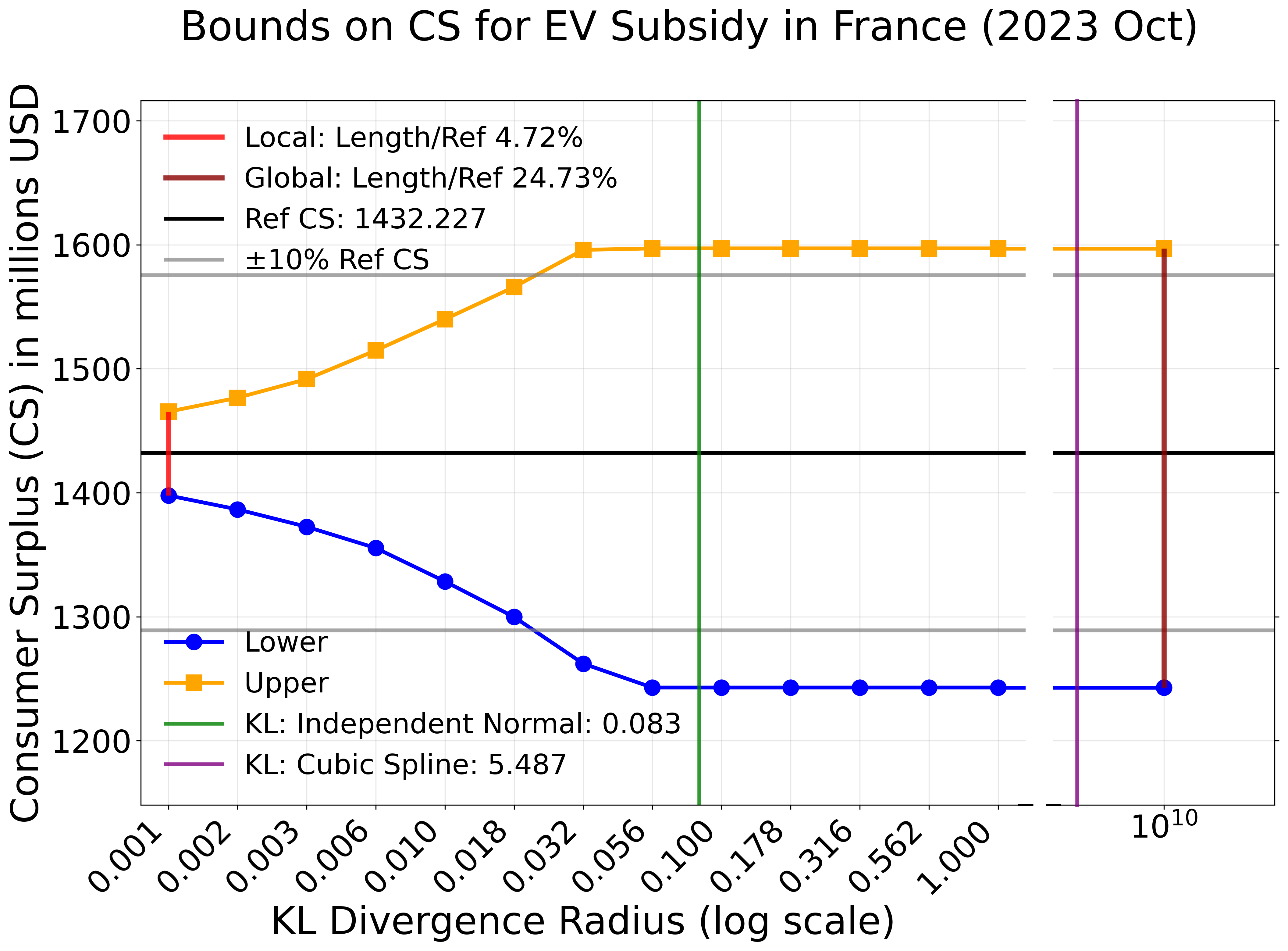}
    \end{minipage}
    \caption{Bounds on Industrywide Elasticity and Consumer Surplus for France}
    \label{fig:elasticity_CS_bounds_France}
\end{figure}

\begin{figure}[H]
    \centering
    \begin{minipage}{0.48\textwidth}
        \centering
        \includegraphics[width=\linewidth]{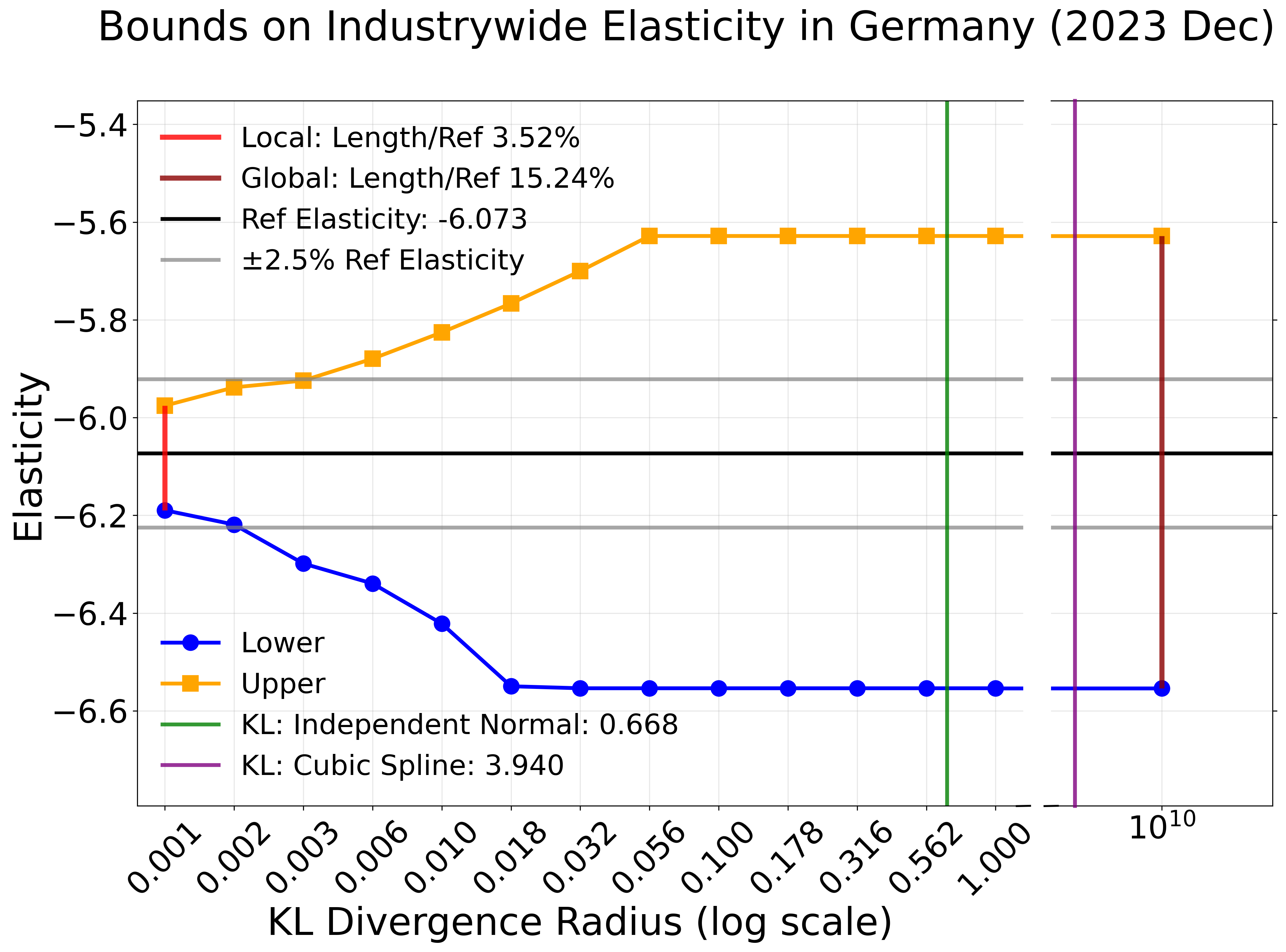}
    \end{minipage}
    \hfill
    \begin{minipage}{0.48\textwidth}
        \centering
        \includegraphics[width=\linewidth]{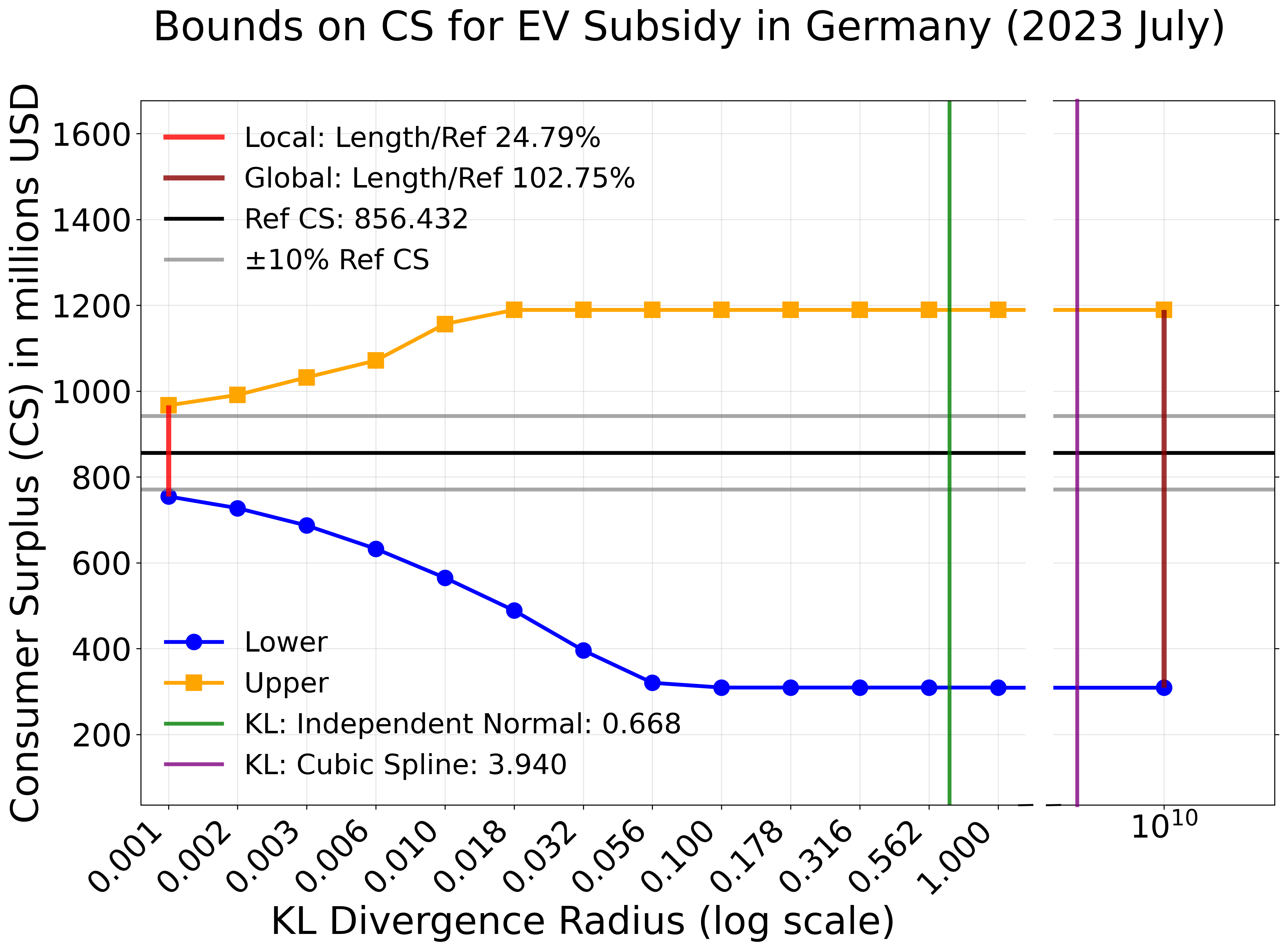}
    \end{minipage}
    \caption{Bounds on Industrywide Elasticity and Consumer Surplus for Germany}
    \label{fig:elasticity_CS_bounds_Germany}
\end{figure}

\begin{figure}[H]
    \centering
    \begin{minipage}{0.8\textwidth}
        \centering
        \includegraphics[width=\linewidth]{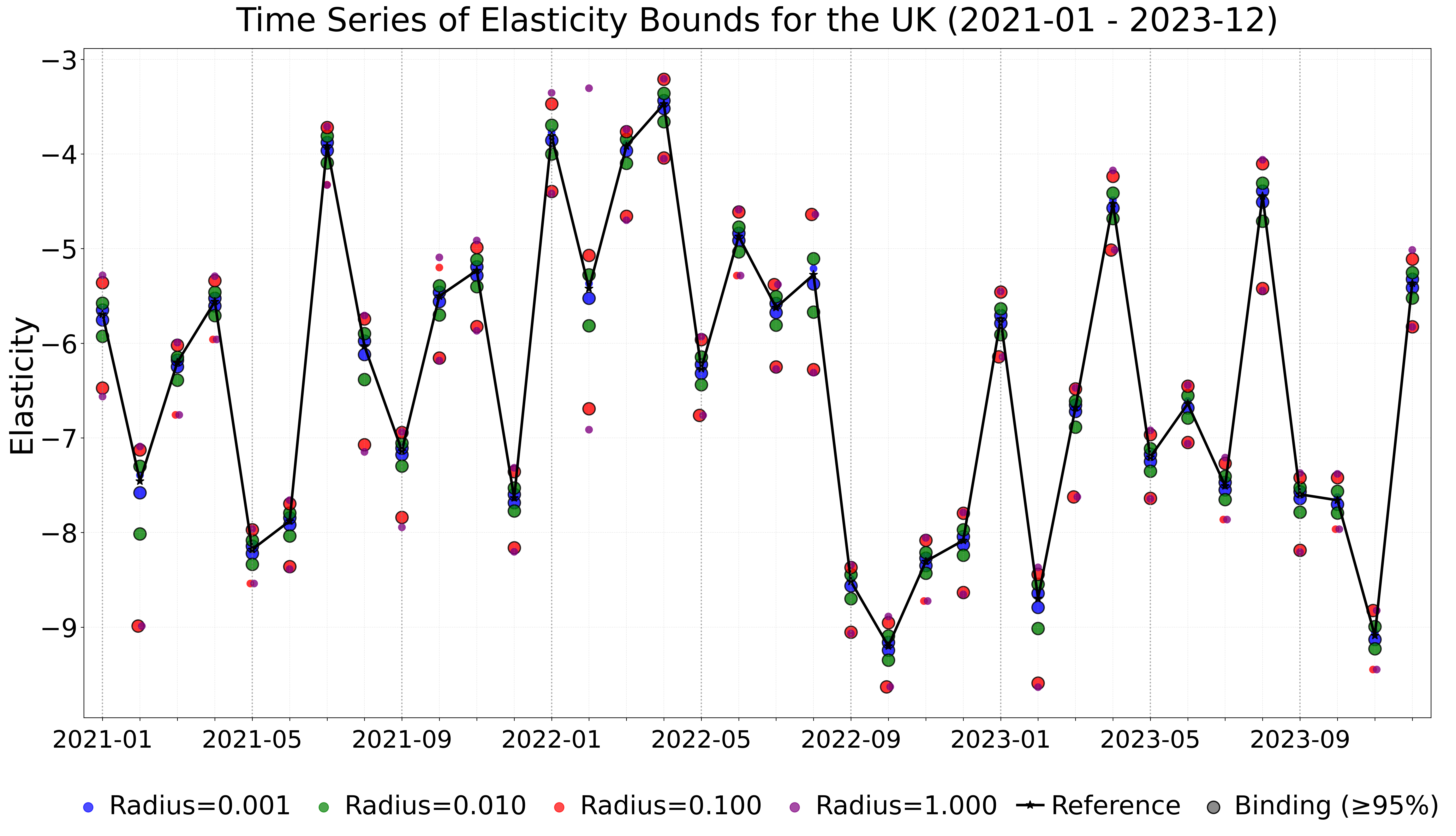}
    \end{minipage}
    
    
    \begin{minipage}{0.8\textwidth}
        \centering
        \includegraphics[width=\linewidth]{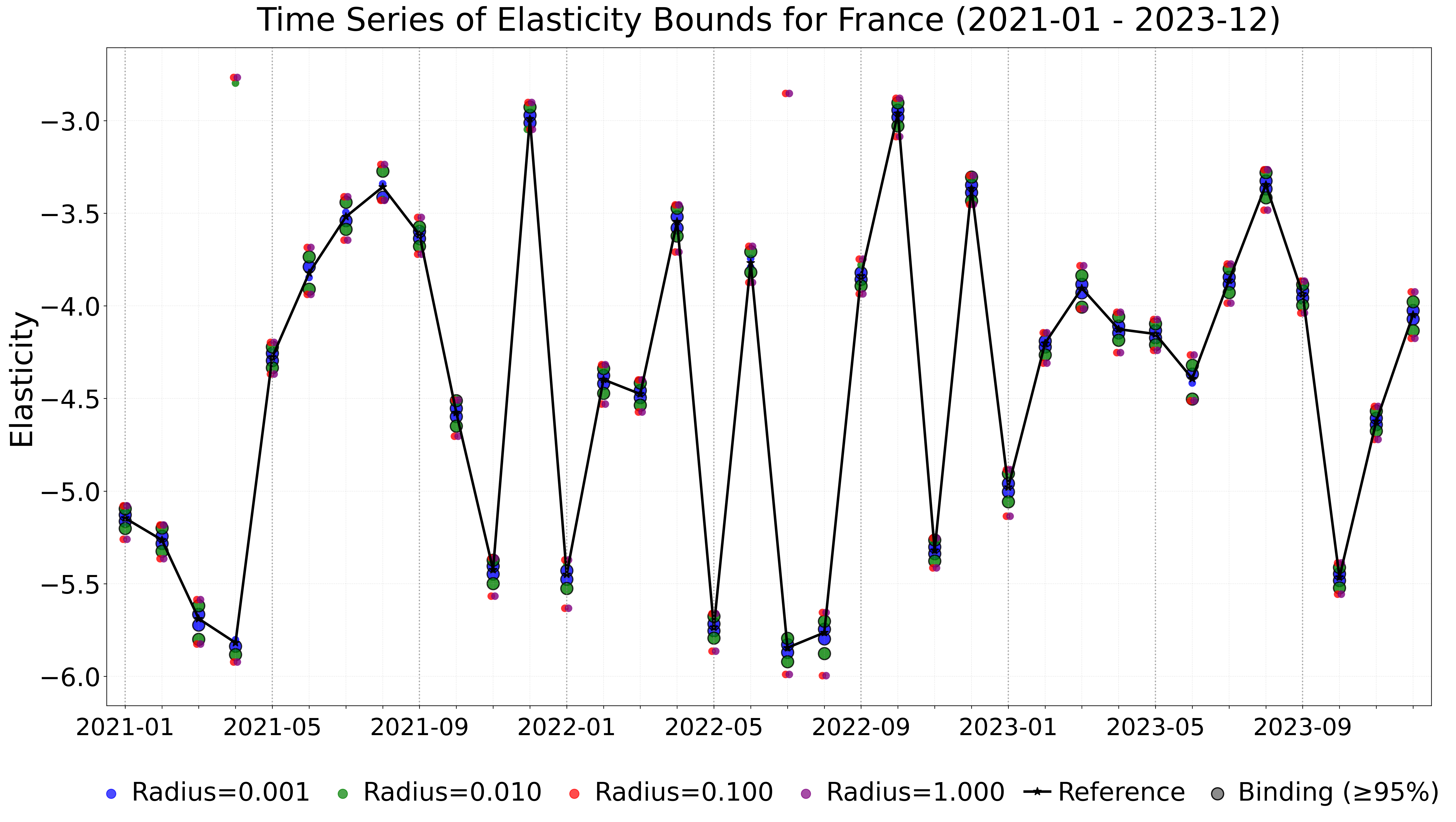}
    \end{minipage}
    
    
    \begin{minipage}{0.8\textwidth}
        \centering
        \includegraphics[width=\linewidth]{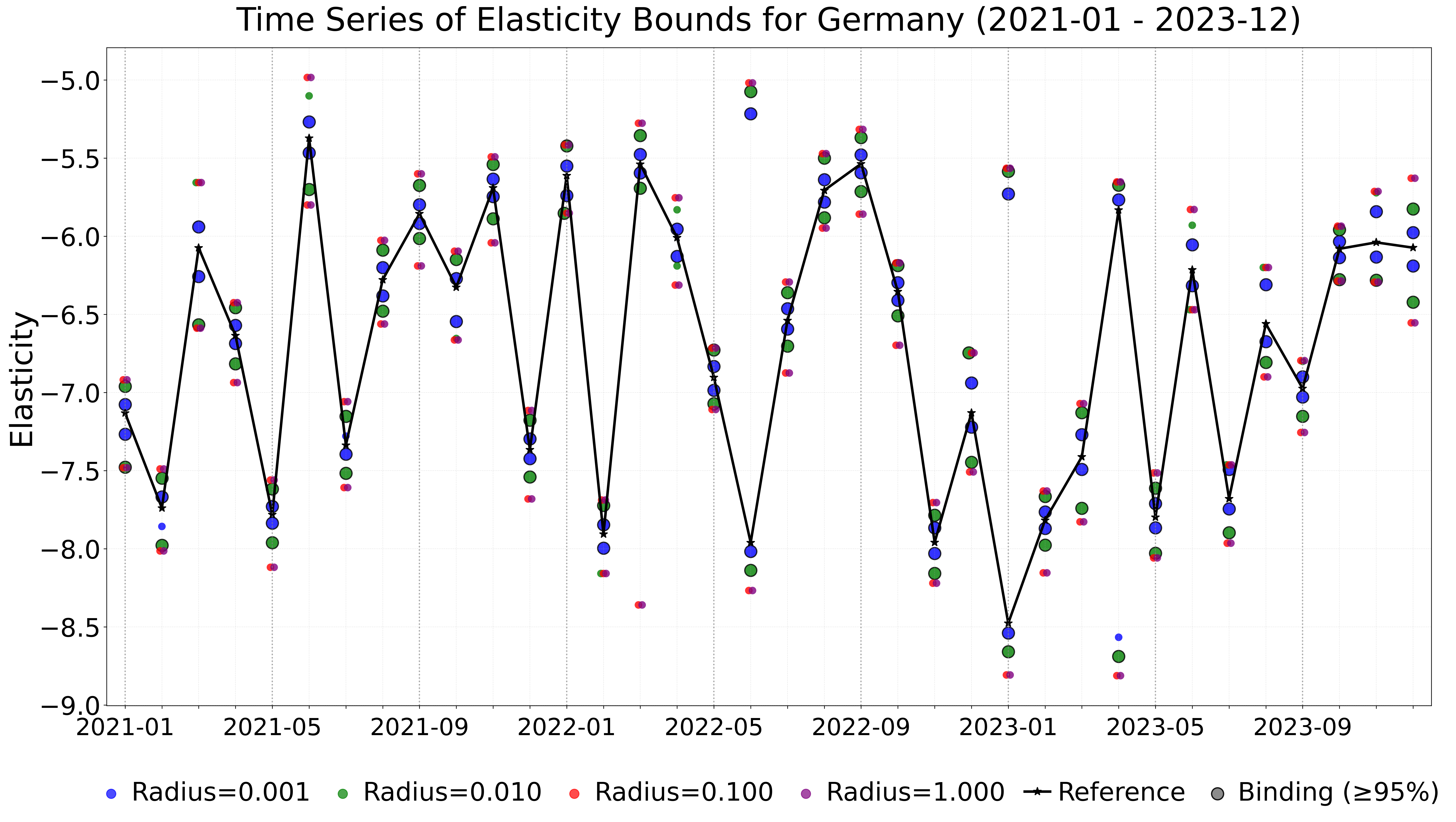}
    \end{minipage}

    \caption{Time Series of Bounds on Industrywide Elasticity for UK, Germany, and France}
    \label{fig:elasticity_bounds_time_series}
\end{figure}

\section{Empirical Application: Finite Horizon DDC} \label{sec: Application finite horizon}

\noindent This section applies our framework to a finite-horizon dynamic labor supply model for taxi drivers in New York City (NYC). In the model, the market-level supply shock is the latent variable. Our reference dynamic process is an AR(1) process. We consider the sensitivity analysis of the labor supply elasticity with respect to such distributional assumption.

\subsection{The Data}

\noindent We use data from New York City’s Taxi and Limousine Commission's (TLC) Taxi Passenger Enhancement Project (TPEP). The TPEP data contain a complete record of all trips operated by licensed drivers. The day shift starts at 5 AM and ends at 5 PM, and the night shift starts at 5 PM and ends at 5 AM. We choose a sample of 10,500 drivers that were active in 2013 as in \cite{kalouptsidi2021linear}. We restrict the sample to day shift drivers who were only working during the day shifts. We aggregate the transaction-level data to the driver-hour level. In addition, we create 10 uniformly divided bins for weekdays (Monday-Thursday) and 4 bins\footnote{The number of bins is chosen so that a Gaussian distribution approximates the stationary distribution of the market-level supply shock recovered from the last period (see \Cref{fig:last_omega}).} for weekends (Friday-Sunday) between the lowest and highest hourly earnings and calculate the average hourly earnings in each bin. Then, we remove important days (i.e., Memorial Day, the Fourth of July, and New Year's Eve). Finally, we restrict the sample to shifts that started between 5 AM and 8 AM, which accounts for 86.84\% shifts for weekdays, and 71.84\% for weekends. The final sample contains 3,562 drivers and 206 days for weekdays, and 3,322 drivers and 156 days for weekends.

\Cref{tab:summary_stats} presents the hourly summary statistics. The average hourly earnings range from \$24.22 at 4:00 PM to \$37.00 at 8:00 AM. The share of drivers who continue working is high in the early morning, with 100\% of drivers working at 6:00 AM and 7:00 AM. This share starts to decline after 2 PM, and drops substantially to 52.24\% and 53.59\% by 4:00 PM. Therefore, we assume that drivers can only choose to stop working between 8 AM and 4 PM.

\begin{table}[!htbp]
\centering
\footnotesize
\caption{Summary statistics.}
\label{tab:summary_stats}
\begin{threeparttable}
\begin{tabular}{lcccccc}
\toprule
\multirow{3}{*}{\textbf{Hour}} & \multicolumn{2}{c}{\textbf{Share of Drivers that}} & \multicolumn{4}{c}{\textbf{Hourly Earnings (\$)}} \\
\cmidrule(lr){2-3} \cmidrule(lr){4-7}
 & \multicolumn{2}{c}{\textbf{Continue Working (\%)}} & \multicolumn{2}{c}{\textbf{Weekday}} & \multicolumn{2}{c}{\textbf{Weekend}} \\
\cmidrule(lr){2-3} \cmidrule(lr){4-5} \cmidrule(lr){6-7}
 & \textbf{Weekday} & \textbf{Weekend} & \textbf{Mean} & \textbf{Std Dev} & \textbf{Mean} & \textbf{Std Dev} \\
\midrule
6:00 AM  & 100.00 & 100.00 & 32.66 & 3.00 & 31.47 & 3.86 \\
7:00 AM  & 100.00 & 100.00 & 35.04 & 2.88 & 30.11 & 4.49 \\
8:00 AM  & 98.18  & 96.21  & 37.00 & 2.84 & 31.21 & 5.68 \\
9:00 AM  & 97.64  & 96.34  & 34.73 & 2.37 & 30.56 & 4.39 \\
10:00 AM & 96.57  & 96.58  & 30.02 & 2.19 & 30.72 & 3.33 \\
11:00 AM & 95.36  & 95.80  & 29.12 & 2.30 & 31.48 & 3.11 \\
12:00 PM & 95.53  & 93.98  & 30.86 & 2.26 & 32.88 & 2.70 \\
1:00 PM  & 95.13  & 94.80  & 30.62 & 2.34 & 33.55 & 2.65 \\
2:00 PM  & 92.60  & 92.84  & 33.80 & 2.32 & 35.02 & 2.62 \\
3:00 PM  & 80.61  & 82.09  & 34.67 & 2.02 & 35.45 & 2.64 \\
4:00 PM  & 52.24  & 53.59  & 24.22 & 2.08 & 25.61 & 2.40 \\
5:00 PM  & 0.00   & 0.00   & --    & --   & --    & --   \\[0.5em]
\midrule
\# of Drivers & 3,562 & 3,322 & \multicolumn{4}{c}{} \\
\# of Days    & 206 & 156 & \multicolumn{4}{c}{} \\
\bottomrule
\end{tabular}
\textbf{Note:} The table uses TPEP Data from January 1, 2013 to December 31, 2013. An observation is defined by a driver-hour.
\end{threeparttable}
\end{table}

\subsection{The Model} \label{subsec: Model Application finite horizon model}

\noindent At the beginning of hour $t$ of day $m$, a taxi driver $i$ decides whether to continue working ($a=1$) or not ($a=0$). The decision to stop working is a terminating action, meaning the driver exits the market upon stopping. The period utility of working is given by:
\begin{align*}
    u(a_{imt},k_{imt},w_{mt},\xi_{mt},\eps_{imt};\theta) = 
    \begin{cases}
         \theta_{0} + \theta_{1} k_{imt} + \theta_{2} k_{imt}^{2} + \theta_{3} w_{mt} + \xi_{mt} + \eps_{i1mt} & \text{if } a_{imt} = 1 \\
        \eps_{i0mt} & \text{if } a_{imt} = 0
    \end{cases}
\end{align*}
where $k_{imt}$ is the number of hours worked, $w_{mt}$ is the average hourly earnings, $\eps_{imt} := (\eps_{i1mt},\eps_{i0mt})$ is i.i.d. type I extreme value utility shocks, and $\xi_{mt}$ is an exogenously evolved stationary unobserved market-level supply shock. It captures the market-level time-variant unobserved heterogeneity such as weather, congestion, or city events.

We assume markets are i.i.d, and suppress the subscript $(i,m)$ for brevity. Let $u(k_{t},w_{t};\theta)$ be the deterministic part of period of utility of working up to $(\xi_{t},\eps_{t})$. Let $\beta := 0.999999$ be the discount factor. The smoothed Bellman equation at time $t$ is given by:
\begin{equation*}
    V_{t}(k_{t},w_{t},\xi_{t}) = \log\left(\exp\left(v_{0t}(k_{t},w_{t},\xi_{t})\right) + \exp\left(v_{1t}(k_{t},w_{t},\xi_{t})\right)\right)
\end{equation*}
where the conditional value functions of working and not working are given by:
\begin{align}
    & v_{1t}(k_{t},w_{t},\xi_{t}) = u(k_{t},w_{t};\theta) + \xi_{t} + \beta \bE_{\xi_{t+1}|\xi_{t}} \bE_{w_{t+1}|w_{t}}[V_{t+1}(k_{t}+1,w_{t+1},\xi_{t+1})] \label{eq: v1} \\
    & v_{0t}(k_{t},w_{t},\xi_{t}) = 0 \nonumber
\end{align}

\subsection{First-Stage Estimation}

\noindent With only one terminating action, the utility parameters cannot be identified using \eqref{eq: log odds ratio}. Therefore, we estimate the utility parameters using the Euler Equations in Conditional Choice Probabilities (ECCP) estimator introduced in \cite{kalouptsidi2021linear}. By the \textit{Hotz-Miller Inversion Lemma} (\cite{hotz1993conditional}), we have:
\begin{equation} \label{eq: CCP finite horizon}
    V_{t}(k_{t},w_{t},\xi_{t}) = - \log p_{t}(k_{t},w_{t},\xi_{t}) = v_{1t}(k_{t},w_{t},\xi_{t}) - \log(1-p_{t}(k_{t},w_{t},\xi_{t}))
\end{equation}
where $p_{t}(k_{t},w_{t})$ is the CCP of not working. Combining \eqref{eq: v1} and \eqref{eq: CCP finite horizon} gives:
\begin{equation}
    \log\left(\frac{1-p_{t}(k_{t},w_{t},\xi_{t})}{p_{t}(k_{t},w_{t},\xi_{t})}\right) = u(k_{t},w_{t};\theta) + \xi_{t} - \beta \bE_{\xi_{t+1}|\xi_{t}} \bE_{w_{t+1}|w_{t}} \left[\log p_{t+1}(k_{t}+1,w_{t+1},\xi_{t+1})\right] \label{eq: ECCP finite horizon}
\end{equation}

Without a distributional assumption for $\xi_{t}$, we cannot calculate the conditional expectation in \eqref{eq: ECCP finite horizon}. However, the cross-sectional data allows us to estimate CCPs, denoted as $\hat{p}_{t}(k,w)$. We estimate $\hat{p}_{t}(k,w)$ by a flexible logit for each $t$. Let the expectational error be:
\begin{equation*}
    \hat{e}(k_{t},w_{t},k_{t+1},w_{t+1},\xi_{t}) := \beta \log \hat{p}_{t+1}(k_{t+1},w_{t+1}) -\beta \bE_{\xi_{t+1}|\xi_{t}} \bE_{w_{t+1}|w_{t}} \left[\log p_{t+1}(k_{t+1},w_{t+1},\xi_{t+1})\right]
\end{equation*}
where $k_{t+1} = k_{t} + 1$. Then, we can rewrite \eqref{eq: ECCP finite horizon} as:
\begin{equation*} \label{eq: ECCP finite horizon 2}
    \log\left(\frac{1-\hat{p}_{t}(k_{t},w_{t})}{\hat{p}_{t}(k_{t},w_{t})}\right) + \beta \log \hat{p}_{t+1}(k_{t+1},w_{t+1}) = u(k_{t},w_{t};\theta) + \xi_{t} + \hat{e}(k_{t},w_{t},k_{t+1},w_{t+1},\xi_{t})
\end{equation*}
Therefore $\theta$ can be identified using an instrument for $\xi_{t} + \hat{e}(k_{t},w_{t},k_{t+1},w_{t+1},\xi_{t})$. Denote by $\mK_{t}$ the set of possible hours worked at $t$.\footnote{Note that $|\mK_{t}| = 4$ for $t \geq 9$ AM and $|\mK_{t}| = 3$ for $t = 8$ AM.} The ECCP estimator stacks all $k \in \mK_{t}$. A unit of observation is defined by day-hour.

Following \cite{kalouptsidi2021linear}, we use the previous day's average hourly earnings for the same hour as the IV. \Cref{tab:2sls_results_styled} shows the estimation results. The implied marginal value of time defined by $-\frac{\theta_{1} + 2\theta_{2}k}{\theta_{3}}$ ranges from \$0 at around $k=7$ hours to \$5.35 at $k=11$ hours for weekdays, and from \$0 at around $k=8$ hours to \$9.95 at $k=11$ hours for weekends.

\begin{table}[!htbp]
\centering
\footnotesize
\caption{ECCP Estimation Results}
\label{tab:2sls_results_styled}
\begin{threeparttable}
\begin{tabular}{l cccc}
\toprule
& \multicolumn{4}{c}{\textbf{Model Estimates}} \\
\cmidrule(lr){2-5}
& \multicolumn{2}{c}{\textbf{Weekday}} & \multicolumn{2}{c}{\textbf{Weekend}} \\
\cmidrule(lr){2-3} \cmidrule(lr){4-5}
\textbf{Variable} & \textbf{Coef.} & \textbf{Std. Err.} & \textbf{Coef.} & \textbf{Std. Err.} \\
\midrule
\textbf{Constant} & $-2.1720$ & $(0.0304)$ & $-0.5631$ & $(0.0724)$ \\
\textbf{Hours worked} & $\phantom{-}0.4035$ & $(0.0025)$ & $\phantom{-}0.1866$ & $(0.0032)$ \\
\textbf{Hours worked (squared)} & $-0.0274$ & $(0.0002)$ & $-0.0116$ & $(0.0003)$ \\
\textbf{Average hourly earnings} & $\phantom{-}0.0373$ & $(0.0010)$ & $\phantom{-}0.0069$ & $(0.0023)$ \\[0.5em]
\midrule
\# of Drivers & \multicolumn{2}{c}{3,562} & \multicolumn{2}{c}{3,322} \\
\# of Days & \multicolumn{2}{c}{206} & \multicolumn{2}{c}{156} \\
\bottomrule
\end{tabular}
\textbf{Notes:} A unit of observation is defined by day-hour. Each observation stacks all $k_{t} \in \mK_{t}$. The model is estimated using 2SLS with the previous day's average hourly earnings as the IV. The standard errors are clustered at the day-$k$ level.
\end{threeparttable}
\end{table}

\subsection{The Reference Distribution and Scalar Parameters of Interest}

\noindent We define the reference distribution and introduce the scalar parameters of interest. Let $N_{tk}$ be the number of drivers who has worked $k$ hours at hour $t$. The market-level $\xi_{t}$ equates the model-implied weighted average CCP of not working with the observed weighted average:
\begin{equation} \label{eq: xi_t}
    \sum_{k \in \mK_{t}} \frac{N_{tk}}{\sum_{k \in \mK_{t}} N_{tk}} p_{t}(k,w_{t},\xi_{t}) = \sum_{k \in \mK_{t}} \frac{N_{tk}}{\sum_{k \in \mK_{t}} N_{tk}} \hat{p}_{t}(k,w_{t})
\end{equation}

As we assume $\xi_{t}$ is stationary, its marginal distribution at $T$ is its stationary distribution. At hour $T$, drivers solve a static problem, and the CCP of not working is:
\begin{equation*}
    p_{T}(k_{T},w_{T},\xi_{T}) = \frac{1}{1+\exp\left(u(k_{T},w_{T},\xi_{T})\right)}
\end{equation*} 
Therefore, the stationary distribution is identified by recovering $\xi_{T}$ to satisfy \eqref{eq: xi_t} at $T$. \Cref{fig:last_omega} plots the kernel density estimator of $\xi_{T}$. We fit a Gaussian distribution to the last period $\xi_{T}$. Denote its mean and standard deviation as $\mu_{\xi}$ and $\sigma_{\xi}$, respectively.
\begin{figure}[!htbp]
    \centering
    \begin{subfigure}[b]{0.48\textwidth}
        \centering
        \includegraphics[width=\textwidth]{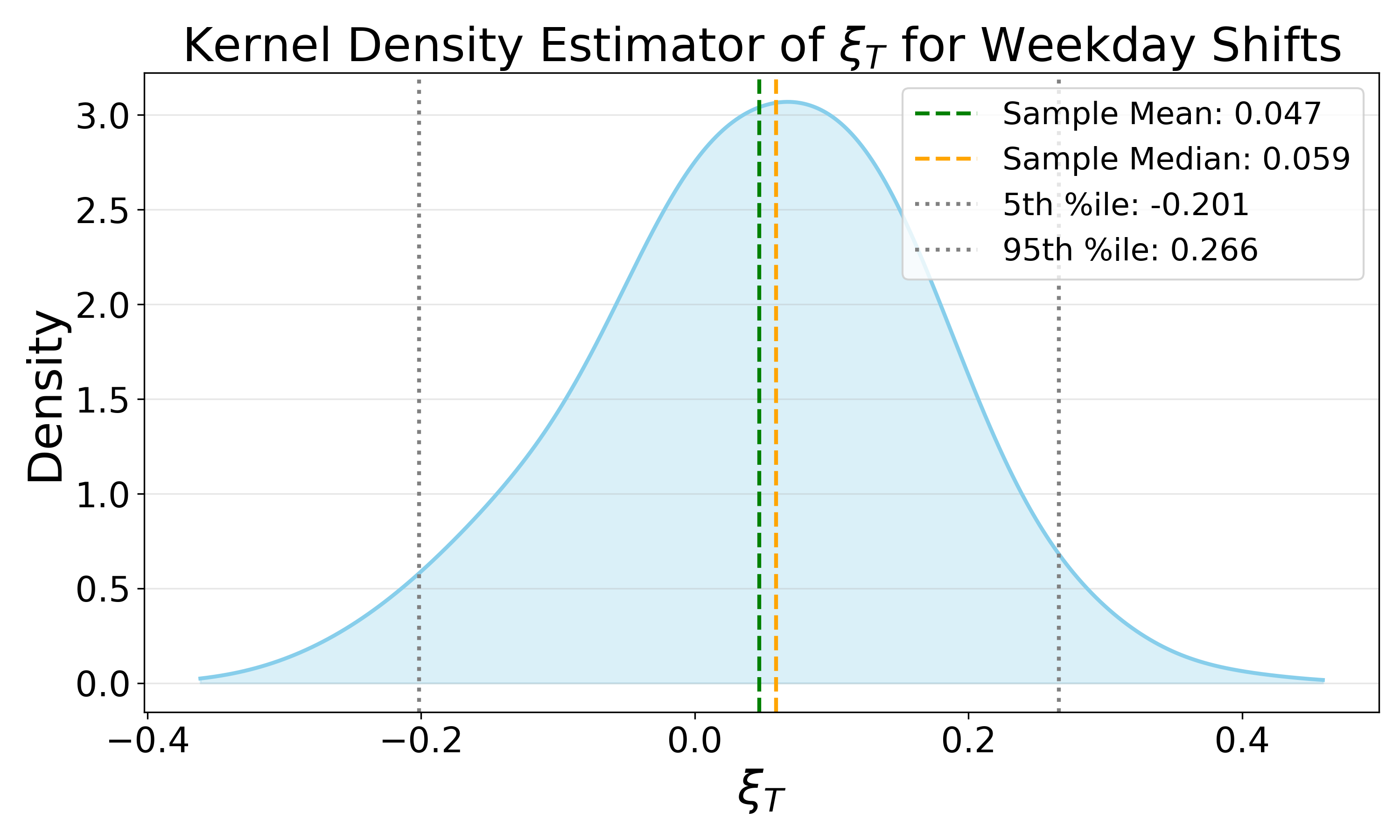}
    \end{subfigure}
    \hfill
    \begin{subfigure}[b]{0.48\textwidth}
        \centering
        \includegraphics[width=\textwidth]{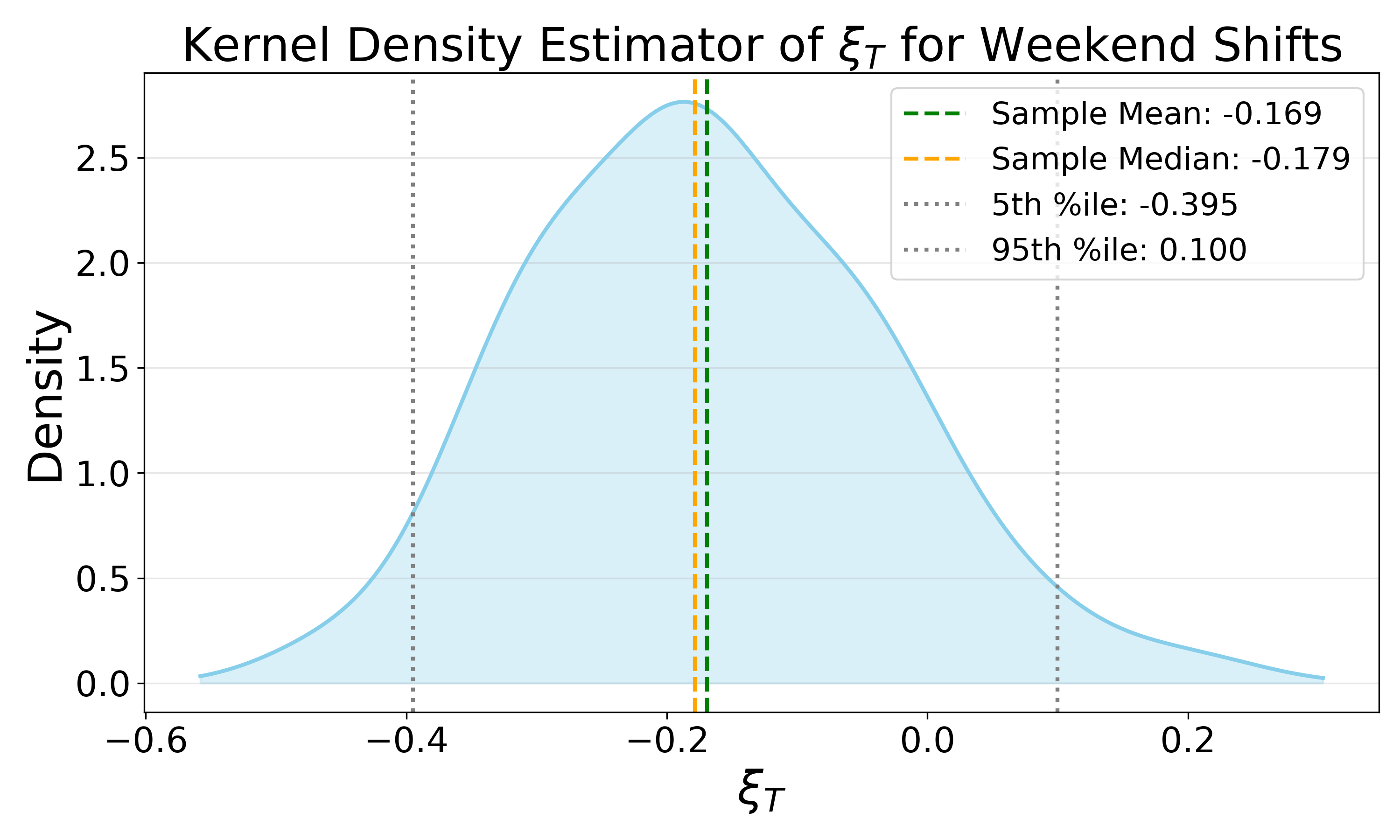}
    \end{subfigure}
    \caption{Kernel Density Estimator of the Last Period $\xi_{T}$}
    \label{fig:last_omega}
\end{figure}

The reference model for $\xi_{t}$ is an AR(1) process:
\begin{equation} \label{eq: AR(1) xi_t}
    \xi_{t} = \mu + \rho \xi_{t-1} + \eta_{t}
\end{equation}
where $\eta_{t}$ follows i.i.d normal distribution with mean 0 and variance $\sigma^{2}$. As we have identified its stationary distribution, we only need to solve the fixed point problem for $\rho$: we begin with an initial guess of $\rho$ and circulate between: (i) setting $\mu = \mu_{\xi}\cdot(1-\rho)$ and $\sigma = \sigma_{\xi} \cdot \sqrt{1-\rho^{2}}$, (ii) solving the Bellman equation using backward induction, (iii) recovering $\xi_{t}$ using \eqref{eq: xi_t}, and (iv) updating $\rho$ by refitting the AR(1) process to the recovered $\xi_{t}$ until convergence. The reference distribution $F_{0}$ for $(\xi,\xi')$ is the product of the transition kernel of the estimated AR(1) process and its stationary distribution $\nu_{0}$. The perturbation set is defined as:
\begin{equation*}
    \mF := \left\{ F \in \mP(\mU) \mid F \in \Pi(\nu_{0},\nu_{0}), D_{KL}(F\|F_{0}) \leq \delta \right\}
\end{equation*}

We consider two scalar parameters of interest: the elasticity of stopping working, and the Frisch elasticity of labor supply. Both elasticities are at the individual level, meaning the demand side remains unchanged. Therefore, we keep the transitions of $w_{t}$ and $\xi_{t}$ fixed.

For the elasticity of stopping working, we increase the average hourly earnings from the current bin, $w_{mt}$, to the next bin, $w'_{mt}$. The weighted average of the elasticity at hour $t$ is:
\begin{equation*}
    \sum_{m}\sum_{k \in \mK_{mt}} \frac{N_{mtk}}{\sum_{m,k \in \mK_{mt}}N_{mtk}} \frac{p_{t}(k,w'_{mt},\xi_{mt})-\hat{p}_{kmt}}{\hat{p}_{kmt}} \times \frac{w'_{mt} - w_{mt}}{w_{mt}}
\end{equation*}

As shown in \Cref{tab:summary_stats}, the share of drivers who continue working begins to decline around 11 AM.Moreover, our model does not endogenize the initial entry decision (i.e., the choice of when to start a shift). To compute the Frisch elasticity, we assume 11 AM is the first hour drivers can choose to stop working, and consider a 1\% increase in average hourly earnings beginning at 11 AM. The total expected hours worked at day $m$ is:
\begin{equation*}
    H(m,\bm{\xi}_{m},\bm{w}_{m},p_{}) := \sum_{k \in \{3,4,5,6\}} N_{mk} \sum_{t=12}^{16} (t-11)p_{mt}(k,w_{mt},\xi_{mt})\Pi_{t_{1}=11}^{t-1} (1-p_{t_{1}}(k,w_{mt_{1}},\xi_{mt_{1}}))
\end{equation*}
where $N_{mk}$ is the number of drivers whose hours worked is $k$ at 11 AM of day $m$, and $\bm{\xi}_{m},\bm{w}_{m}$ are the vectors of $\xi_{mt},w_{mt}$ for $t=11,\ldots,16$. Then, the Frisch elasticity is:
\begin{equation*}
    \frac{\sum_{m} H(m,\bm{\xi}_{m},\bm{w}'_{m},p'_{}) - \sum_{m} H(m,\bm{\xi}_{m},\bm{w}_{m},p_{})}{\sum_{m} H(m,\bm{\xi}_{m},\bm{w}_{m},p_{})} \times 100
\end{equation*}
where $\bm{w}'_{m} = 1.01 \cdot \bm{w}_{m}$, and $p'_{}$ is derived from the Bellman equation with $\bm{w}'_{m}$.

\subsection{Sensitivity Analysis}

\noindent We convert \eqref{eq: ECCP finite horizon} into an unconditional moment constraint by assuming that $p_{t}(k,w,\xi)$ is the solution to \eqref{eq: ECCP finite horizon} for given $p_{t+1}(k+1,w,\xi)$ if and only if:
\begin{equation} \label{eq: structural constraint1}
    \sup_{g_{tk} \in C(\mW \times \Xi)} \bE_{F} \bE_{w_{t+1},w_{t}} \left[g_{tk}(w_{t},\xi_{t})\left(\log\left(\frac{1-p_{t}(k,w_{t},\xi_{t})}{p_{t}(k,w_{t},\xi_{t})}\right) - u(k,w_{t}) - \xi_{t} + \beta \log p_{t+1}(k+1,w_{t+1},\xi_{t+1})\right)\right]
\end{equation}
Let the term inside the expectation in \eqref{eq: structural constraint1} be $\psi_{t}(k,w_{t},w_{t+1},\xi_{t},\xi_{t+1};u,p_{t},p_{t+1},g_{tk})$.

To profile out $\xi_{mt}$ from recovering the probability of stop working, we assume:
\begin{assumption} \label{assumption: invertible1}
    For $\forall \ F \in \mF$, the solution $p_{tk}(w,\xi)$ corresponding to \eqref{eq: structural constraint1} satisfies the following: for all $m,t$, there exists a unique $\xi_{mt} \in \Xi$ that satisfies \eqref{eq: xi_t}.
\end{assumption}

The final constraint is a fixed point constraint similar to the procedure used to estimate the AR(1) process. Suppose $F$ is used in \eqref{eq: structural constraint1}. Denote by $\hat{F}$ an estimator of the distribution of the pair $(\xi,\xi')$ using the recovered $\{\xi_{mt}\}_{m=1,t=1}^{M,T-1}$ from \eqref{eq: xi_t}. Then, our fixed point constraint is:
\begin{equation*}
    D_{KL}(F\|\hat{F}) \leq \epsilon_{M}
\end{equation*}
where $\epsilon_{M}$ is the tolerance level. As the sample size is large, we use Scott's Rule to initialize the bandwidth and then use 5-fold cross-validation to select bandwidth candidates around Scott's estimate that maximizes the log-likelihood. To choose $\epsilon_{M}$, we estimate the joint distribution of supply shocks recovered from the AR(1) process by the kernel density estimator with Gaussian kernel. Then, we set $\epsilon_{M}$ to be the KL divergence between the kernel density estimator and the reference distribution.

Then, the lower bound on the elasticity of stopping working at $t$ is:
\begin{align}
    \inf_{\{p_{t}\}_{t=1}^{T-1}} \inf_{F \in \mF} & \quad \sum_{m}\sum_{k \in \mK_{mt}} \frac{N_{mtk}}{\sum_{m,k \in \mK_{mt}}N_{mtk}} \frac{p_{t}(k,w'_{mt},\xi_{mt})-\hat{p}_{kmt}}{\hat{p}_{kmt}} \times \frac{w'_{mt} - w_{mt}}{w_{mt}} \nonumber \\
    \text{s.t.} \quad
    & \sum_{k \in \mK_{mt}} \frac{N_{mtk}}{\sum_{k \in \mK_{mt}} N_{mtk}} p_{t}(k,w_{mt},\xi_{mt}) = \sum_{k \in \mK_{mt}} \frac{N_{mtk}}{\sum_{k \in \mK_{mt}} N_{mtk}} \hat{p}_{t}(k,w_{mt}) \text{ for } \forall \ m, t \nonumber \\
    & \sup_{g_{tk}} \bE_{F} \bE_{w_{t+1},w_{t}} \left[\psi_{t}(k,\xi_{t},\xi_{t+1},w_{t},w_{t+1};u,p_{t},p_{t+1},g_{tk})\right] \text{ for } \forall \ t \leq T-1,k  \label{eq: finite horizon}\\
    & D_{KL}(F\|\hat{F}) \leq \epsilon_{M} \nonumber
\end{align}
The last period $T$ is a static problem, therefore $p_{T}$ is not an optimization parameter. For $\delta = 0$, the reference distribution is the unique solution to the above problem. The corresponding elasticity is the reference elasticity.

For the Frisch elasticity, we increase the earnings coefficient by 1\% from 11 AM. This allows us to leave the discretized state space and its transition unchanged at the cost of solving additional Bellman equations. The lower bound on the Frisch elasticity is:
\begin{align*}
    \inf_{\{p_{t},p'_{t}\}_{t=1}^{T-1}} \inf_{F \in \mF} & \quad \frac{\sum_{m} H(m,\bm{\xi}_{m},\bm{w}'_{m},p'_{}) - \sum_{m} H(m,\bm{\xi}_{m},\bm{w}_{m},p_{})}{\sum_{m} H(m,\bm{\xi}_{m},\bm{w}_{m},p_{})} \times 100  \\
    \text{s.t.} \quad
    & \sum_{k \in \mK_{mt}} \frac{N_{mtk}}{\sum_{k \in \mK_{mt}} N_{mtk}} p_{t}(k,w_{mt},\xi_{mt}) = \sum_{k \in \mK_{mt}} \frac{N_{mtk}}{\sum_{k \in \mK_{mt}} N_{mtk}} \hat{p}_{t}(k,w_{mt}) \text{ for } \forall \ m, t \\
    & \sup_{g^{1}_{tk}} \bE_{F} \bE_{w_{t+1},w_{t}} \left[\psi_{t}(k,\xi_{t},\xi_{t+1},w_{t},w_{t+1};u,p_{t},p_{t+1},g^{1}_{tk})\right] \text{ for } \forall \ t \leq T-1,k \\
    & \sup_{g^{2}_{tk}} \bE_{F} \bE_{w_{t+1},w_{t}} \left[\psi_{t}(k,\xi_{t},\xi_{t+1},w_{t},w_{t+1};u',p_{t}',p_{t+1}',g^{2}_{tk})\right] \text{ for } \forall \ t \leq T-1,k \\
    & D_{KL}(F\|\hat{F}) \leq \epsilon_{M}
\end{align*}
where $u'$ is the utility function with the earning coefficient increased by 1\% from 11 AM.

\subsection{Implementation}

\noindent We adjust the procedure in \Cref{sec: Implementation infinite horizon DDC} to account for the finite horizon model. The main difference is that we solve the Bellman equation \eqref{eq: ECCP finite horizon} by backward induction. The number of Bellman equations \eqref{eq: finite horizon} is 31.\footnote{Note that $\mK_{8} = \{1,2,3\}, \mK_{t} = \{t-8,t-7,t-6,t-5\}$ for $t=9,\cdots,15$.} The number of optimization parameters is $31 \cdot N_{w} N_{\xi} + 1$ for the elasticity of stopping working, and $62 \cdot N_{w} N_{\xi} + 1$ for the Frisch elasticity, where $N_{w}$ is the number of bins for average hourly earnings, $N_{\xi}$ is the number of grid points for $\xi$, and 1 is for the KL divergence constraint. We choose $N_{\xi} = 99$, 5,000 MCMC steps for Frisch elasticities, 2,500 MCMC steps for elasticities of stopping working, 5 optimization steps, 14 radii (the last is $10^{10}$), and 100 as the simulated annealing multiplier. The covariance matrix for the random walk in \Cref{alg:Initial MCMC Optimization} is restricted to be diagonal.

\subsection{Results}

\noindent The alternative model is an independent model where the supply shocks are i.i.d. and follow the stationary distribution identified at $T$. The independent model is closer to the reference model for weekends than for weekdays. The KL divergence between the independent model and the reference model is 0.115 for weekends and 3.673 for weekdays.

\Cref{fig:elasticity_bounds_NYC} plots the elasticity of stopping working from 9 AM to 3 PM for weekdays and weekends. We set $\delta=0.001$ for local deviation and $\delta = 1$ for global deviation. Large points indicate that the KL divergence constraint is binding. When two consecutive points align horizontally, this indicates that increasing the radius does not affect the bounds. The elasticity of stopping working is negative and decreases over time. In particular, the labor supply is inelastic before 11 AM on weekdays and 12 PM on weekends, and elastic after that. At 3 PM, the elasticity of stopping working is around -2 for weekdays and -2.5 for weekends. Overall, both weekday and weekend elasticities are not sensitive to the distributional assumption. For weekdays, elasticity in the morning is more sensitive than in the afternoon in terms of both local and global sensitivity, while it is the opposite for weekends.

\Cref{fig:frisch_elasticity_bounds_NYC} plots the Frisch elasticity bounds for weekdays and weekends. The reference Frisch elasticity is 0.472 for weekends and 0.698 for weekdays. Our reference estimates are consistent with labor supply literature. For example, \cite{buchholz2023rethinking} reports a Frisch elasticity ranging from 0.47-0.54 for NYC taxi drivers. For both weekdays and weekends, the bounds flatten around 0.032-0.056. In terms of local sensitivity, the results appear sensitive to the distributional assumption, with a deviation of 28.37\% for weekdays and 21.76\% for weekends. In terms of global sensitivity, the deviation is larger, with 76.83\% for weekdays and 42.84\% for weekends. For the robustness metric approach, we consider a 15\% deviation from the reference Frisch elasticity. The weekday's robustness metric is around 0.001 for the upper bound, while the lower bound never reaches the 15\% deviation. The weekend's robustness metric is around 0.008 for the upper bound, and 0.01 for the lower bound. Therefore, the weekday's lower bound is more robust than the weekend's lower bound, while the weekday's upper bound is less robust than the weekend's upper bound.

\begin{figure}[H]
    \centering
    \begin{subfigure}[b]{0.48\textwidth}
        \centering
        \includegraphics[width=\textwidth]{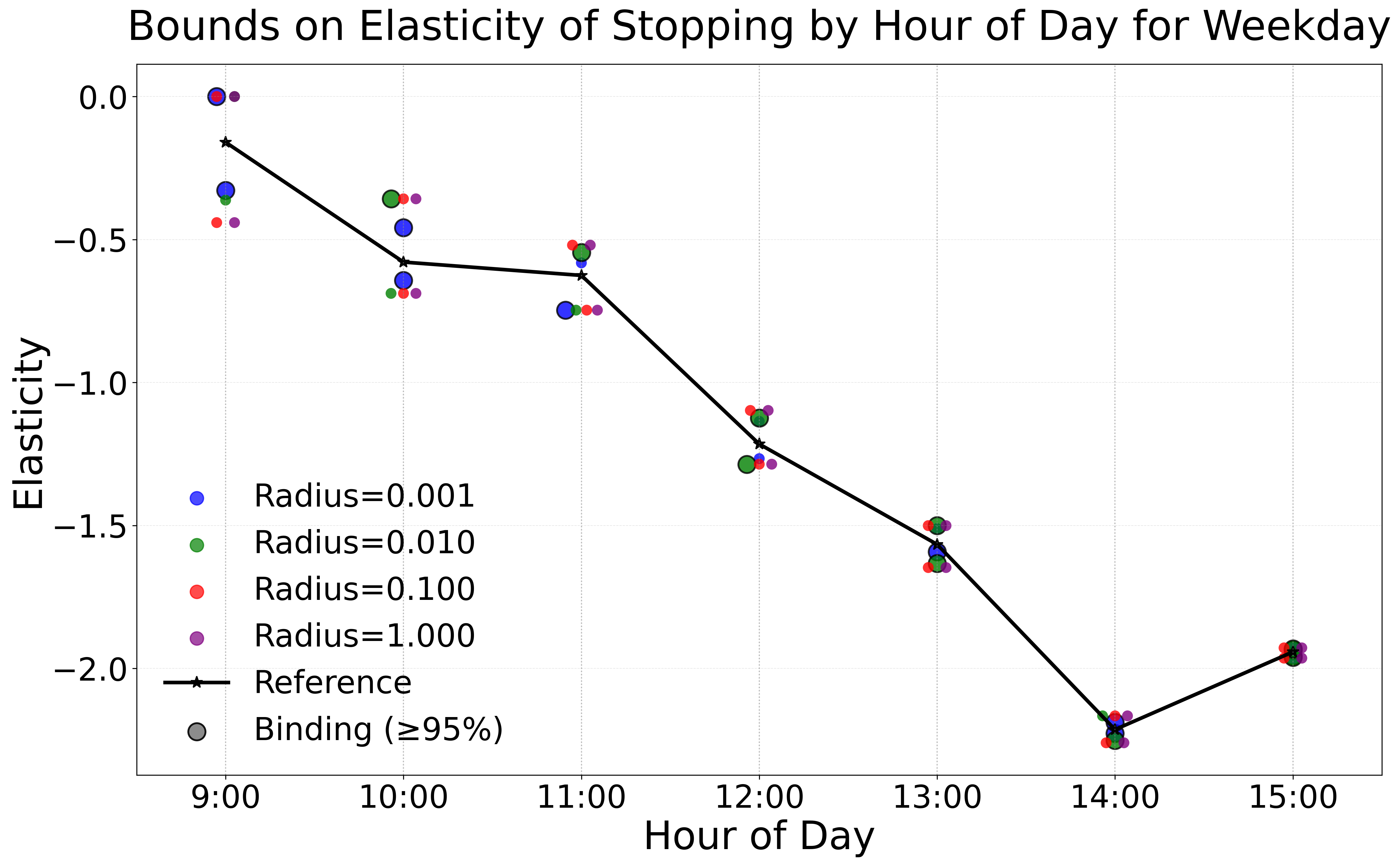}
    \end{subfigure}
    \hfill
    \begin{subfigure}[b]{0.48\textwidth}
        \centering
        \includegraphics[width=\textwidth]{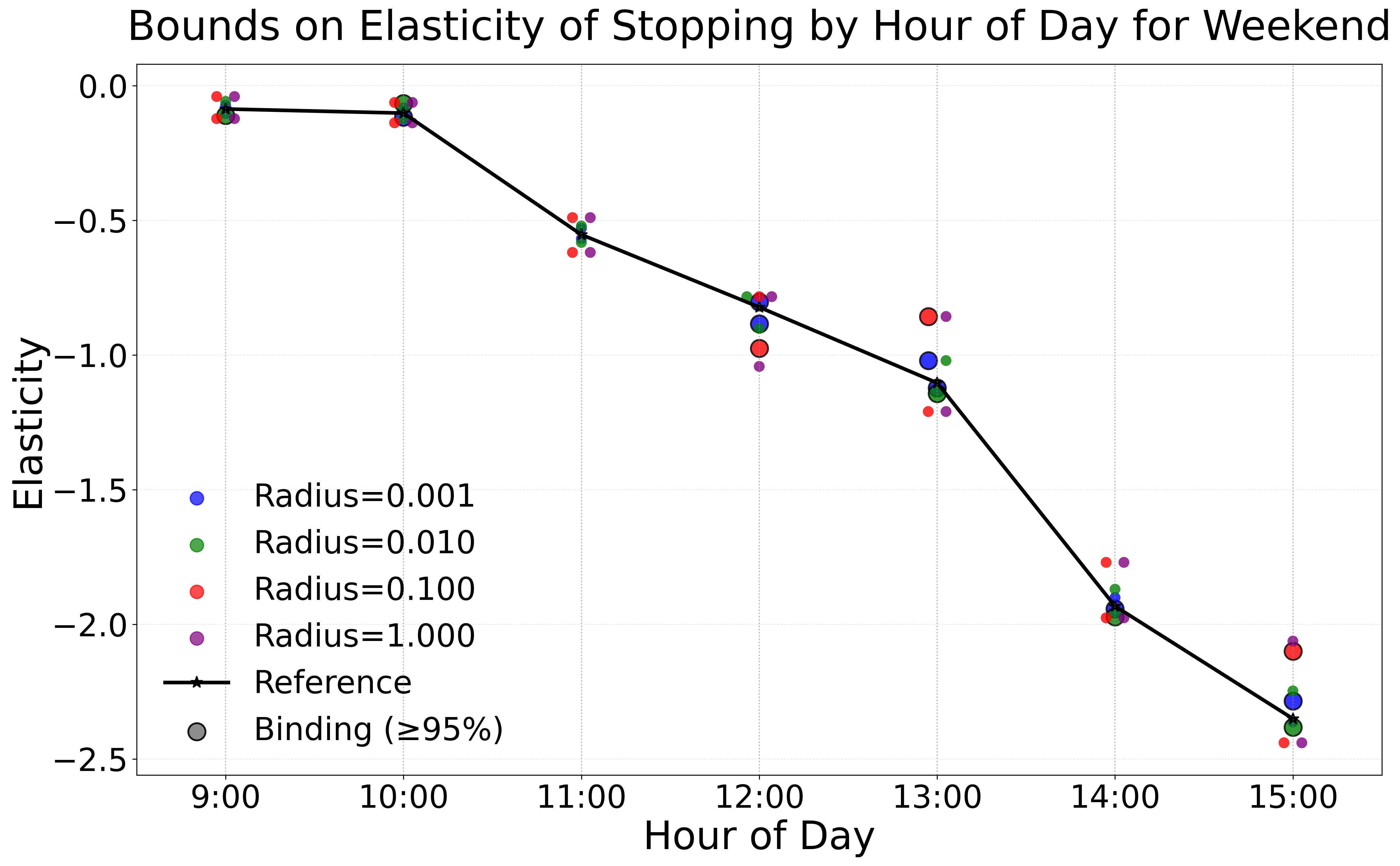}
    \end{subfigure}
    \caption{Elasticity of Stopping Working by Hour of Day for NYC Taxi Drivers}
    \label{fig:elasticity_bounds_NYC}
\end{figure}

\begin{figure}[H]
    \centering
    \begin{subfigure}[b]{0.48\textwidth}
        \centering
        \includegraphics[width=\textwidth]{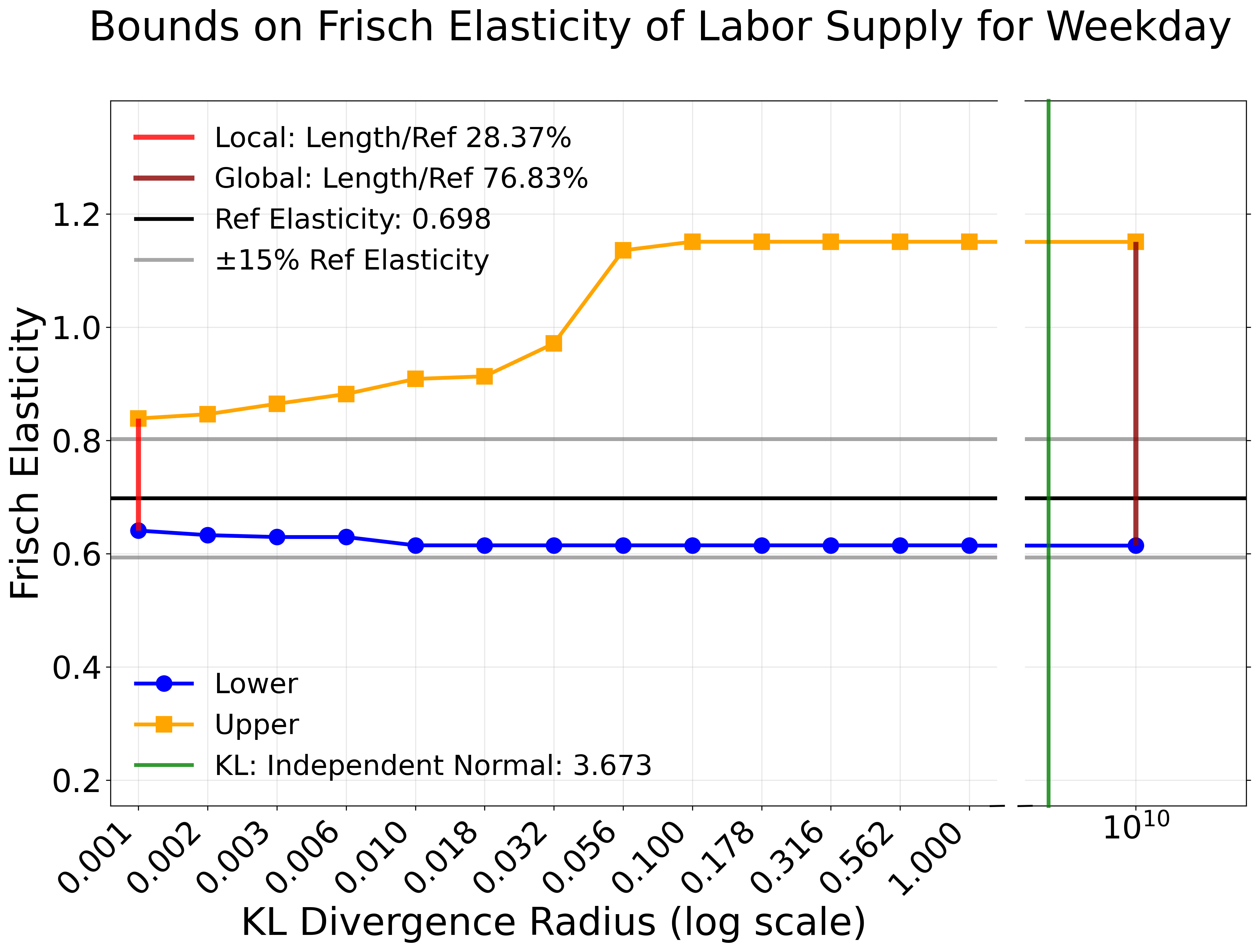}
    \end{subfigure}
    \hfill
    \begin{subfigure}[b]{0.48\textwidth}
        \centering
        \includegraphics[width=\textwidth]{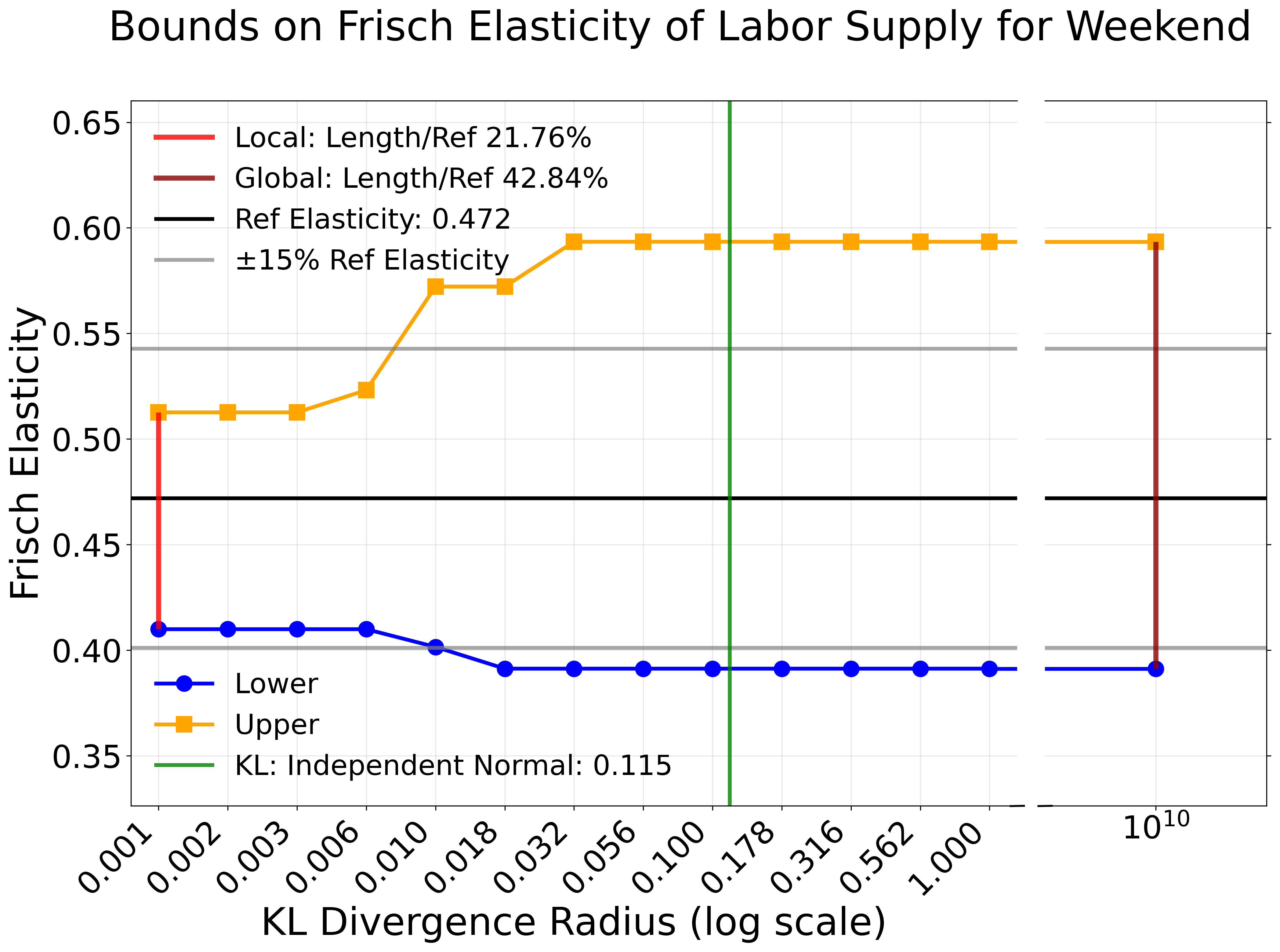}
    \end{subfigure}
    \caption{Frisch Elasticity Bounds by Hour of Day for NYC Taxi Drivers}
    \label{fig:frisch_elasticity_bounds_NYC}
\end{figure}

\section{Conclusion} \label{sec: Conclusion}

\noindent We propose a computationally tractable framework to quantify the sensitivity of outcomes of interest to misspecified latent-state dynamics in structural models. We derive bounds on a scalar parameter of interest by perturbing a reference dynamic process, while imposing a stationarity condition for time-homogeneous models or a Markovian condition for time-inhomogeneous models. We apply the approach to an infinite-horizon dynamic demand model for new cars in the UK, Germany, and France, and a finite-horizon dynamic labor supply model for taxi drivers in New York City.

\begin{singlespace}
    \setlength{\bibsep}{0pt}
    \setcitestyle{authoryear,round}
    \bibliographystyle{apalike}
    \bibliography{reference}
\end{singlespace}

\appendix
\section{Additional Examples} \label{sec: Additional Examples}

\begin{example}[Infinite Horizon Dynamic Discrete Choice Models] \label{ex: Dynamic Discrete Choice Models}
    This example considers the serial independence assumption on utility shocks in a single-agent DDC model, as in \cite{rust1987optimal}. Agents solve the Bellman equation for the conditional value function $v \in \mV$:
    \begin{equation} \label{eq: ex2 Bellman Equation}
        v_{j}(x,\eps) = u_{j}(x,\eps;\theta) + \beta \bE_{\eps'|\eps}\bE_{x'|x,j} \max_{j' \in \mJ} v_{j'}(x',\eps')
    \end{equation}
    where $\eps \in \bR^{J}$ is a vector of utility shocks for each action $j \in \mJ$, $x \in \mX$ is the observable state variable, $\beta \in (0,1)$ is the discount factor, $u_{j}(x,\eps;\theta)$ is the period utility parameterized by $\theta \in \Theta$, and $\mV$ is the class of conditional value functions.

    Let $U:=(\eps,\eps')$ be a vector of current and future utility shocks. The serial independence is often imposed on utility shocks, which implies a reference distribution:
    \begin{equation*}
        dF_{0}(U) := \nu_{0}(d\eps) \nu_{0}(d\eps')
    \end{equation*}
    whose perturbation set is defined as:
    \begin{equation*}
        \mF := \left\{ F \in \mP(\mU) \mid F \in \Pi(\nu_{0},\nu_{0}), D_{KL}(F\|F_{0}) \leq \delta \right\}
    \end{equation*}

    Suppose the scalar parameter of interest is the social welfare, defined as:
    \begin{equation*}
        \bE_{\nu_{0}} \bE_{x} \max_{j \in \mJ} v_{j}(x,\eps)
    \end{equation*}

    We convert the Bellman equation \eqref{eq: ex2 Bellman Equation} into a restriction that depends on the joint distribution $F \in \mF$. We assume there exists a class of Lagrange multiplier functions $\mG$ such that $v$ solves the Bellman equation \eqref{eq: ex2 Bellman Equation} if and only if:
    \begin{equation*}
        \sup_{g \in \mG} \bE_{F} \bE_{x,j,x'} \left[g_{j}(x,\eps)\left(v_{j}(x,\eps) - u_{j}(x,\eps;\theta) - \beta \max_{j' \in \mJ} v_{j'}(x',\eps')\right)\right] = 0
    \end{equation*}
    We can rewrite the structural constraint as:
    \begin{equation*}
        \sup_{g \in \mG} \bE_{F} \left[\psi(U;\theta,v,g)\right] = 0
    \end{equation*}
    where $\psi(U;\theta,v,g) := \bE_{x,j,x'} \left(g_{j}(x,\eps)\left(v_{j}(x,\eps) - u_{j}(x,\eps;\theta) - \beta \max_{j' \in \mJ} v_{j'}(x',\eps')\right)\right)$.

    We consider the following moment conditions for estimation:
    \begin{equation*}
        \bE_{F}\mathbbm{1}(v_{j}(x,\eps) = \max_{j' \in \mJ} v_{j'}(x,\eps)) = P_{0}(j|x) \quad \forall \ (j,x) \in \mJ \times \mX
    \end{equation*}
    where $\mathbbm{1}$ is the indicator function, and $P_{0}(j|x)$ is the population CCP. We assume $\mX$ has discrete support, and rewrite the moment conditions as:
    \begin{equation*}
        \bE_{F} \left[m(U;v)\right] = P_{0}
    \end{equation*}
    where $m(U;v)$ stacks the indicator functions for each $(x,j)$ given $v$, and $P_{0}$ stacks the CCPs.

    Then, the lower bound on the social welfare is given by:
    \begin{equation*}
        \begin{aligned}
            \inf_{(\theta,v,F) \in \Theta \times \mV \times \mF} 
            & \bE_{\nu_{0}} \bE_{x} \max_{j \in \mJ} v_{j}(x,\eps) \\
            \text{s.t.} \quad
            & \bE_{F} \left[m(U;v)\right] = P_{0} \\
            & \sup_{g \in \mG} \bE_{F} \left[\psi(U;\theta,v,g)\right] = 0
        \end{aligned}
    \end{equation*}

\end{example}

\begin{example}[Infinite Horizon Dynamic Discrete-Continuous Choice Models] \label{ex: Dynamic Discrete-Continuous Choice Models}
    This example considers the serial independence assumption on the consumption shock in a single-agent dynamic discrete-continuous choice model, as in \cite{iskhakov2017endogenous} and \cite{levy2022identification}. At each period, individuals make a discrete choice $j \in \mJ$, and a continuous choice $q \in \bR$. The value function solves the Bellman equation:
    \begin{equation*}
        V(x,I,\xi,\eps) = \max_{j,q}\left\{u_{j}(q,x,I,\xi;\theta) + \eps_{j} + \beta \bE_{\xi'|\xi}\bE_{x'|x,j} V(x',I',\xi',\eps') \right\}
    \end{equation*}
    where $I$ is the resource constraint (e.g., inventory) evolving deterministically over time according to $I' = L(x,I,q,j)$, $x \in \mX$ is the observable state variable, $\eps \in \bR^{J}$ is a vector of i.i.d. Extreme Value Type I utility shocks, $\xi \in \Xi$ is the consumption shock, $u_{j}(q,x,I,\xi;\theta)$ is the period utility parameterized by $\theta \in \Theta$, and $\beta \in (0,1)$ is the discount factor.

    Let $q^{*}_{j} := q^{*}_{j}(x,I,\xi)$ be the conditional optimal continuous choice for $(j,x,I,\xi)$. Given $q^{*}_{j}$, the conditional value function $v_{j} \in \mV$ solves:
    \begin{equation}
        v_{j}(x,I,\xi) = u_{j}(q^{*}_{j},x,I,\xi;\theta) + \beta \bE_{\xi'|\xi} \bE_{x'|x,j}\left[\log\left(\sum_{j' \in \mJ} \exp\left(v_{j'}(x',I',\xi')\right)\right)\right] + \beta \gamma \label{eq: ex3 Bellman Equation}
    \end{equation}
    The model-implied conditional choice probability is $p(j|x,I,\xi) = \frac{\exp\left(v_{j}(x,I,\xi)\right)}{\sum_{j' \in \mJ} \exp\left(v_{j'}(x,I,\xi)\right)}$.
    
    We assume that the optimal continuous choice is attained at an interior point and the dominated convergence theorem holds. Therefore, the first-order condition for \eqref{eq: ex3 Bellman Equation} holds:
    \begin{equation*}
        \frac{\partial u_{j}(q^{*}_{j},x,I,\xi;\theta)}{\partial q}+ \beta \bE_{\xi'|\xi} \bE_{x'|x,j}\left[\sum_{j' \in \mJ} p(j'|x',I',\xi') \frac{\partial v_{j'}(x',I',\xi')}{\partial I'}\right] \frac{\partial L(x,I,q^{*}_{j},j)}{\partial q} = 0
    \end{equation*}
    By the envelope theorem, we have:
    \begin{equation*}
        \frac{\partial v_{j'}(x',I',\xi')}{\partial I'} = \frac{\partial u_{j'}(q^{*}_{j'},x',I',\xi';\theta)}{\partial q}
    \end{equation*}
    Therefore, the Euler equation for the optimal continuous choice is:
    \begin{equation} \label{eq: ex2 Euler Equation}
        \frac{\partial u_{j}(q^{*}_{j},x,I,\xi;\theta)}{\partial q} + \beta \bE_{\xi'|\xi} \bE_{x'|x,j}\left[\sum_{j' \in \mJ} p(j'|x',I',\xi') \frac{\partial u_{j'}(q^{*}_{j'},x',I',\xi';\theta)}{\partial q}\right] \frac{\partial L(x,I,q^{*}_{j},j)}{\partial q} = 0
    \end{equation}

    Let $U:=(\xi,\xi')$ be a vector of current and future consumption shocks. The serial independence is often imposed on consumption shocks, which implies a reference distribution:
    \begin{equation*}
        dF_{0}(U) := \nu_{0}(d\xi) \nu_{0}(d\xi')
    \end{equation*}
    whose perturbation set is defined as:
    \begin{equation*}
        \mF := \left\{ F \in \mP(\mU) \mid F \in \Pi(\nu_{0},\nu_{0}), D_{KL}(F\|F_{0}) \leq \delta \right\}
    \end{equation*}
    Suppose the scalar parameter of interest is the social welfare, defined as:
    \begin{equation*}
        \bE_{\nu_{0}} \bE_{x,I}\left[\log\left(\sum_{j \in \mJ} \exp\left(v_{j}(x,I,\xi)\right)\right)\right]
    \end{equation*}

    We convert the Bellman equation \eqref{eq: ex3 Bellman Equation} and the Euler equation \eqref{eq: ex2 Euler Equation} into restrictions that depend on the joint distribution $F \in \mF$. We assume there exists a class of Lagrange multiplier functions $\mG$ such that $v := (v_{1},\cdots,v_{J},q^{*}_{1},\cdots,q^{*}_{J})$ solves the Bellman equation \eqref{eq: ex3 Bellman Equation}, and the Euler equation \eqref{eq: ex2 Euler Equation} if and only if:
    \begin{align*}
        & \sup_{g^{1} \in \mG} \bE_{F} \bE_{x,I,j,x'} \Biggl[g^{1}_{j}(x,I,\xi) \Biggl(v_{j}(x,I,\xi) - u_{j}(q^{*}_{j},x,I,\xi;\theta) - \beta \log\left(\sum_{j' \in \mJ} \exp\left(v_{j'}(x',I',\xi')\right)\right) - \beta \gamma\Biggr)\Biggr] = 0 \\
        & \sup_{g^{2} \in \mG} \bE_{F} \bE_{x,I,j,x'} \Biggl[g^{2}_{j}(x,I,\xi) \Biggl(\frac{\partial u_{j}(q^{*}_{j},x,I,\xi;\theta)}{\partial q}+ \beta\sum_{j' \in \mJ} p(j'|x',I',\xi') \frac{\partial u_{j'}(q^{*}_{j'},x',I',\xi';\theta)}{\partial q} \frac{\partial I'(x,I,q^{*}_{j},j)}{\partial q} \Biggr) \Biggr] = 0
    \end{align*}
    We rewrite the structural constraints as $\sup_{g \in \mG} \bE_{F} \left[\psi(U;\theta,v,g)\right] = 0$ where $\psi$ is the sum of the two expressions inside the expectations and $g := (g^{1},g^{2})$.

    We consider the following moment conditions for estimation: $\forall \ (j,x,I) \in \mJ \times \mX \times \mI$
    \begin{align*}
        \bE_{F} \left[p(j|x,I,\xi)\right] = P_{0}(j|x,I), \quad \bE_{F} \left[q^{*}_{j}(x,I,\xi)\right] = q_{0}(j,x,I)
    \end{align*}
    where $P_{0}(j|x,I)$ is the population CCP and $q_{0}(j,x,I)$ is the population continuous choice function. We discretize $\mX \times \mI$ for estimation, and rewrite the moment conditions as $\bE_{F} \left[m(U;v)\right] = P_{0}\bE_{F} \left[m(U;v)\right] = P_{0}$ where $m$ ($P_{0}$) stacks the model-implied (population) CCPs and continuous choice functions.

    Then, the lower bound on the social welfare is given by:
    \begin{align*}
        \inf_{(\theta,v,F) \in \Theta \times \mV \times \mF} \
        & \bE_{\nu_{0}} \bE_{x,I} \log\left(\sum_{j \in \mJ} \exp\left(v_{j}(x,I,\xi)\right)\right) \\
        \text{s.t.} \quad
        & \bE_{F} \left[m(U;v)\right] = P_{0} \\
        & \sup_{g \in \mG} \bE_{F} \left[\psi(U;\theta,v,g)\right] = 0
    \end{align*}
\end{example}

\section{Proofs} \label{sec:proofs}

\subsection{Supporting Lemmas}

\begin{lemma}[Fan's Minimax Theory] \label{lemma: Minimax Theory}
    Suppose the following conditions hold:
    \begin{enumerate}
        \item $X$ be a compact Hausdorff space and $Y$ a nonempty set (not necessarily topologized).
        \item Let $f: X \times Y \rightarrow \bR$ be a real-valued function.
        \item For $\forall \ y \in Y$, $f(\cdot,y)$ is convexlike on $X$, i.e., for all $x_{1}, x_{2} \in X$ and $\lambda \in [0,1]$, there exists $x_{0} \in X$ such that $f(x_{0},y) \leq \lambda f(x_{1},y) + (1-\lambda)f(x_{2},y)$.
        \item For $\forall \ x \in X$, $f(x,\cdot)$ is concavelike on $Y$, i.e., for all $y_{1}, y_{2} \in Y$ and $\lambda \in [0,1]$, there exists $y_{0} \in Y$ such that $f(x,y_{0}) \geq \lambda f(x,y_{1}) + (1-\lambda)f(x,y_{2})$.
        \item For $\forall \ y \in Y$, $f(\cdot,y)$ is lower semicontinuous on $X$.
    \end{enumerate}
    Then, we have:
    \begin{equation*}
        \sup_{Y} \inf_{X} f(x,y) = \inf_{X} \sup_{Y} f(x,y)
    \end{equation*}
\end{lemma}
\begin{proof}
    See \cite{ricceri2013minimax} Theorem 1.3.
\end{proof}

\begin{lemma} \label{lemma: consistency and convergence rate of the projection}
    Let $\{\mA_{n}\}$ be a sequence of compact sets such that $d_{H}(\mA_{n}, \mA) = o_{p}(1)$ where $\mA$ is compact. Then, the following holds:
    \begin{itemize}
        \item (Consistency) If $f: \mA \to \bR$ is continuous, then: $|\inf_{\mA_{n}} f - \inf_{\mA} f| = o_{p}(1)$.
        \item (Convergence Rate) If $d_{H}(\mA_{n}, \mA) = O_{p}(c_{n})$ for some $c_{n} \to 0$, and $f$ is Lipschitz continuous, then: $|\inf_{\mA_{n}} f - \inf_{\mA} f| = O_{p}(c_{n})$.
    \end{itemize}
\end{lemma}
\begin{proof}
    As $\mA_{n}$ and $\mA$ are compact and $f$ is continuous, the infimum is achieved by the extreme value theorem. Denote minimizers as $a_{n}^{*} \in \mA_{n}$ and $\alpha^{*} \in \mA$.

\begin{enumerate}
    \item As $d_{H}(\mA_{n}, \mA) = o_{p}(1)$, there exists a sequence $a_{n} \in \mA$ such that $d(a_{n}^{*},a_{n})=o_{p}(1)$. By the continuity of $f$, we have: $|f(a_{n}^{*}) - f(a_{n})| = o_{p}(1)$, which implies: $f(a_{n}^{*}) = f(a_{n}) + o_{p}(1) \geq f(a^{*}) + o_{p}(1)$. Similarly, there exists a sequence $b_{n} \in \mA_{n}$ such that $d(b_{n}, a^{*}) = o_{p}(1)$. By the continuity of $f$, we have: $|f(b_{n}) - f(a^{*})| = o_{p}(1)$, which implies: $f(a^{*}) = f(b_{n}) + o_{p}(1) \geq f(a_{n}^{*}) + o_{p}(1)$. Combining both inequalities, we have $f(a_{n}^{*}) \geq f(a^{*}) + o_{p}(1) \geq f(a_{n}^{*}) + o_{p}(1) + o_{p}(1)$, which implies: $|\inf_{\mA_{n}} f - \inf_{\mA} f| = |f(a_{n}^{*}) - f(a^{*})| = o_{p}(1)$.
    \item For the second part, let $L$ be the Lipschitz constant. There exists a sequence $a_{n} \in \mA$ such that $d(a_{n}^{*},a_{n})=O_{p}(c_{n})$, by the Lipschitz continuity of $f$, we have: $|f(a_{n}^{*}) - f(a_{n})| \leq L d(a_{n}^{*}, a_{n}) = O_{p}(c_{n})$, which implies: $f(a_{n}^{*}) \geq f(a_{n}) - L d(a_{n}^{*}, a_{n}) \geq f(a^{*}) - L d(a_{n}^{*}, a_{n})$. Similarly, there exists a sequence $b_{n} \in \mA_{n}$ such that $d(b_{n}, a^{*}) = O_{p}(c_{n})$. By the Lipschitz continuity of $f$, we have: $|f(b_{n}) - f(a^{*})| \leq L d(b_{n}, a^{*}) = O_{p}(c_{n})$, which implies: $f(a^{*}) \geq f(b_{n}) - L d(b_{n}, a^{*}) \geq f(a_{n}^{*}) - L d(b_{n}, a^{*})$. Combining both inequalities, we have: $f(a_{n}^{*}) \geq f(a^{*}) - L d(b_{n}, a^{*}) \geq f(a_{n}^{*}) - 2 L d(b_{n}, a^{*})$, which implies: $|f(a_{n}^{*}) - f(a^{*})| \leq 2L d(b_{n}, a^{*}) = O_{p}(c_{n})$.
\end{enumerate}
\end{proof}

\begin{lemma} \label{Lemma: compact space KL}
    Let $\delta > 0$ and $F_{0} \in \mP(\mU)$. Let $\mF_{KL} := \{F \in \mP(\mU) \mid D_{KL}(F\|F_{0}) \leq \delta \}$.
    \begin{lemmaitems}
        \item $\mF_{KL}$ is compact in the topology of weak convergence. \label{lemma: compact space KL}
        \item $\mF_{KL}$ is closed in the topology of weak convergence. \label{lemma: closed space}
    \end{lemmaitems}
\end{lemma}
\begin{proof}
    \begin{enumerate}
        \item See \cite{pinski2015kullback} Proposition 2.1.
        \item By \cite{nutz2021introduction} Lemma 1.3, $D_{KL}(F\|F_{0})$ is lower-semicontinuous in the topology of weak convergence, i.e., for $F_{n} \to F$ weakly, we have $\liminf_{n \rightarrow \infty} D_{KL}(F_{n}\|F_{0}) \geq D_{KL}(F\|F_{0})$. Since $F_{n} \in \mF_{KL}$, we have $D_{KL}(F\|F_{0}) \leq \delta$.
    \end{enumerate}
\end{proof}

\begin{lemma} \label{lemma: closed compact convex space}
    Let $\nu_{i} \in \mP(\mU_{i})$ for $i \in \{1,\cdots,k\}$, then:
    \begin{lemmaitems}
        \item $\Pi(\nu_{1},\cdots,\nu_{k})$ is closed in the topology of weak convergence. \label{lemma: closed space Pi}
        \item $\Pi(\nu_{1},\cdots,\nu_{k})$ is compact and convex in the topology of weak convergence. \label{lemma: compact convex space Pi}
        \item $\Pi(\nu_{1},\cdots,\nu_{k})$ is a Hausdorff topological space. \label{lemma: Hausdorff space Pi}
        \item Under \Cref{assumption: compact Support}, $\Pi(\nu_{1},\nu_{k})$ is closed in the topology of weak convergence. \label{lemma: closed space Pi2}
        \item $\Pi(\nu_{1},\nu_{k})$ is convex in the topology of weak convergence. \label{lemma: convex space Pi2}
        \item Under \Cref{assumption: compact Support}, $\Pi(\nu_{1},\nu_{k})$ is compact in the topology of weak convergence. \label{lemma: compact space Pi2}
        \item $\Pi(\nu_{1},\nu_{k})$ is a Hausdorff topological space. \label{lemma: Hausdorff space Pi2}
        \item Suppose $\mN$ is convex and closed and \Cref{assumption: compact Support} holds, $\left\{ F \in \mP(\mU) \mid F \in \Pi(\nu,\nu_{k}), \nu \in \mN\right\}$ is closed in the topology of weak convergence. \label{lemma: closed space Pi3}
        \item Suppose $\mN$ is convex and closed and \Cref{assumption: compact Support} holds, then $\left\{ F \in \mP(\mU) \mid F \in \Pi(\nu,\nu_{k}), \nu \in \mN\right\}$ is convex in the topology of weak convergence. \label{lemma: convex space Pi3}
        \item Suppose $\mN$ is convex and closed and \Cref{assumption: compact Support} holds, $\left\{ F \in \mP(\mU) \mid F \in \Pi(\nu,\nu_{k}), \nu \in \mN\right\}$ is compact in the topology of weak convergence. \label{lemma: compact space Pi3}
        \item Suppose $\mN$ is convex and closed and \Cref{assumption: compact Support} holds, $\left\{ F \in \mP(\mU) \mid F \in \Pi(\nu,\nu_{k}), \nu \in \mN\right\}$ is a Hausdorff topological space. \label{lemma: Hausdorff space Pi3}
        \item Lemmas \ref{lemma: closed space Pi3}-\ref{lemma: Hausdorff space Pi3} also hold for $\left\{ F \in \mP(\Xi^{2}) \mid F \in \Pi(\nu,\nu), \nu \in \mN\right\}$. \label{lemma: stationary perturbation set}
    \end{lemmaitems}
\end{lemma}
\begin{proof}
    \begin{enumerate}
        \item \textbf{Closedness and compactness of $\Pi(\nu_{1},\cdots,\nu_{k})$:} see the proof of \cite{villani2009optimal} Theorem 4.1.
        \item \textbf{Closedness and compactness of $\Pi(\nu_{1},\nu_{k})$:} By \Cref{assumption: compact Support}, $\Pi(\nu_{1},\nu_{k})$ is tight, and by Prokhorov's theorem it has a compact closure. By passing to the limit in the equation for marginals, $\Pi(\nu_{1},\nu_{k})$ is closed. Therefore, it is compact.
        \item \textbf{Closedness and compactness of $\left\{ F \in \mP(\mU) \mid F \in \Pi(\nu,\nu_{k}), \nu \in \mN\right\}$:} By \Cref{assumption: compact Support}, the set $\left\{ F \in \mP(\mU) \mid F \in \Pi(\nu,\nu_{k}), \nu \in \mN\right\}$ is tight, and by Prokhorov's theorem it has a compact closure. By passing to the limit in the equation for marginals, it is closed. Therefore, it is compact.
        \item \textbf{Closedness and compactness of $\left\{ F \in \mP(\mU) \mid F \in \Pi(\nu,\nu), \nu \in \mN\right\}$:} The proof is similar to the previous part.
        \item \textbf{Convexity:} It is straightforward.
        \item \textbf{Hausdorff:} By \cite{billingsley2013convergence} Page 72(i), the Prohorov distance is a metric on the space of probability measures. Metrizable topological spaces are Hausdorff.
    \end{enumerate}
\end{proof}

\begin{lemma} \label{lemma: perturbation set properties}
    Let $\delta \geq 0$, $F_{0} \in \mP(\mU)$, and $\nu_{i} \in \mP(\mU_{i})$ for $i \in \{1,\cdots,k\}$. Define:
    \begin{equation*}
        \mF := \{F \in \mP(\mU) \mid F \in \Pi(\nu_{1},\cdots,\nu_{k}), D_{KL}(F\|F_{0}) \leq \delta \}
    \end{equation*}
    \begin{lemmaitems}
        \item $\mF$ is closed and compact in the topology of weak convergence. \label{lemma: compact space}
        \item $\mF$ is convex in the topology of weak convergence. \label{lemma: convex space}
        \item $\mF$ is a Hausdorff topological space. \label{lemma: Hausdorff space}
    \end{lemmaitems}
\end{lemma}
\begin{proof}
\begin{enumerate}
    \item By Lemmas \ref{lemma: closed space} and \ref{lemma: closed space Pi}, $\mF$ is the intersection of two closed sets. Therefore, it is closed. By \cite{rudin1976principles} Theorem 2.35, $\mF$ is compact.
    \item Since KL divergence is jointly convex (see \cite{nutz2021introduction} Lemma 1.3), we have $D_{KL}(\lambda F_{1} + (1-\lambda)F_{2}\|F_{0}) \leq \lambda D_{KL}(F_{1}\|F_{0}) + (1-\lambda)D_{KL}(F_{2}\|F_{0}) \leq \delta$ for $\lambda \in [0,1]$ and $F_{1},F_{2} \in \mF$. Moreover, $\lambda F_{1} + (1-\lambda)F_{2} \in \Pi(\nu_{1},\cdots,\nu_{k})$. Therefore, $\mF$ is convex.
    \item By \cite{billingsley2013convergence} Page 72(i), the Prohorov distance is a metric on the space of probability measures. Metrizable topological spaces are Hausdorff.
\end{enumerate}
\end{proof}

\begin{lemma}
    Suppose \Cref{assumption: compact Support} holds. Let $\delta \geq 0$, $F_{0} \in \mP(\mU)$, and $\nu_{i} \in \mP(\mU_{i})$ for $i \in \{1,k\}$. Let $\mF_{\text{relaxed}} := \{F \in \mP(\mU) \mid F \in \Pi(\nu_{1},\nu_{k}), D_{KL}(F\|F_{0}) \leq \delta \}$.
    \begin{lemmaitems}
        \item $\mF_{\text{relaxed}}$ is closed and compact in the topology of weak convergence. \label{lemma: relaxed compact space}
        \item $\mF_{\text{relaxed}}$ is convex in the topology of weak convergence. \label{lemma: relaxed convex space}
        \item $\mF_{\text{relaxed}}$ is a Hausdorff topological space. \label{lemma: relaxed Hausdorff space}
    \end{lemmaitems}
\end{lemma}
\begin{proof}
\begin{enumerate}
    \item By Lemmas \ref{lemma: closed space} and \ref{lemma: closed space Pi2}, $\mF_{\text{relaxed}}$ is the intersection of two closed sets. Therefore, it is closed. By \cite{rudin1976principles} Theorem 2.35, $\mF_{\text{relaxed}}$ is compact.
    \item The proof is identical to the proof of \Cref{lemma: convex space}.
    \item The proof is identical to the proof of \Cref{lemma: Hausdorff space}.
\end{enumerate}
\end{proof}

\begin{lemma}
    Suppose \Cref{assumption: compact Support} holds. Let $\delta \geq 0$, $F_{0} \in \mP(\mU)$ and $\nu_{k} \in \mP(\mU_{k})$. Suppose $\mN \subseteq \mP(\mU_{1})$ is convex. Let $\mF_{\mN,\text{Relaxed}} := \left\{ F \in \mP(\mU) \mid F \in \Pi(\nu,\nu_{k}), \nu \in \mN, D_{KL}(F \| F_{0}) \leq \delta \right\}$.
    \begin{lemmaitems}
        \item $\mF_{\mN,\text{Relaxed}}$ is closed and compact in the topology of weak convergence. \label{lemma: relaxed initial compact space}
        \item $\mF_{\mN,\text{Relaxed}}$ is convex in the topology of weak convergence. \label{lemma: relaxed initial convex space}
        \item $\mF_{\mN,\text{Relaxed}}$ is a Hausdorff topological space. \label{lemma: relaxed initial Hausdorff space}
    \end{lemmaitems}
\end{lemma}
\begin{proof}
\begin{enumerate}
    \item By \Cref{lemma: closed space Pi3}, and \cite{rudin1976principles} Theorem 2.35, $\mF_{\mN,\text{Relaxed}}$ is compact.
    \item Note that $\mN$ is convex. For any $F_{1}, F_{2} \in \mF_{\mN,\text{Relaxed}}$, we have $\nu_{1}, \nu_{2} \in \mN$. By the convexity of $\mN$, we have $\lambda \nu_{1} + (1-\lambda) \nu_{2} \in \mN$ for $\lambda \in [0,1]$. Therefore, $\lambda F_{1} + (1-\lambda) F_{2} \in \Pi(\lambda \nu_{1} + (1-\lambda) \nu_{2},\nu_{k})$ and $D_{KL}(\lambda F_{1} + (1-\lambda) F_{2}\|F_{0}) \leq \delta$. Thus, $\mF_{\mN,\text{Relaxed}}$ is convex.
    \item The proof is identical to the proof of \Cref{lemma: Hausdorff space}.
\end{enumerate}
\end{proof}

\subsection{Proofs in \Cref{sec: Methodology}}

\subsubsection{Proof of Theorems \ref{thm: Strong Duality} and \ref{thm: Strong Duality smallest delta}} \label{sec: proof of strong duality}

\noindent We only prove \Cref{thm: Strong Duality}. The proofs for Theorem \ref{thm: Strong Duality smallest delta} is similar. First, we prove the minimax part. The Lagrangian of the \blueref{Primal} problem is:
\begin{equation*}
    \kappa(\delta,P) = \inf_{\substack{(\theta,v) \in \Theta \times \mathcal{V} \\ F \in \Pi(\nu_{1},\ldots,\nu_{k})}} \sup_{\substack{\lambda \in \bR^{d_{P}} \\ \lambda_{KL} \geq 0, g \in \mG}} \bE_{F} \left[c(U;\theta,v,g,\lambda)\right] + \lambda_{KL}D_{KL}(F\|F_{0}) - \lambda_{KL} \delta - \lambda^{T}P
\end{equation*}

\begin{lemma}[Minimax] \label{theorem: Minimax Theory}
    Under \Cref{assumption: Strong Duality}, we have:
    \begin{equation*}
       \kappa(\delta,P) = \inf_{(\theta,v) \in \Theta \times \mV} \sup_{\substack{\lambda \in \bR^{d_{P}} \\ \lambda_{KL} \geq 0, g \in \mG}} \inf_{F \in \Pi(\nu_{1},\ldots,\nu_{k})} \bE_{F} \left[c(U;\theta,v,g,\lambda)\right] + \lambda_{KL}D_{KL}(F\|F_{0}) - \lambda_{KL} \delta - \lambda^{T}P
    \end{equation*}
\end{lemma}
\begin{proof}
    In the proof, we show that the minimax theorem holds for given $(\theta,v)$ by verifying the conditions in \Cref{lemma: Minimax Theory}. For notational simplicity, define:
    \begin{equation*}
        \mL(F,g,\lambda,\lambda_{KL}) := \bE_{F} \left[c(U;\theta,v,g,\lambda)\right] + \lambda_{KL}D_{KL}(F\|F_{0}) - \lambda_{KL} \delta - \lambda^{T}P
    \end{equation*}
    \begin{itemize}
        \item \textbf{Compactness:} By \Cref{lemma: compact convex space Pi}, $\Pi(\nu_{1},\cdots,\nu_{k})$ is compact.
        \item \textbf{Hausdorff:} By \Cref{lemma: Hausdorff space Pi}, $\Pi(\nu_{1},\cdots,\nu_{k})$ is Hausdorff.
        \item \textbf{Concavelike:} Note that $\mL(F,g,\lambda,\lambda_{KL})$ is linear in $(\lambda,\lambda_{KL},g)$. Therefore, the concavelike condition is satisfied.
        \item \textbf{Convexlike:} By \Cref{lemma: compact convex space Pi}, $\Pi(\nu_{1},\cdots,\nu_{k})$ is a convex space. Therefore, convex combinations of elements in $\Pi(\nu_{1},\cdots,\nu_{k})$ are also in $\Pi(\nu_{1},\cdots,\nu_{k})$. Since $D_{KL}(F\|F_{0})$ is jointly convex (see \cite{nutz2021introduction} Lemma 1.3), and the expectation is linear in $F$, the convexlike condition also holds. 
        \item \textbf{Lower-semicontinuity:} Let $ h(U):= -(1+\|\lambda\|_{1}) C_{\theta,v,g}(1 + d(U,\hat{U}))$. Under \Cref{assumption: growth rate}, we have $h \in L^{1}(F)$ for all $F \in \Pi(\nu_{1},\cdots,\nu_{k})$. Therefore, for given $(\lambda,\lambda_{KL},g)$, $\bE_{F}\left[c(U;\theta,v,g,\lambda)\right]$ is lower-semicontinuous in $F$ by \cite{villani2009optimal} Lemma 4.3. By \cite{nutz2021introduction} Lemma 1.3, $D_{KL}(F\|F_{0})$ is lower-semicontinuous in $F$. By the superadditivity of $\liminf$, $\mL(F,g,\lambda,\lambda_{KL})$ is also lower-semicontinuous in $F$ for given $(\lambda,\lambda_{KL},g)$.
    \end{itemize}
\end{proof}

\begin{lemma} \label{lemma: convert arbitrary reference to product measure}
    Define:
    \begin{equation*}
        \mL(\delta,\theta,v,g,\lambda) := \sup_{\lambda_{KL} \geq 0} \inf_{F \in \Pi(\nu_{1},\ldots,\nu_{k})} \bE_{F} \left[c(U;\theta,v,g,\lambda)\right] + \lambda_{KL}D_{KL}(F\|F_{0}) - \lambda_{KL} \delta - \lambda^{T}P
    \end{equation*}
    Under \Cref{assumption: Strong Duality}, we have:
    \begin{equation*}
        \mL(\delta,\theta,v,g,\lambda) = \sup_{\lambda_{KL} \geq 0} \inf_{F \in \Pi(\nu_{1},\cdots,\nu_{k})} \bE_{F} \left[c(U;\theta,v,g,\lambda) + \lambda_{KL} \rho(U)\right] + \lambda_{KL} (D_{KL}(F\|F_{\otimes}) - \delta) - \lambda^{T}P
    \end{equation*}
\end{lemma}
\begin{proof}
By \Cref{assumption: reference}, we have:
\begin{equation*}
    D_{KL}(F\|F_{0}) = \int \log(\frac{dF}{dF_{\otimes}} \frac{dF_{\otimes}}{dF_{0}})dF = \int \log(\frac{dF}{dF_{\otimes}})dF + \int \log(\frac{dF_{\otimes}}{dF_{0}})dF = D_{KL}(F\|F_{\otimes}) + \int \rho(U)dF
\end{equation*}
Moreover, by \Cref{assumption: growth rate}, we have $c(U;\theta,v,g,\lambda), \rho \in L^{1}(F)$ for $\forall \ F \in \Pi(\nu_{1},\cdots,\nu_{k})$. Therefore, the sum of two expectations is also well-defined.
\end{proof}

\begin{lemma}
    Suppose \Cref{assumption: Strong Duality} holds. If $c(U;\theta,v,g,\lambda)$ and $\rho(U)$ are continuous in $U$, then optimizing over $\lambda_{KL} > 0$ gives the same value as optimizing over $\lambda_{KL} \geq 0$, i.e.,
    \begin{equation*}
        \mL(\delta,\theta,v,g,\lambda) = \sup_{\lambda_{KL} > 0} \inf_{F \in \Pi(\nu_{1},\cdots,\nu_{k})} \bE_{F} \left[c(U;\theta,v,g,\lambda) + \lambda_{KL} \rho(U)\right] + \lambda_{KL} (D_{KL}(F\|F_{\otimes}) - \delta) - \lambda^{T}P
    \end{equation*}
\end{lemma}
\begin{proof}
    By \cite{eckstein2023convergence} Proposition 3.1, we have:
    \begin{equation*}
        \lim_{\lambda_{KL} \downarrow 0} \inf_{F \in \Pi(\nu_{1},\cdots,\nu_{k})} \bE_{F} \left[c(U;\theta,v,g,\lambda) + \lambda_{KL} \rho(U)\right] + \lambda_{KL} (D_{KL}(F\|F_{\otimes}) - \delta) = \inf_{F \in \Pi(\nu_{1},\cdots,\nu_{k})}\bE_{F}\left[c(U;\theta,v,g,\lambda)\right]
    \end{equation*}
    Therefore, for any sequence $\lambda_{KL}^{i} \to 0$ as $i \to \infty$, we have:
    \begin{equation*}
        \lim_{i \to \infty} \inf_{F \in \Pi(\nu_{1},\cdots,\nu_{k})} \bE_{F} \left[c(U;\theta,v,g,\lambda) + \lambda_{KL} \rho(U)\right] + \lambda_{KL}^{i} (D_{KL}(F\|F_{\otimes}) - \delta) = \inf_{F \in \Pi(\nu_{1},\cdots,\nu_{k})}\bE_{F}\left[c(U;\theta,v,g,\lambda)\right]
    \end{equation*}
\end{proof}
Then, we show the EOT duality part. For $\lambda_{KL} > 0$, consider the following EOT problem:
\begin{equation*}
    \mC(\theta,v,g,\lambda,\lambda_{KL}) := \inf_{F \in \Pi(\nu_{1},\cdots,\nu_{k})} \bE_{F} \left[c(U;\theta,v,g,\lambda) + \lambda_{KL}\rho(U)\right] + \lambda_{KL} D_{KL}(F\|F_{\otimes}) 
\end{equation*}

\begin{theorem}[EOT Duality] \label{thm: Strong Duality for EOT}
    Under \Cref{assumption: Strong Duality}, for $\lambda_{KL} > 0$, the following holds:
    \begin{align*}
        \mC(\theta,v,g,\lambda,\lambda_{KL}) = \sup_{\left\{\phi_{i} \in L^{1}(\nu_{i})\right\}_{i=1}^{k}} \sum_{i=1}^{k} \bE_{\nu_{i}} \phi_{i}(U_{i}) - \lambda_{KL} \bE_{F_{0}} \exp(\frac{\sum_{i=1}^{k} \phi_{i}(U_{i}) - c(U;\theta,v,g,\lambda)}{\lambda_{KL}}) + \lambda_{KL}
    \end{align*}
    Moreover, the worst-case distribution is given by:
    \begin{equation*}
        \frac{dF^{*}(U)}{dF_{0}(U)} = \exp(\frac{\sum_{i=1}^{k} \phi_{i}^{*}(U_{i}) - c(U;\theta,v,g,\lambda)}{\lambda_{KL}}) \quad F_{0}\text{-a.s.}
    \end{equation*}
    where $\{\phi^{*}_{i}\}_{i=1}^{k}$ are unique maximizers up to additive constants $F_{0}$-almost surely.
\end{theorem}
\begin{proof}
    Under \Cref{assumption: growth rate}, $c(U;\theta,v,g,\lambda), \rho(U) \in L^{1}(F_{\otimes})$. For two marginals ($k=2$), see \cite{nutz2021introduction} Theorem 4.7 and Remark 4.8(a). The results generalize directly to the multi-marginal case (see \cite{nutz2022entropic} Section 6). Their results hold $F_{\otimes}$-a.s. Since $F_{0} \ll F_{\otimes}$, the results also hold $F_{0}$-a.s. Moreover, note that:
    \begin{equation*}
        \bE_{F_{\otimes}} \exp(\frac{\sum_{i=1}^{k} \phi_{i}(U_{i}) - c(U;\theta,v,g,\lambda)-\lambda_{KL}\rho(U)}{\lambda_{KL}}) = \bE_{F_{0}} \exp(\frac{\sum_{i=1}^{k} \phi_{i}(U_{i}) - c(U;\theta,v,g,\lambda)}{\lambda_{KL}})
    \end{equation*}
\end{proof}

\subsubsection{Proof of \Cref{thm: Properties of the optimal test function}}

\begin{proof}[\textbf{Proof of \Cref{thm: Properties of the optimal test function}}]
    As showed in \cite{nutz2021introduction} equation 4.11, the fact that $F^{*}$ in \Cref{thm: Strong Duality for EOT} is a probability measure with marginals $(\nu_{1},\cdots,\nu_{k})$ implies the Schr\"odinger equations: for $\forall \ j \leq k$:
    \begin{align*}
        & \phi_{i}^{*}(U_{j}) = - \lambda_{KL} \log \bE_{F_{\otimes,-j}}\exp(\frac{\sum_{i \neq j} \phi_{i}^{*}(U_{i}) - c(U;\theta,v,g,\lambda)}{\lambda_{KL}}) \quad \nu_{j}\text{-a.s.}
    \end{align*}
    Since $c$ is $k$-times continuously differentiable in $U$, the right-hand side is also $k$-times continuously differentiable in $U_{j}$. Therefore, $\phi_{j}^{*}(U_{j})$ is $k$-times continuously differentiable in $U_{j}$ for $\forall \ j \leq k$.
\end{proof}

\subsubsection{Proof of \Cref{thm: stationary perturbation duality}}

\begin{proof}[\textbf{Proof of Theorem \ref{thm: stationary perturbation duality}}]
    In the proof, we first swap the order of infimum over $F$ and the supremum over $(\lambda,\lambda_{KL},g)$ by the minimax theorem. Then, we swap the order of infimum over $\nu$ and the supremum over $(\phi_{1},\phi_{2})$ by the minimax theorem. Define:
    \begin{equation*}
        \mL(F,g,\lambda,\lambda_{KL}) := \bE_{F} \left[c(\xi,\xi';\theta,v,g,\lambda)\right] + \lambda_{KL}D_{KL}(F\|F_{0}) - \lambda_{KL} \delta - \lambda^{T}P
    \end{equation*}
    \begin{itemize}
        \item \textbf{Compactness and Hausdorff:} By \Cref{lemma: stationary perturbation set}, $\left\{ F \in \mP(\Xi^{2}) \mid F \in \Pi(\nu,\nu), \nu \in \mN\right\}$ is compact and Hausdorff.
        \item \textbf{Concavelike:} Note that $\mL(F,g,\lambda,\lambda_{KL})$ is linear in $(\lambda,\lambda_{KL},g)$. Therefore, the concavelike condition is satisfied.
        \item \textbf{Convexlike:} By \Cref{lemma: stationary perturbation set}, $\left\{ F \in \mP(\Xi^{2}) \mid F \in \Pi(\nu,\nu), \nu \in \mN\right\}$ is convex. Since $D_{KL}(F\|F_{0})$ is jointly convex (see \cite{nutz2021introduction} Lemma 1.3), and the expectation is linear in $F$, the convexlike condition also holds.
        \item \textbf{Lower-semicontinuity:} Let $ h(\xi,\xi'):= -(1+\|\lambda\|_{1}) C_{\theta,v,g}(1 + d(\xi,\xi',\hat{\xi},\hat{\xi}'))$. Under \Cref{assumption: growth rate} and as $\Xi$ is compact, we have $h \in L^{\infty}(\Xi^{2})$. Therefore, for given $(\lambda,\lambda_{KL},g)$, $\bE_{F}\left[c(\xi,\xi';\theta,v,g,\lambda)\right]$ is lower-semicontinuous in $F$ by \cite{villani2009optimal} Lemma 4.3. By \cite{nutz2021introduction} Lemma 1.3, $D_{KL}(F\|F_{0})$ is lower-semicontinuous in $F$. By the superadditivity of $\liminf$, $\mL(F,g,\lambda,\lambda_{KL})$ is also lower-semicontinuous in $F$ for given $(\lambda,\lambda_{KL},g)$.
    \end{itemize}
    Therefore, by \Cref{lemma: Minimax Theory} and the EOT duality in \Cref{thm: Strong Duality}, we have:
    \begin{align*}
        & \kappa_{stationary}(\delta_{1},\delta,P) = \inf_{(\theta,v) \in \Theta \times \mV} \sup_{\lambda \in \bR^{d_{P}}, \lambda_{KL} \geq 0, g \in \mG} \inf_{\nu \in \mN} \ \mC(\theta,v,g,\lambda,\lambda_{KL},\nu) - \lambda_{KL} \delta - \lambda^{T}P \\
        & \mC(\theta,v,g,\lambda,\lambda_{KL},\nu) = \sup_{\phi_{1},\phi_{2} \in L^{1}(\nu)} \bE_{\nu} \left[\phi_{1}(\xi) + \phi_{2}(\xi')\right] - \lambda_{KL} \bE_{F_{0}} \exp\left(\frac{\phi_{1}(\xi) + \phi_{2}(\xi') - c(\xi,\xi';\theta,v,g,\lambda)}{\lambda_{KL}}\right) + \lambda_{KL}
    \end{align*}
    Furthermore, as $\Xi$ is compact, we can replace $L^{1}(\nu)$ with $L^{\infty}(\Xi)$. Moreover, $\phi_{1}^{*}$ and $\phi_{2}^{*}$ are bounded (see \cite{nutz2021introduction} Lemma 4.9) Moreover, by the lower semicontinuity of $c(\xi,\xi';\theta,v,g,\lambda)$, the Fatou's lemma, and the proof of \Cref{thm: Strong Duality for EOT}, $\phi_{1}^{*}$ and $\phi_{2}^{*}$ are lower semicontinuous.
    
    Next, we swap the order of infimum over $\nu$ and the supremum over $(\phi_{1},\phi_{2})$:
    \begin{itemize}
        \item \textbf{Compactness and Hausdorff:} By \Cref{Lemma: compact space KL}, $\mN$ is compact. Note that metrizable spaces are Hausdorff.
        \item \textbf{Convexlike:} As the expectation is linear in $\nu$ and $\mN$ is convex, the convexlike condition is satisfied.
        \item \textbf{Concavelike:} It is straightforward to see that the objective function is concave in $(\phi_{1},\phi_{2})$ as $L^{\infty}(\Xi)$ is convex.
        \item \textbf{Lower-semicontinuity:} For given $(\phi_{1},\phi_{2})$, $\bE_{\nu} \left[\phi_{1}(\xi) + \phi_{2}(\xi')\right]$ is lower-semicontinuous in $\nu$ by \cite{villani2009optimal} Lemma 4.3 as $\phi_{1}$ and $\phi_{2}$ are bounded.
    \end{itemize}
    The KL divergence DRO duality holds by \cite{hu2013kullback} Theorem 1.
\end{proof}

\subsection{Proofs in \Cref{sec:extension nonstationary}}

\noindent We only prove \Cref{theorem: nonstationary duality}. The proof for \Cref{thm: Strong Duality smallest delta nonstationary} is the same.

\subsubsection{Proof of Theorem \ref{theorem: nonstationary duality}}

\begin{proof}[\textbf{Proof of Theorem \ref{theorem: nonstationary duality}}]
    \begin{enumerate}
        \item \textbf{Minimax Part:} The proof of minimax part is similar to \Cref{theorem: Minimax Theory}. Convexity and compactness of $\Pi(\nu_{1},\nu_{T})$ are given by Lemmas \ref{lemma: convex space Pi2} and \ref{lemma: compact space Pi2}. Moreover, the bounds on the cost function is replaced by the finite constant in \Cref{assumption: bounded1}. Finally, \Cref{lemma: Hausdorff space Pi} is replaced by \Cref{lemma: Hausdorff space Pi2}.
        \item \textbf{Duality Part:} For notational simplicity, let $c(U) := c(U;\theta,v,g,\lambda,\lambda_{s})$. Then, \blueref{eq: nonstationary EOT} can be rewritten as:
        \begin{equation} \label{eq: nonstationary EOT2}
            \inf_{F \in \Pi(\nu_{1},\nu_{T})} \bE_{F} \left[c(U)\right] + \lambda_{KL}D_{KL}(F\|F_{0})
        \end{equation}

        By \Cref{assumption: bounded1}, $\exp(\frac{-c(U)}{\lambda_{KL}})$ is bounded and in particular, $\exp(\frac{-c(U)}{\lambda_{KL}}) \in L^{1}(F_{0})$. Therefore, we can define the auxiliary reference measure $R$ as follows:
        \begin{equation*}
            dR(U) := \exp(\frac{-c(U)}{\lambda_{KL}}) dF_{0}(U)
        \end{equation*}
        Note that $R \sim F_{0}$. Then, \eqref{eq: nonstationary EOT2} is equivalent to:\footnote{This is the Schr{\"o}dinger bridge problem, see \cite{leonard2013survey} for continuous time setting and \cite{de2021diffusion} for discrete time setting.}
        \begin{equation}
            \inf_{F \in \Pi(\nu_{1},\nu_{T})} \lambda_{KL} D_{KL}(F\|R) \label{eq: nonstationary EOT3}
        \end{equation}
        By \cite{leonard2014some} Theorem 2.4, we have:
        \begin{equation*}
            D_{KL}(F\|R) = D_{KL}(F_{1,T}\|R_{1,T}) + \bE_{F_{1,T}} \left[D_{KL}(F_{|1,T}\|R_{|1,T})\right]
        \end{equation*}
        where $F_{1,T}$ and $R_{1,T}$ are the two-period marginals of $F$ and $R$, respectively. $F_{|1,T},R_{|1,T}$ are the conditional distribution given $(\xi_{1},\xi_{T})$. In particular, we have:
        \begin{equation*}
            dR_{1,T}(\xi_{1},\xi_{T}) := \int_{ \xi_{2} \cdots \xi_{T-1}} dR(\xi_{1},\cdots,\xi_{T}) = \int_{\xi_2, \dots, \xi_{T-1}} \exp\left(\frac{-c(\xi_1, \dots, \xi_T)}{\lambda_{KL}}\right) \, dF_{0}(\xi_1, \dots, \xi_T)
        \end{equation*}
        The second term is minimized at $F_{|1,T}^{*} = R_{|1,T}$. Therefore, \eqref{eq: nonstationary EOT3} can be reduced to the static Schr{\"o}dinger bridge problem:
        \begin{equation}
            \inf_{F_{1,T} \in \tilde{\Pi}(\nu_{1},\nu_{T})} D_{KL}(F_{1,T}\|R_{1,T}) \label{eq: static SB problem}
        \end{equation}     
        where $\tilde{\Pi}(\nu_{1},\nu_{T})$ is the set of all joint distributions of $(\xi_{1},\xi_{T})$ with marginals $\nu_{1}$ and $\nu_{T}$.

        Note that:
        \begin{equation*}
            dR_{1,T}(\xi_{1},\xi_{T}) = \left(\int_{\xi_2, \dots, \xi_{T-1}} \exp\left(\frac{-c(\xi_1, \dots, \xi_T)}{\lambda_{KL}}\right) \, dF_0(\xi_2, \dots, \xi_{T-1}|\xi_{1},\xi_{T})\right) dF_{0}^{1,T}(\xi_{1},\xi_{T})
        \end{equation*}
        By \Cref{assumption: bounded1}, we have $R_{1,T} \sim F_{0}^{1,T}$. By \Cref{assumption: reference1}, we have $R_{1,T} \sim \nu_{1} \otimes \nu_{T}$. Therefore, $\Pi_{\text{fin}}(\nu_{1},\nu_{T}):= \{F \in \tilde{\Pi}(\nu_{1},\nu_{T}) \mid D_{KL}(F\|R_{1,T}) < + \infty\} \neq \emptyset$. By \Cref{assumption: reference1} and \cite{nutz2021introduction} Theorem 2.1, the unique solution to \eqref{eq: static SB problem} has the form:
        \begin{equation*}
            \frac{dF_{1,T}^{*}(\xi_{1},\xi_{T})}{dR_{1,T}(\xi_{1},\xi_{T})} = \exp\left(\phi_{1}^{*}(\xi_{1}) + \phi_{T}^{*}(\xi_{T})\right) \quad R_{1,T} \text{-a.s.}
        \end{equation*}
        and $\phi_{1}^{*} \in L^{1}(\nu_{1})$, $\phi_{T}^{*} \in L^{1}(\nu_{T})$. By \cite{nutz2021introduction} Theorem 3.2, we have:
        \begin{equation*}
            \inf_{F_{1,T} \in \tilde{\Pi}(\nu_{1},\nu_{T})} D_{KL}(F_{1,T}\|R_{1,T}) = \sup_{\phi_{1} \in L^{1}(\nu_{1}),\phi_{T} \in L^{1}(\nu_{T})}\bE_{\nu_{1}} \phi_{1} + \bE_{\nu_{T}} \phi_{T} - \int \exp\left(\phi_{1}+\phi_{T}\right)dR_{1,T} + 1
        \end{equation*}
        Multiplying both sides by $\lambda_{KL}$ and letting $\phi_{1} := \lambda_{KL} \phi_{1}$ and $\phi_{T} := \lambda_{KL} \phi_{T}$, we have:
        \begin{equation*}
           \inf_{F_{1,T} \in \tilde{\Pi}(\nu_{1},\nu_{T})} \lambda_{KL}D_{KL}(F_{1,T}\|R_{1,T}) = \sup_{\phi_{1} \in L^{1}(\nu_{1}),\phi_{T} \in L^{1}(\nu_{T})}\bE_{\nu_{1}} \phi_{1} + \bE_{\nu_{T}} \phi_{T} - \lambda_{KL} \bE_{R_{1,T}} \exp\left(\frac{\phi_{1}+\phi_{T}}{\lambda_{KL}}\right) + \lambda_{KL}
        \end{equation*}
        Then, we have:
        \begin{align*}
            dF^{*}(U) = dF_{1,T}^{*} dF_{|1,T}^{*}  = \exp\left(\frac{\phi_{1}^{*}(\xi_{1}) + \phi_{T}^{*}(\xi_{T}) - c(U)}{\lambda_{KL}}\right) dF_{0}(U)
        \end{align*}
        where $\phi_{1}^{*}$ and $\phi_{T}^{*}$ are the unique maximizers up to an additive constant.
        \item \textbf{Markov Property:}
        \Cref{assumption: nonstationary Markov} implies that $c(U)$ is also pairwise additive, i.e., $c(U) = \sum_{t=1}^{T-1} c_{t}(\xi_{t},\xi_{t+1})$ for some $c_{t}(\xi_{t},\xi_{t+1})$. Therefore, we can write:
        \begin{equation*}
            dF^{*}(U) = \exp\left(\sum_{t=1}^{T-1} \frac{-c_{t}(\xi_{t},\xi_{t+1})}{\lambda_{KL}} + \frac{\phi_{1}^{*}(\xi_{1}) + \phi_{T}^{*}(\xi_{T})}{\lambda_{KL}}\right) dF_{0}(U)
        \end{equation*}
        which has Markov property as $F_{0}$ has Markov property.\footnote{The left part is the (unnormalized) pairwise Markov random field \cite{wainwright2008graphical}.} 

        \item Since $\kappa_{\text{TI}}(\delta,P) \geq \tilde{\kappa}_{\text{TI}}(\delta,P)$ and the solution to $\tilde{\kappa}_{\text{TI}}(\delta,P)$ corresponding to $\lambda_{KL}^{*} > 0$ has the Markov property, we have: $\kappa_{\text{TI}}(\delta,P) = \tilde{\kappa}_{\text{TI}}(\delta,P)$.
    \end{enumerate}
\end{proof}

\subsubsection{Proof of \Cref{thm: nonstationary duality with perturbation}}

\begin{proof}[\textbf{Proof of \Cref{thm: nonstationary duality with perturbation}}]
    The proof of minimax part is similar to \Cref{theorem: Minimax Theory}. \Cref{lemma: compact convex space Pi} is replaced by Lemmas \ref{lemma: compact space Pi3} and \ref{lemma: convex space Pi3}. \Cref{lemma: Hausdorff space Pi} is replaced by \Cref{lemma: Hausdorff space Pi3}. Finally, the bounds on the cost function is replaced by the finite constant in \Cref{assumption: bounded1}.
\end{proof}

\subsubsection{Proof of \Cref{lemma: convexity of EOT}}

\begin{proof}[\textbf{Proof of \Cref{lemma: convexity of EOT}}]
    Let $\pi_{1}, \pi_{2}$ be the unique solutions to the EOT problem with respect to $\nu_{1}, \nu_{2}$, respectively. Then, $\pi := \lambda \pi_{1} + (1-\lambda) \pi_{2} \in \tilde{\Pi}(\lambda \nu_{1} + (1-\lambda) \nu_{2},\nu_{T})$ for all $\lambda \in [0,1]$. 
    By \cite{nutz2021introduction} Lemma 1.3, the KL divergence is jointly convex. Therefore, we have:
    \begin{equation*}
        D_{KL}(\pi\|(\lambda \nu_{1} + (1-\lambda) \nu_{2})\otimes \nu_{T}) \leq \lambda D_{KL}(\pi_{1}\|\nu_{1} \otimes \nu_{T}) + (1-\lambda) D_{KL}(\pi_{2}\|\nu_{2} \otimes \nu_{T})
    \end{equation*}
    Therefore, we have:
    \begin{align*}
        \text{EOT}(\lambda \nu_{1} + (1-\lambda) \nu_{2},\nu_{T}) & \leq \int \tilde{c}(\xi_{1},\xi_{T}) d\pi(\xi_{1},\xi_{T}) + \lambda_{KL} D_{KL}(\pi\|(\lambda \nu_{1} + (1-\lambda) \nu_{2}) \otimes \nu_{T}) \\
        & \leq \lambda \int \tilde{c}(\xi_{1},\xi_{T}) d\pi_{1}(\xi_{1},\xi_{T}) + (1-\lambda) \int \tilde{c}(\xi_{1},\xi_{T}) d\pi_{2}(\xi_{1},\xi_{T}) \\
        & \hspace{2cm} + \lambda_{KL} \left(\lambda D_{KL}(\pi_{1}\|\nu_{1} \otimes \nu_{T}) + (1-\lambda) D_{KL}(\pi_{2}\|\nu_{2} \otimes \nu_{T})\right) \\
        & = \lambda \text{EOT}(\nu_{1},\nu_{T}) + (1-\lambda) \text{EOT}(\nu_{2},\nu_{T})
    \end{align*}
    
    See \cite{goldfeld2024statistical} Lemma E.23 for the directional derivative.
\end{proof}

\subsection{Proofs in \Cref{sec: Large Sample Properties and Inference}}

\subsubsection{Proof of \Cref{lemma: identified and estimator are compact}}

\begin{proof}[\textbf{Proof of \Cref{lemma: identified and estimator are compact}}]
    By the continuity of $\bE_{F} \left[m(U;\theta,v(\alpha))\right]$, $\mA_{I}, \hat{\mA_{I}}$ are closed. By \cite{rudin1976principles} Theorem 2.35, a closed subset of a compact set is compact. Therefore, $\mA_{I}, \hat{\mA_{I}}$ are compact as $\mA$ is compact (Lemmas \ref{lemma: compact space} and \ref{lemma: relaxed compact space}). For the nonempty part, see the proof of \Cref{thm: Consistency and Convergence Rate of the Identified Set}.
\end{proof}

\subsubsection{Proof of \Cref{thm: Consistency and Convergence Rate of the Identified Set}}

\begin{proof}[\textbf{Proof of \Cref{thm: Consistency}}]
    \begin{enumerate}
        \item Since $\epsilon_{n} \geq \|P_{0} - P_{n}\|_{\infty}$, we have $\mA_{I} \subseteq \hat{\mA_{I}}$.
        \item If $\mA_{I} = \mA$, the result is trivial.
        \item Suppose $\mA_{I} \neq \mA$. By \Cref{assumption: Asymptotic Property continuous}, we have for some $\delta(\eps) > 0$:
        \begin{equation*}
            \inf_{d(\alpha,\mA_{I}) > \eps} \|\bE_{F} \left[m(U;\theta,v(\alpha))\right] - P_{n}\|_{\infty} \geq \inf_{d(\alpha,\mA_{I}) > \eps} \|\bE_{F} \left[m(U;\theta,v(\alpha))\right] - P_{0} + o_{p}(1)\|_{\infty} \geq \delta(\eps) + o_{p}(1)
        \end{equation*}
        
        Similarly, we have:
        \begin{equation*}
            \sup_{\hat{\mA_{I}}} \|\bE_{F} \left[m(U;\theta,v(\alpha))\right] - P_{0}\|_{\infty} \leq \sup_{\hat{\mA_{I}}} \|\bE_{F} \left[m(U;\theta,v(\alpha))\right] - P_{n} + o_{p}(1)\|_{\infty} \leq \frac{c_{n}}{\sqrt{n}} + o_{p}(1) = o_{p}(1)
        \end{equation*}

        Therefore, with probability approaching 1:
        \begin{equation*}
            \sup_{\hat{\mA_{I}}} \|\bE_{F} \left[m(U;\theta,v(\alpha))\right] - P_{0}\|_{\infty} < \delta(\eps) \leq \inf_{d(\alpha,\mA_{I}) > \eps} \|\bE_{F} \left[m(U;\theta,v(\alpha))\right] - P_{n}\|_{\infty}
        \end{equation*}
        which implies that: with probability approaching 1, $\hat{\mA_{I}} \bigcap \{\alpha : d(\alpha,\mA_{I}) > \eps\} = \emptyset$. Thus, $\hat{\mA_{I}} \subseteq \{\alpha : d(\alpha,\mA_{I}) \leq \eps\}$ with probability approaching 1. Therefore, we have with probability approaching 1, $d_{H}(\hat{\mA}_{I},\mA_{I}) \leq \eps$. As $\eps$ is arbitrary, $d_{H}(\hat{\mA}_{I},\mA_{I}) = o_{p}(1)$.
\end{enumerate}
\end{proof}

\begin{lemma}[Existence of a Polynomial Minorant] \label{lemma: Moment Polynomial Minorant}
    Under Assumptions \ref{assumption: Asymptotic Properties1}, \ref{assumption: Asymptotic Properties2}, and \ref{assumption: Polynomial Minorant}, we have for $\forall \ \eps \in (0,1)$ there exists $(\kappa_{\eps},n_{\eps})$ such that for all $n \geq n_{\eps}$, we have:
    \begin{equation*}
        \|\bE_{F} \left[m(U;\theta,v(\alpha))\right] - P_{n}\|_{\infty} \geq \frac{C_{1}}{2} \min \{C_{2}, d(\alpha, \mA_{I})\}
    \end{equation*}
    uniformly on $\{\alpha | d(\alpha, \mA_{I}) \geq \frac{\kappa_{\eps}}{\sqrt{n}}\}$ with probability at least $1-\eps$.
\end{lemma}
\begin{proof}
    Note that:
    \begin{align*}
        \sqrt{n}\|\bE_{F} \left[m(U;\theta,v(\alpha))\right] - P_{n}\|_{\infty} 
        & = \| \sqrt{n}(\bE_{F} \left[m(U;\theta,v(\alpha))\right] - P_{0}) + \sqrt{n}(P_{0} - P_{n})\|_{\infty} \\
        & \geq \left|\|\sqrt{n}(\bE_{F} \left[m(U;\theta,v(\alpha))\right] - P_{0})\|_{\infty} - \|\sqrt{n}(P_{0} - P_{n})\|_{\infty} \right| \\
        & = \left|\|\sqrt{n}(\bE_{F} \left[m(U;\theta,v(\alpha))\right] - P_{0})\|_{\infty} - O_{p}(1) \right|
    \end{align*}
    where we used $\|x + y\| \geq |\|x\| - \|y\||$ and $\sqrt{n}(P_{0} - P_{n}) = O_{p}(1)$. Therefore, for $\forall \ \eps$ there exists $M_{\eps} > 0$ and $n_{\eps,1}$ such that for all $n \geq n_{\eps,1}$, with probability at least $1-\eps$: $|O_{p}(1)| \leq M_{\eps}$. Choose $(\kappa_{\eps},n_{\eps})$ such that $n_{\eps} \geq n_{\eps,1}$, $C_{1} \kappa_{\eps} \geq 2M_{\eps}$, and $\sqrt{n} C_{1} C_{2} \geq 2M_{\eps}$ for all $n \geq n_{\eps}$. Then, with probability at least $1-\eps$: uniformly on $\{ \alpha: d(\alpha, \mA_{I}) \geq \frac{\kappa_{\eps}}{\sqrt{n}}\}$:
    \begin{align*}
        \|\sqrt{n}(\bE_{F} \left[m(U;\theta,v(\alpha))\right] - P_{0})\|_{\infty} - O_{p}(1)
        & \geq \|\sqrt{n}(\bE_{F} \left[m(U;\theta,v(\alpha))\right] - P_{0})\|_{\infty} - M_{\eps} \\
        & \geq \sqrt{n} C_{1} \min\{C_{2}, d(\alpha, \mA_{I})\} - M_{\eps} \\
        & \geq \frac{1}{2} \sqrt{n} C_{1} \min\{C_{2}, d(\alpha, \mA_{I})\} + \frac{1}{2} \sqrt{n} C_{1} \min \{C_{2}, \frac{\kappa_{\eps}}{\sqrt{n}}\} - M_{\eps} \\
        & \geq \frac{1}{2} \sqrt{n} C_{1} \min\{C_{2}, d(\alpha, \mA_{I})\}
    \end{align*}
    Therefore, with probability at least $1-\eps$: uniformly on $\{ \alpha: d(\alpha, \mA_{I}) \geq \frac{\kappa_{\eps}}{\sqrt{n}}\}$:
    \begin{equation*}
        \|\bE_{F} \left[m(U;\theta,v(\alpha))\right] - P_{n}\|_{\infty} \geq \frac{1}{2} C_{1} \min\{C_{2}, d(\alpha, \mA_{I})\}
    \end{equation*}
\end{proof}

\begin{proof}[\textbf{Proof of \Cref{thm: Convergence Rate of the Identified Set}}]
    For $\forall \ \eps > 0$, let the positive constants $(\kappa_{\eps},n_{\eps})$ be as specified in \Cref{lemma: Moment Polynomial Minorant}. Let $\bar{c} := \max \{\frac{C_{1}}{2} \kappa_{\eps}, c_{n}\}$. Furthermore, there exists $n_{\eps'} \geq n_{\eps}$ such that for all $n \geq n_{\eps'}$, with probability at least $1-\eps$: $\eps_{n} := \frac{\bar{c}}{\frac{C_{1}}{2}\sqrt{n}} \leq C_{2}$ as $\frac{c_{n}}{\sqrt{n}} = o_{p}(1)$, and $\eps_{n} \geq \frac{\kappa_{\eps}}{\sqrt{n}}$. Therefore, with probability at least $1-\eps$, for all $n \geq n_{\eps'}$:
    \begin{align*}
        \inf_{\alpha, d(\alpha, \mA_{I}) \geq \eps_{n}} \|\bE_{F} \left[m(U;\theta,v(\alpha))\right] - P_{n}\|_{\infty}
        & \geq \frac{C_{1}}{2} \min\{C_{2}, d(\alpha, \mA_{I})\} \geq \frac{C_{1}}{2} \min\{C_{2}, \eps_{n}\} \geq \frac{C_{1}}{2} \eps_{n} = \frac{\bar{c}}{\sqrt{n}} \geq \frac{c_{n}}{\sqrt{n}}
    \end{align*}
    Therefore, combining the first part in \Cref{thm: Consistency} we have: with probability at least $1-\eps$: $\mA_{I} \subseteq \hat{\mA}_{I} \subseteq \{\alpha | d(\alpha, \mA_{I}) \leq \eps_{n}\}$. Thus, $d_{H}(\hat{\mA}_{I},\mA_{I}) \leq \eps_{n}$. Therefore, we have: for $\forall \ \eps > 0$, there exists $n_{\eps'}$ such that for all $n \geq n_{\eps'}$, we have: with probability at least $1-\eps$: $d_{H}(\hat{\mA}_{I},\mA_{I}) \leq \eps_{n}$. As $\eps$ is arbitrary, we have $|\frac{d_{H}(\hat{\mA}_{I},\mA_{I})}{\frac{\max \{1, c_{n}\}}{\sqrt{n}}}| \leq \frac{\max \{\frac{C_{1}}{2} \kappa_{\eps}, c_{n}\}}{\max \{1, c_{n}\} \frac{C_{1}}{2}} := M_{1,\eps}$ with probability at least $1-\eps$. Therefore, we have $d_{H}(\hat{\mA}_{I},\mA_{I}) = O_{p}(\frac{\max \{1, c_{n}\}}{\sqrt{n}})$.
\end{proof}

\subsubsection{Proof of Theorem \ref{thm: Consistency and Convergence Rate of the Scalar of Interest}}

\begin{proof}
    It is a direct consequence of \Cref{lemma: consistency and convergence rate of the projection} and \Cref{thm: Consistency and Convergence Rate of the Identified Set}
\end{proof}

\subsubsection{Proof of \Cref{thm: Hadamard Directional Differentiability}} \label{appendix: proof of Hadamard Directional Differentiability}

\begin{proof}[\textbf{Proof of \Cref{thm: Hadamard Directional Differentiability}}]
    We verify conditions in \cite{bonnans2013perturbation} Theorem 4.25.
    
    We first verify the setting in \cite{bonnans2013perturbation} Page 260. By \cite{bogachev2007measure} Theorem 4.6.1, all real countably additive measures on $(\mU,\mB(\mU))$ with the variation norm is a Banach space. Therefore, the product of $\bR^{d_{\theta}}$ and all real countably additive measures with the variation norm plus the Euclidean norm is also a Banach space. Therefore, the setting is satisfied.
    
    Moreover, by Lemmas \ref{lemma: compact space}, \ref{lemma: convex space}, \ref{lemma: relaxed compact space}, and \ref{lemma: relaxed convex space}, the sets $\mF$ and $\mF_{\text{relaxed}}$ are closed and convex. By \Cref{assumption: compact}, $\Theta$ is also closed (as it is compact) and convex. Therefore, $\mA$ is closed and convex. According to \cite{bonnans2013perturbation} equation (2.194), the Robinson's constraint qualification (defined in their equation (2.163)) is equivalent to \Cref{assumption: constraint qualification}. In their notation, $G(\alpha,P) := P(\alpha) - P$ and $f(\alpha,P) = s(\alpha)$. Thus, the Robinson's constraint qualification holds.

    By \cite{bonnans2013perturbation} Theorem 4.9. The Robinson's constraint qualification implies the directional regularity condition for any direction.

    Therefore, by \cite{bonnans2013perturbation} Theorem 4.25, the maps $\kappa(\delta,P)$ and $\tilde{\kappa}_{\text{TI}}(\delta,P)$ are Hadamard directionally differentiable at $P_{0}$, and note that $D_{P}\mL(\alpha,\lambda,P_{0})h = -\lambda^{T}h$ where $\mL(\alpha,\lambda,P) := s(\alpha) + \lambda^{T} (P(\alpha) - P)$.

    The asymptotic distribution follows \cite{fang2019inference} Theorem 2.1.
\end{proof}

\subsection{Proofs in \Cref{sec: Derivative with respect to delta}}

\subsubsection{Proof of \Cref{thm: global sensitivity}}

\begin{proof}[\textbf{Proof of \Cref{thm: global sensitivity}}]
    The quantization rate is in \cite{eckstein2024convergence} Remark 2.1. Other parts follow directly from \cite{eckstein2024convergence} Theorem 3.1(i).
\end{proof}

\subsubsection{Proof of \Cref{thm: derivative delta}}

\begin{lemma} \label{lemma: compact identified set Lip}
    Under Assumptions \ref{assumption: Asymptotic Properties1} and \ref{assumption: compact Support for derivative}, $\mA_{I,Lip}^{\delta}$ is compact.
\end{lemma}
\begin{proof}
    By \Cref{assumption: compact Support for derivative}, $\Pi_{\text{TH}}$ and $\Pi_{\text{TI}}$ are both tight, and by Prokhorov's theorem they have compact closure. By passing to the limit in the equation for marginals, $\Pi_{\text{TH}}$ and $\Pi_{\text{TI}}$ are closed, which implies $\mA_{Lip}^{\delta} = \{F \in \mP(\mU) \mid F \in \Pi, D_{KL}(F\|F_{0}) \leq \delta\}$ is closed where $\Pi$ is either $\Pi_{\text{TH}}$ or $\Pi_{\text{TI}}$. By \Cref{lemma: compact space KL} and \cite{rudin1976principles} Theorem 2.35, $\mA_{Lip}^{\delta}$ is compact. By \Cref{assumption: Asymptotic Property continuous}, $\mA_{I,Lip}^{\delta}$ is closed. Therefore, it is compact.
\end{proof}

\begin{lemma} \label{lemma: derivative delta}
    Let $\Pi$ be either $\Pi_{\text{TH}}$ or $\Pi_{\text{TI}}$. Under \Cref{assumption: hadamard for delta}, $D_{KL}(F\|F_{0})$ is continuously differentiable on $\{F \in \Pi | D_{KL}(F\|F_{0}) \leq \delta\}$.
\end{lemma}
\begin{proof}
    The directional derivative\footnote{See \cite{nutz2021introduction} equation 1.10.} of $D_{KL}(F\|F_{0})$ in the direction $F_{1}$ is:
    \begin{equation*}
        D_{F}D_{KL}(F\|F_{0})(F_{1} - F) = - D_{KL}(F\|F_{0}) + \int \log \frac{dF}{dF_{0}} dF_{1}
    \end{equation*}
    which is linear in $F_{1}$. Under \Cref{assumption: hadamard for delta}, we have:
    \begin{equation*}
        |\log \frac{dF(U)}{dF_{0}(U)} -\log \frac{dF(U')}{dF_{0}(U')}| \leq \frac{1}{C_{3}} \left(|dF(U') - dF(U)| + |dF_{0}(U') - dF_{0}(U)|\right) \leq \frac{2L}{C_{3}} \|U' - U\|
    \end{equation*}
    Therefore, $\log \frac{dF(U)}{dF_{0}(U)}$ is Lipschitz continuous in $U$. Moreover, as $\log \frac{dF(U)}{dF_{0}(U)}$ is bounded from above, we have $\log \frac{dF}{dF_{0}} \in L^{1}(F_{2})$ for any $F_{2} \in \Pi$. By the Kantorovich-Rubinstein duality\footnote{See \cite{villani2021topics}. The cost function is $c(U,U') = \|U' - U\|$.}, it holds that:
    \begin{equation*}
        \left|\int \log \frac{dF}{dF_{0}} (dF_{2}-dF_{3})\right| \leq \frac{2L}{C_{3}} W_{1}(F_{2},F_{3})
    \end{equation*}
    which implies the directional derivative is continuous in $F$ as $W_{1}$ is a metric on $\mP(\mU)$, (see \cite{villani2021topics} Remark 7.13(iii)), and thus G\^{a}teaux differentiable. 
    
    Moreover, for any $F_{1}, F_{2},F \in \Pi$, let $\|\cdot\|_{TV}$ be the total variation norm, we have:
    \begin{equation*}
        \left|\int \log \frac{dF_{1}}{dF_{2}} dF\right| \leq \frac{C_{4}}{C_{3}} \int |dF_{1}(U) - dF_{2}(U)|dU = \frac{2C_{4}}{C_{3}} \|F_{1} - F_{2}\|_{TV}
    \end{equation*}
    
    Furthermore, note that:
    \begin{align*}
        & |D_{KL}(F_{1}\|F_{0}) - D_{KL}(F_{2}\|F_{0})| \\
        & \leq \left|\int dF_{1}\log dF_{1} -\int dF_{2} \log dF_{2}\right| + \left|\int (dF_{1} - dF_{2}) \log dF_{0}\right| \\
        & = \left|\int dF_{1}(\log dF_{1}-\log dF_{2}) + \int (dF_{1} - dF_{2}) \log dF_{2}\right| + \left|\int (dF_{1} - dF_{2})\log dF_{0}\right| \\
        & \leq \frac{2C_{4}}{C_{3}} \|F_{1} - F_{2}\|_{TV} + 4| \log C_{4}| \|F_{1} - F_{2}\|_{TV} = \left(\frac{2C_{4}}{C_{3}} + 4| \log C_{4}| \right) \|F_{1} - F_{2}\|_{TV}
    \end{align*}

    Therefore, we have:
    \begin{equation*}
        \|D_{F_{1}}D_{KL}(F_{1}\|F_{0}) - D_{F_{2}}D_{KL}(F_{2}\|F_{0})\|_{op} \leq \left(\frac{4C_{4}}{C_{3}} + 4| \log C_{4}| \right) \|F_{1} - F_{2}\|_{TV}
    \end{equation*}
    which implies its G\^{a}teaux derivative is also continuous in $F$ in the operator norm topology. Therefore, $D_{KL}(F\|F_{0})$ is continuously differentiable on $\{F \in \Pi | D_{KL}(F\|F_{0}) \leq \delta\}$.
\end{proof}

\begin{proof}[\textbf{Proof of \Cref{thm: derivative delta}}]
    Consider the following optimization problem:
    \begin{align*}
        \kappa(\delta, P_{0}) = \inf_{\alpha \in \Theta \times \Pi_{\text{TH}}} s(\alpha) \quad \text{s.t.} \quad P(\alpha) = P_{0}, \quad D_{KL}(F\|F_{0}) \leq \delta
    \end{align*}
    Note that $\Pi_{\text{TH}}$ is convex and closed, $\mA_{I,Lip}^{\delta}$ is compact by \Cref{lemma: compact identified set Lip}, and $s(\alpha)$ is continuous. By the extreme value theorem, the infimum is achieved. Therefore, $\mA_{I,Lip}^{\delta,*}$ is nonempty.

    We aim to show the Hadamard directionally differentiability of $\kappa(\delta, P_{0})$ in the directional $d:=(1,0^{d_{P}})$. We will verify the conditions in \cite{bonnans2013perturbation} Theorem 4.25. In particular, we verify the directional regularity condition in the direction $d:=(1,0^{d_{P}})$ for all $\alpha^{*} \in \mA_{I,Lip}^{\delta,*}$. By \Cref{lemma: derivative delta}, $D_{KL}(F\|F_{0})$ is continuously differentiable on $\Pi_{\text{TH}}$. Therefore, the setting in \cite{bonnans2013perturbation} Page 260 is satisfied.

    For $\alpha^{*} = (\theta^{*},F^{*}) \in \mA_{I,Lip}^{\delta,*}$, by \cite{bonnans2013perturbation} Theorem 4.9, the directional regularity condition in the direction $d$ is equivalent to:
    \begin{equation*}
        0 \in \text{int} \left\{
        \left( 
        \begin{array}{c}
            D_{KL}(F^{*}\|F_{0}) - \delta \\
            P(\alpha^{*}) - P_{0} \\
        \end{array} \right)
        +
        \left( 
        \begin{array}{c}
                D_{F^{*}}D_{KL}(F^{*}\|F_{0})(\Pi - F^{*}) - [0,+\infty) \\
                DP(\alpha^{*})(\Theta \times \Pi-\alpha^{*}) - 0^{d_{P}} \\
        \end{array} \right)
        -
        \left( 
        \begin{array}{c}
            (-\infty,0] \\
            0^{d_{P}} \\
        \end{array} \right) \right\}
    \end{equation*}

    By \Cref{assumption: constraint qualification for delta}, the second part is satisfied. The first part is straightforward as $- [0,+\infty)-(-\infty,0] = \bR$.

    \Cref{assumption: convergence of optimizer for delta} is the same as Assumption (iii) of Theorem 4.25 in \cite{bonnans2013perturbation}. Moreover, \Cref{assumption: convergence of optimizer for delta1} is the same as Assumption (iv) of Theorem 4.25 in \cite{bonnans2013perturbation}. Therefore, \cite{bonnans2013perturbation} Theorem 4.25 shows the right differentiability of $\kappa(\delta,P_{0})$.

    The proof for $\tilde{\kappa}_{\text{TI}}(\delta,P_{0})$ is the same.
\end{proof}

\subsection{Proof in \Cref{sec: Application}}

\subsubsection{Proof of \Cref{lemma: fixed point relationship}}

\begin{proof}[\textbf{Proof of \Cref{lemma: fixed point relationship}}]
    \begin{enumerate}[label=(\roman*)]
        \item See \cite{rust2002there} Theorem 1.
        \item 
        \begin{itemize}
            \item ($\Leftarrow$) Suppose $V^{*}(\omega)$ is the unique fixed point of \eqref{eq: value function}.
        
            \textbf{(Existence)} Note that $s_{0}(\omega) := 1-\exp\left(\omega - V^{*}(\omega)\right)$ is a fixed point of \eqref{eq: structural constraint}.
            
            \textbf{(Uniqueness)} Suppose $s_{0}(\omega)$ is a fixed point of \eqref{eq: structural constraint}. Let $V(\omega) := \omega - \log (1-s_{0}(\omega))$. Then, we show that $V(\omega)$ is a fixed point of \eqref{eq: value function}. It suffices to show that:
            \begin{equation*}
                \beta \bE \left[V(\omega') | \omega\right] = \log \left(\exp(V(\omega)) - \exp(\omega)\right)
            \end{equation*}
            where the right-hand Side equals to $\omega + \log\left(\frac{s_{0}(\omega)}{1-s_{0}(\omega)}\right)$,
            and the left-hand side equals to $\beta \bE \left[\omega' - \log(1-s_{0}(\omega')) | \omega\right]$. Since $s_{0}(\omega)$ is a fixed point of \eqref{eq: structural constraint}, we have RHS = LHS. Therefore, $V(\omega)$ is a fixed point of \eqref{eq: value function}.

            Prove by contradiction. Suppose there are two fixed points $s_{0}(\omega)$ and $\tilde{s}_{0}(\omega)$ to \eqref{eq: structural constraint}. Then, we have: $V(\omega) = \omega - \log(1-s_{0}(\omega))$, $\tilde{V}(\omega) = \omega - \log(1-\tilde{s}_{0}(\omega))$ are both the fixed points to \eqref{eq: value function} as shown above. Since \eqref{eq: value function} has a unique fixed point, we have $V(\omega) = \tilde{V}(\omega)$ for all $\omega$. Therefore, we have $s_{0}(\omega) = \tilde{s}_{0}(\omega)$ for all $\omega$. Contradiction.
            \item ($\Rightarrow$) Suppose \eqref{eq: structural constraint} has a unique fixed point $s^{*}_{0}(\omega)$.

            \textbf{(Existence)} Note that $V(\omega) := \omega - \log(1-s^{*}_{0}(\omega))$ is a fixed point of \eqref{eq: value function}.

            \textbf{(Uniqueness)} Suppose $V(\omega)$ is a fixed point. Then, we show that $s_{0}(\omega) := 1 - \exp\left(\omega - V(\omega)\right)$ is a fixed point of \eqref{eq: structural constraint}. As $V(\omega) = \omega - \log(1-s_{0}(\omega))$ is a fixed point of \eqref{eq: value function}, we have: $\exp\left(\omega - \log(1-s_{0}(\omega))\right) = \exp\left(\omega\right) + \exp\left(\beta \bE\left[\omega' - \log(1-s_{0}(\omega')) | \omega\right]\right)$. It implies $\frac{1}{1 - s_{0}(\omega)} = 1 + \exp\left(\beta \bE\left[\omega' - \log(1-s_{0}(\omega')) | \omega\right] - \omega\right)$. Rearranging it shows that $s_{0}(\omega)$ is a fixed point of \eqref{eq: structural constraint}.

            Prove by contradiction. Suppose there are two fixed points $V_{1}(\omega)$ and $V_{2}(\omega)$ to \eqref{eq: value function}. Then, we have: $s_{0}(\omega) = 1 - \exp\left(\omega - V_{1}(\omega)\right)$, $\tilde{s}_{0}(\omega) = 1 - \exp\left(\omega - V_{2}(\omega)\right)$ are both the fixed points to \eqref{eq: structural constraint} as shown above. Since \eqref{eq: structural constraint} has a unique fixed point, we have $s_{0}(\omega) = \tilde{s}_{0}(\omega)$ for all $\omega$. Therefore, we have $V_{1}(\omega) = V_{2}(\omega)$ for all $\omega$. Contradiction.
        \end{itemize}
        \item This is a direct consequence of the above argument.
    \end{enumerate}
\end{proof}

\end{document}